\documentclass[10pt, conference, letterpaper]{IEEEtran}
\usepackage{amsmath,amssymb,amsthm}
\usepackage{xcolor}
\usepackage[linesnumbered,algoruled,vlined]{algorithm2e}
\usepackage{xparse} 
\usepackage{array,multirow}
\usepackage{mathabx} 
\usepackage{graphicx}
\usepackage{comment}
\usepackage{url}
\usepackage{relsize}

 \newcolumntype{F}{>{$\displaystyle\,}r<{$}@{\hspace{0.0em}}}
 \newcolumntype{C}{>{$\displaystyle\,}c<{$}@{\hspace{0.0em}}}
 \newcolumntype{B}{>{$\displaystyle\,}r<{$}@{\hspace{0.0em}}}
 \newcolumntype{R}{>{$\displaystyle}r<{$}@{\hspace{0.2em}}}
 \newcolumntype{S}{>{$\displaystyle}r<{$}@{\hspace{0.2em}}}
 \newcolumntype{L}{>{$\displaystyle}l<{$}@{\hspace{0.2em}}}
 \newcolumntype{Q}{>{$\displaystyle}l<{$}@{\hspace{0.3em}}}

 \newcounter{IPnumber}
 \newcommand{\tagIt}[1]{\refstepcounter{equation}\textnormal{({\theequation})} \label{#1}}

\newtheorem{theorem}{Theorem}
\newtheorem{lemma}[theorem]{Lemma}
\newtheorem{fact}[theorem]{Fact}
\newtheorem{definition}[theorem]{Definition}

\newtheorem{corollary}[theorem]{Corollary}

\begin{document}

\title{{\smaller[0]{Service Chain and Virtual Network Embeddings:}} \\ {\smaller[0]{Approximations using Randomized Rounding}} }

\author{
\IEEEauthorblockN{Matthias Rost}
\IEEEauthorblockA{TU Berlin, Germany \\
\texttt{mrost@inet-tu-berlin.de}}
\and
\IEEEauthorblockN{Stefan Schmid}
\IEEEauthorblockA{Aalborg University, Denmark \\
\texttt{schmiste@cs.aau.dk}}
}


\date{}

\newcommand{\TODO}[1]{\textcolor{red}{TODO: #1}}

\newcommand\numberthis{\addtocounter{equation}{1}\tag{\theequation}}

\newcommand{\NULL}{\textnormal{\texttt{NULL}}}
\newcommand{\scale}{\ensuremath{\lambda}}


\newcommand{\preals}{\ensuremath{\mathbb{R}_{\geq 0}}}


\newcommand{\requests}{\ensuremath{\mathcal{R}}}
\newcommand{\requestsP}{\ensuremath{\mathcal{R}'}}
\newcommand{\req}{r}
\newcommand{\types}{\ensuremath{\mathcal{T}}}
\newcommand{\type}{\ensuremath{\tau}}
\newcommand{\SVTypes}[1][\type]{\ensuremath{V^{#1}_{S}}}
\newcommand{\SVTypesCycle}[1][C_k]{\ensuremath{V^{#1}_{S,t}}}
\newcommand{\VG}[1][\req]{\ensuremath{G_{#1}}}
\newcommand{\VV}[1][\req]{\ensuremath{V_{#1}}}
\newcommand{\VE}[1][\req]{\ensuremath{E_{#1}}}
\newcommand{\VGbar}[1][\req]{\ensuremath{\bar{G}_{#1}}}
\newcommand{\VVbar}[1][\req]{\ensuremath{\bar{V}_{#1}}}
\newcommand{\VEbar}[1][\req]{\ensuremath{\bar{E}_{#1}}}
\newcommand{\Vstart}[1][\req]{\ensuremath{s_{#1}}}
\newcommand{\Vend}[1][\req]{\ensuremath{t_{#1}}}

\newcommand{\VGext}[1][\req]{\ensuremath{G^{\textnormal{ext}}_{#1}}}
\newcommand{\VGextFlow}[1][\req]{\ensuremath{G^{\textnormal{ext}}_{#1,f}}}
\NewDocumentCommand{\VGP}{O{\req} O{i} O{j}}{\ensuremath{G^{#2,#3}_{#1}}}
\newcommand{\VVext}[1][\req]{\ensuremath{V^{\textnormal{ext}}_{#1}}}
\newcommand{\VEext}[1][\req]{\ensuremath{E^{\textnormal{ext}}_{#1}}}
\newcommand{\VEextHorizontal}[1][\req]{\ensuremath{E^{\textnormal{ext}}_{#1,u,v}}}
\newcommand{\VEextVertical}[1][\req]{\ensuremath{E^{\textnormal{ext}}_{#1,\type,u}}}

\newcommand{\Vsource}[1][\req]{\ensuremath{o^+_{\req}}}
\newcommand{\Vsink}[1][\req]{\ensuremath{o^-_{\req}}}

\newcommand{\VMultiplicity}[1][\req]{\ensuremath{M_{#1}}}
\newcommand{\Vprofit}[1][\req]{\ensuremath{b_{#1}}}
\newcommand{\VprofitMax}{\ensuremath{b_{\max}}}

\newcommand{\Vcap}[1][\req]{\ensuremath{d_{#1}}}
\newcommand{\Vloc}[1][\req]{\ensuremath{l_{#1}}}
\newcommand{\Vtype}[1][\req]{\ensuremath{\tau_{#1}}}


\newcommand{\SG}{\ensuremath{G_S}}
\newcommand{\SR}{\ensuremath{R_{S}}}
\newcommand{\SRV}{\ensuremath{R^V_{S}}}
\newcommand{\SV}{\ensuremath{V_S}}
\newcommand{\SE}{\ensuremath{E_S}}

\newcommand{\Scap}{\ensuremath{d_{S}}}
\newcommand{\ScapType}[1][\type]{\ensuremath{d^{#1}_{\SV}}}
\newcommand{\ScapTypePrime}[1][\type]{\ensuremath{{{d'}^{#1}_{\SV}}}}

\newcommand{\Scost}{\ensuremath{c_{S}}}
\newcommand{\ScostType}[1][\type]{\ensuremath{c^{#1}_{\SV}}}


\newcommand{\map}[1][\req]{\ensuremath{m_{#1}}}
\newcommand{\mapV}[1][\req]{\ensuremath{m^V_{#1}}}
\newcommand{\mapE}[1][\req]{\ensuremath{m^E_{#1}}}

\newcommand{\FeasibleLP}{\ensuremath{\mathcal{F}_{\textnormal{LP}}}}
\newcommand{\FeasibleIP}{\ensuremath{\mathcal{F}_{\textnormal{IP}}}}


\makeatletter
\newcommand{\removelatexerror}{\let\@latex@error\@gobble}
\makeatother

\newcounter{ipCounter}
\NewDocumentEnvironment{IPFormulation}{m}{%
\refstepcounter{ipCounter}
\begin{algorithm}[#1]%
\renewcommand\thealgocf{\arabic{ipCounter}}
}{%
\end{algorithm}
\addtocounter{algocf}{-1}
}

\NewDocumentEnvironment{IPFormulationStar}{m}{%
\refstepcounter{ipCounter}
\begin{algorithm*}[#1]%
\renewcommand\thealgocf{\arabic{ipCounter}}
}{%
\end{algorithm*}
\addtocounter{algocf}{-1}
}



\newcommand{\spaceSolReq}[1][\req]{\ensuremath{\mathcal{M}_{#1}}}
\newcommand{\spaceLP}{\ensuremath{\mathcal{F}^{\textnormal{mcf}}_{\textnormal{LP}}}}
\newcommand{\spaceIP}{\ensuremath{\mathcal{F}^{\textnormal{mcf}}_{\textnormal{IP}}}}
\newcommand{\spaceLPNew}{\ensuremath{\mathcal{F}^{\textnormal{new}}_{\textnormal{LP}}}}
\newcommand{\spaceLPD}{\ensuremath{\mathcal{F}^{\mathcal{D}}_{\textnormal{LP}}}}
\newcommand{\spaceLPDreq}[1][\req]{\ensuremath{\mathcal{F}^{\mathcal{D}}_{\textnormal{LP},#1}}}

\DeclareDocumentCommand{\OptProfit}{}{\ensuremath{\hat{P}}}


\DeclareDocumentCommand{\NodeG}{O{\req} O{\pi}}{\ensuremath{G^N_{#1,#2}}}
\DeclareDocumentCommand{\NodeV}{O{\req} O{\pi}}{\ensuremath{V^N_{#1,#2}}}
\DeclareDocumentCommand{\NodeE}{O{\req} O{\pi}}{\ensuremath{E^N_{#1,#2}}}

\DeclareDocumentCommand{\EdgeG}{O{\req} O{i} O{j} O{u}}{\ensuremath{G^E_{#1,#2,#3,#4}}}
\DeclareDocumentCommand{\EdgeV}{O{\req} O{i} O{j} O{u}}{\ensuremath{V^E_{#1,#2,#3,#4}}}
\DeclareDocumentCommand{\EdgeE}{O{\req} O{i} O{j} O{u}}{\ensuremath{E^E_{#1,#2,#3,#4}}}

\DeclareDocumentCommand{\VESD}{O{\req}}{\ensuremath{\overrightarrow{E}_{#1}}}
\DeclareDocumentCommand{\VEOD}{O{\req}}{\ensuremath{\overleftarrow{E}_{#1}}}
\DeclareDocumentCommand{\VESigmaD}{O{\req} O{\sigma}}{\ensuremath{E_{#1,#2}}}

\DeclareDocumentCommand{\NodeVRange}{O{\req} O{i} O{j}}{\ensuremath{V^N_{#1,#2,#3}}}
\DeclareDocumentCommand{\NodeVRangeRange}{O{\req} O{i} O{j} O{u} O{v}}{\ensuremath{V^N_{#1,#2,#3,#4,#5}}}

\DeclareDocumentCommand{\path}{O{\req} O{k }}{\ensuremath{P_{#1,#2}}}

\DeclareDocumentCommand{\NodePaths}{O{\req}}{\ensuremath{\mathcal{P}_{#1}}}


\DeclareDocumentCommand{\loadV}{O{\req} O{u}}{\ensuremath{l_{#1,#2}}}
\DeclareDocumentCommand{\loadE}{O{\req} O{u} O{v}}{\ensuremath{l_{#1,#2,#3}}}
\DeclareDocumentCommand{\loadX}{O{\req} O{x}}{\ensuremath{l_{#1,#2}}}
\DeclareDocumentCommand{\decomp}{O{\req} O{k}}{\ensuremath{D_{#1}^{#2}}}
\DeclareDocumentCommand{\decompHat}{O{\req} O{k}}{\ensuremath{{\hat{D}}_{#1}^{#2}}}
\DeclareDocumentCommand{\load}{O{\req} O{k}}{\ensuremath{l_{#1}^{#2}}}
\DeclareDocumentCommand{\prob}{O{\req} O{k}}{\ensuremath{f_{#1}^{#2}}}
\DeclareDocumentCommand{\mapping}{O{\req} O{k}}{\ensuremath{m_{#1}^{#2}}}

\DeclareDocumentCommand{\loadHat}{O{\req} O{k}}{\ensuremath{\hat{l}_{#1}^{#2}}}
\DeclareDocumentCommand{\probHat}{O{\req} O{k}}{\ensuremath{\hat{f}_{#1}^{#2}}}
\DeclareDocumentCommand{\mappingHat}{O{\req} O{k}}{\ensuremath{\hat{m}_{#1}^{#2}}}

\DeclareDocumentCommand{\loadHat}{O{\req} O{k}}{\ensuremath{\hat{l}_{#1}^{#2}}}
\DeclareDocumentCommand{\probHat}{O{\req} O{k}}{\ensuremath{\hat{f}_{#1}^{#2}}}
\DeclareDocumentCommand{\mappingHat}{O{\req} O{k}}{\ensuremath{\hat{m}_{#1}^{#2}}}

\DeclareDocumentCommand{\Exp}{}{\ensuremath{\mathbb{E}}}
\DeclareDocumentCommand{\randVarX}{O{\req} O{k}}{\ensuremath{X_{#1}^{#2}}}
\DeclareDocumentCommand{\randVarY}{O{\req}}{\ensuremath{Y_{#1}}}
\DeclareDocumentCommand{\randVarZ}{O{\req}}{\ensuremath{Z_{#1}}}
\DeclareDocumentCommand{\randVarL}{O{x}}{\ensuremath{L_{x}}}
\DeclareDocumentCommand{\randVarLX}{O{\req} O{x} O{y}}{\ensuremath{L_{#1,#2,#3}}}
\DeclareDocumentCommand{\randVarLNode}{O{\req} O{\type} O{u}}{\ensuremath{L_{#1,#2,#3}}}
\DeclareDocumentCommand{\randVarLEdge}{O{\req} O{u} O{v}}{\ensuremath{L_{#1,#2,#3}}}
\DeclareDocumentCommand{\randVarM}{O{\req}}{\ensuremath{M_{#1}}}
\DeclareDocumentCommand{\randVarC}{O{\req}}{\ensuremath{C_{#1}}}

\DeclareDocumentCommand{\ProbVarX}{O{1}}{\ensuremath{\mathbb{P}(\randVarX = #1)}}
\DeclareDocumentCommand{\ProbVarY}{O{1}}{\ensuremath{\mathbb{P}(\randVarY = #1)}}
\DeclareDocumentCommand{\ProbVarZ}{O{1}}{\ensuremath{\mathbb{P}(\randVarZ = #1)}}
\DeclareDocumentCommand{\ProbVarL}{O{1}}{\ensuremath{\mathbb{P}(\randVarL = #1)}}
\DeclareDocumentCommand{\ProbVarM}{O{1}}{\ensuremath{\mathbb{P}(\randVarM = #1)}}
\DeclareDocumentCommand{\ProbVarC}{O{1}}{\ensuremath{\mathbb{P}(\randVarC = #1)}}

\DeclareDocumentCommand{\randVarObjApprox}{}{\ensuremath{Obj_{\textnormal{Alg}}}}
\DeclareDocumentCommand{\optLP}{}{\ensuremath{\textnormal{Opt}_{\textnormal{LP}}}}
\DeclareDocumentCommand{\optLPstd}{}{\ensuremath{\textnormal{Opt}^{\textnormal{mcf}}_{\textnormal{LP}}}}
\DeclareDocumentCommand{\optLPnew}{}{\ensuremath{\textnormal{Opt}^{\textnormal{new}}_{\textnormal{LP}}}}
\DeclareDocumentCommand{\optIP}{}{\ensuremath{\textnormal{Opt}_{\textnormal{IP}}}}

\DeclareDocumentCommand{\WAC}{O{\req}}{\ensuremath{\textnormal{WC}_{\req}}}


\DeclareDocumentCommand{\PotEmbeddings}{O{\req}}{\ensuremath{\mathcal{D}_{#1}}}
\DeclareDocumentCommand{\PotEmbeddingsHat}{O{\req}}{\ensuremath{\hat{\mathcal{D}}_{#1}}}


\DeclareDocumentCommand{\maxLoadX}{O{x}}{\ensuremath{\textnormal{max}^{L,\sum}_{#1}}}
\DeclareDocumentCommand{\maxLoadV}{O{\req} O{\type} O{u}}{\ensuremath{\textnormal{max}^{L}_{#1,#2,#3}}}
\DeclareDocumentCommand{\maxLoadE}{O{\req} O{u} O{v}}{\ensuremath{\textnormal{max}^{L}_{#1,#2,#3}}}
\DeclareDocumentCommand{\maxLoadVSum}{O{\req} O{\type} O{u}}{\ensuremath{\textnormal{max}^{L,\sum}_{#1,#2,#3}}}
\DeclareDocumentCommand{\maxLoadESum}{O{u} O{v}}{\ensuremath{\textnormal{max}^{L,\sum}_{#1,#2}}}

\DeclareDocumentCommand{\DeltaV}{}{\ensuremath{\Delta_V}}
\DeclareDocumentCommand{\DeltaE}{}{\ensuremath{\Delta_E}}


\DeclareDocumentCommand{\VVroot}{O{\req}}{\ensuremath{r_{#1}}}
\DeclareDocumentCommand{\VVpred}{O{\req}}{\ensuremath{\pi_{#1}}}

\newcommand{\VGbfs}[1][\req]{\ensuremath{G^{\textnormal{bfs}}_{#1}}}
\newcommand{\VVbfs}[1][\req]{\ensuremath{V^{\textnormal{bfs}}_{#1}}}
\newcommand{\VEbfs}[1][\req]{\ensuremath{E^{\textnormal{bfs}}_{#1}}}

\DeclareDocumentCommand{\Cycles}{O{\req}}{\ensuremath{\mathcal{C}_{#1}}}
\DeclareDocumentCommand{\Paths}{O{\req}}{\ensuremath{\mathcal{P}_{#1}}}

\DeclareDocumentCommand{\VVcycleSource}{O{\req} O{k}}{\ensuremath{s^{C_{#2}}_{#1}}}
\DeclareDocumentCommand{\VVcycleTarget}{O{\req} O{k}}{\ensuremath{t^{C_{#2}}_{#1}}}
\DeclareDocumentCommand{\VVpathSource}{O{\req} O{k}}{\ensuremath{s^{P_{#2}}_{#1}}}
\DeclareDocumentCommand{\VVpathTarget}{O{\req} O{k}}{\ensuremath{t^{P_{#2}}_{#1}}}

\DeclareDocumentCommand{\VGcycle}{O{\req} O{k}}{\ensuremath{G^{C_{#2}}_{#1}}}
\DeclareDocumentCommand{\VVcycle}{O{\req} O{k}}{\ensuremath{V^{C_{#2}}_{#1}}}
\DeclareDocumentCommand{\VEcycle}{O{\req} O{k}}{\ensuremath{E^{C_{#2}}_{#1}}}
\DeclareDocumentCommand{\VEcycleSame}{O{\req} O{k}}{\ensuremath{\overrightarrow{E}^{C_{#2}}_{#1}}}
\DeclareDocumentCommand{\VEcycleDiff}{O{\req} O{k}}{\ensuremath{\overleftarrow{E}^{C_{#2}}_{#1}}}

\DeclareDocumentCommand{\VGpath}{O{\req} O{k}}{\ensuremath{G^{P_{#2}}_{#1}}}
\DeclareDocumentCommand{\VVpath}{O{\req} O{k}}{\ensuremath{V^{P_{#2}}_{#1}}}
\DeclareDocumentCommand{\VEpath}{O{\req} O{k}}{\ensuremath{E^{P_{#2}}_{#1}}}
\DeclareDocumentCommand{\VEpathSame}{O{\req} O{k}}{\ensuremath{\overrightarrow{E}^{P_{#2}}_{#1}}}
\DeclareDocumentCommand{\VEpathDiff}{O{\req} O{k}}{\ensuremath{\overleftarrow{E}^{P_{#2}}_{#1}}}

\DeclareDocumentCommand{\VEDiff}{O{\req}}{\ensuremath{\overleftarrow{E}^{\textnormal{bfs}}_{#1}}}

\DeclareDocumentCommand{\VEcycles}{O{\req}}{\ensuremath{E^{\mathcal{C}}_{#1}}}
\DeclareDocumentCommand{\VEpaths}{O{\req}}{\ensuremath{E^{\mathcal{P}}_{#1}}}

\DeclareDocumentCommand{\VVcycleSourcesTargets}{O{\req}}{\ensuremath{V^{\mathcal{C},\pm}_{#1}}}
\DeclareDocumentCommand{\VVpathSourcesTargets}{O{\req}}{\ensuremath{V^{\mathcal{P},\pm}_{#1}}}
\DeclareDocumentCommand{\VVSourcesTargets}{O{\req}}{\ensuremath{V^{\pm}_{#1}}}

\DeclareDocumentCommand{\VVcycleSources}{O{\req}}{\ensuremath{V^{\mathcal{C},+}_{#1}}}
\DeclareDocumentCommand{\VVpathSources}{O{\req}}{\ensuremath{V^{\mathcal{P},+}_{#1}}}

\DeclareDocumentCommand{\VVcycleTargets}{O{\req}}{\ensuremath{V^{\mathcal{C},-}_{#1}}}
\DeclareDocumentCommand{\VVpathTargets}{O{\req}}{\ensuremath{V^{\mathcal{P},-}_{#1}}}

\DeclareDocumentCommand{\VEcycleBranchR}{O{\req} O{k}}{\ensuremath{B^{C_{#2}}_{#1,1}}}
\DeclareDocumentCommand{\VEcycleBranchL}{O{\req} O{k}}{\ensuremath{B^{C_{#2}}_{#1,2}}}

\DeclareDocumentCommand{\VGdecomp}{O{\req} O{k}}{\ensuremath{G^{\mathcal{D}}_{#1}}}
\DeclareDocumentCommand{\VVdecomp}{O{\req} O{k}}{\ensuremath{V^{\mathcal{D}}_{#1}}}
\DeclareDocumentCommand{\VEdecomp}{O{\req} O{k}}{\ensuremath{E^{\mathcal{D}}_{#1}}}

\DeclareDocumentCommand{\VVbranching}{O{\req} }{\ensuremath{\mathcal{B}_{#1}}}
\DeclareDocumentCommand{\VVbranchingcycle}{O{\req} O{k}}{\ensuremath{\mathcal{B}^{C_{#2}}_{#1}}}
\DeclareDocumentCommand{\VVbranchingpath}{O{\req} O{k}}{\ensuremath{\mathcal{B}^{P_{k}}_{#1}}}
\DeclareDocumentCommand{\VVjoin}{O{\req} }{\ensuremath{\mathcal{J}_{#1}}}
\DeclareDocumentCommand{\VVaggregation}{O{\req} }{\ensuremath{\mathcal{A}_{#1}}}

\DeclareDocumentCommand{\VGextcycle}{O{\req} O{k}}{\ensuremath{G^{C_{#2}}_{#1,\textnormal{ext}}}}

\DeclareDocumentCommand{\VVextcycle}{O{\req} O{k}}{\ensuremath{V^{C_{#2}}_{#1,\textnormal{ext}}}}
\DeclareDocumentCommand{\VVextcycleSources}{O{\req} O{k}}{\ensuremath{V^{C_{#2}}_{#1,+}}}
\DeclareDocumentCommand{\VVextcycleTargets}{O{\req} O{k}}{\ensuremath{V^{C_{#2}}_{#1,-}}}
\DeclareDocumentCommand{\VVextcycleSubstrate}{O{\req} O{k}}{\ensuremath{V^{C_{#2}}_{#1,S}}}

\DeclareDocumentCommand{\VEextcycle}{O{\req} O{k}}{\ensuremath{E^{C_{#2}}_{#1,\textnormal{ext}}}}
\DeclareDocumentCommand{\VEextcycleSources}{O{\req} O{k}}{\ensuremath{E^{C_{#2}}_{#1,+}}}
\DeclareDocumentCommand{\VEextcycleTargets}{O{\req} O{k}}{\ensuremath{E^{C_{#2}}_{#1,-}}}
\DeclareDocumentCommand{\VEextcycleSubstrate}{O{\req} O{k}}{\ensuremath{E^{C_{#2}}_{#1,S}}}
\DeclareDocumentCommand{\VEextcycleF}{O{\req} O{k}}{\ensuremath{E^{C_{#2}}_{#1,F}}}

\DeclareDocumentCommand{\VGextpath}{O{\req} O{k}}{\ensuremath{G^{P_{#2}}_{#1,\textnormal{ext}}}}

\DeclareDocumentCommand{\VVextpath}{O{\req} O{k}}{\ensuremath{V^{P_{#2}}_{#1,\textnormal{ext}}}}
\DeclareDocumentCommand{\VVextpathSources}{O{\req} O{k}}{\ensuremath{V^{P_{#2}}_{#1,+}}}
\DeclareDocumentCommand{\VVextpathTargets}{O{\req} O{k}}{\ensuremath{V^{P_{#2}}_{#1,-}}}
\DeclareDocumentCommand{\VVextpathSubstrate}{O{\req} O{k}}{\ensuremath{V^{P_{#2}}_{#1,S}}}

\DeclareDocumentCommand{\VEextpath}{O{\req} O{k}}{\ensuremath{E^{P_{#2}}_{#1,\textnormal{ext}}}}
\DeclareDocumentCommand{\VEextpathSources}{O{\req} O{k}}{\ensuremath{E^{P_{#2}}_{#1,+}}}
\DeclareDocumentCommand{\VEextpathTargets}{O{\req} O{k}}{\ensuremath{E^{P_{#2}}_{#1,-}}}
\DeclareDocumentCommand{\VEextpathSubstrate}{O{\req} O{k}}{\ensuremath{E^{P_{#2}}_{#1,S}}}
\DeclareDocumentCommand{\VEextpathF}{O{\req} O{k}}{\ensuremath{E^{P_{#2}}_{#1,F}}}


\DeclareDocumentCommand{\varFlowInput}{O{\req} O{i} O{u}}{\ensuremath{f^+_{#1,#2,#3}}}
\DeclareDocumentCommand{\varFlowOutput}{O{\req} O{i} O{u}}{\ensuremath{f^+_{#1,#2,#3}}}

\DeclareDocumentCommand{\VEextcycleHorizontal}{O{\req} O{k} O{u} O{v}}{\ensuremath{E^{C_{#2}}_{#1,\textnormal{ext},#3,#4}}}
\DeclareDocumentCommand{\VEextpathHorizontal}{O{\req} O{k} O{u} O{v}}{\ensuremath{E^{P_{#2}}_{#1,\textnormal{ext},#3,#4}}}
\DeclareDocumentCommand{\VEextcycleVertical}{O{\req} O{k} O{\type} O{u}}{\ensuremath{E^{C_{#2}}_{#1,\textnormal{ext},#3,#4}}}
\DeclareDocumentCommand{\VEextpathVertical}{O{\req} O{k} O{\type} O{u}}{\ensuremath{E^{P_{#2}}_{#1,\textnormal{ext},#3,#4}}}

\DeclareDocumentCommand{\VEextCGHorizontal}{O{\req} O{u} O{v}}{\ensuremath{E^{\textnormal{ext,SCG}}_{#1,#2,#3}}}
\DeclareDocumentCommand{\VEextCGVertical}{O{\req} O{\type} O{u}}{\ensuremath{E^{\textnormal{ext,SCG}}_{#1,#2,#3}}}

\DeclareDocumentCommand{\VVextCGFlowNodes}{O{\req}}{\ensuremath{V^{\textnormal{ext,SCG}}_{#1,\textnormal{flow}}}}
\DeclareDocumentCommand{\VEextCGFlowEdges}{O{\req}}{\ensuremath{E^{\textnormal{ext,SCG}}_{#1,\textnormal{flow}}}}


\DeclareDocumentCommand{\Queue}{}{\ensuremath{\mathcal{Q}}}
\DeclareDocumentCommand{\QueueC}{}{\ensuremath{\mathcal{Q}_{\mathcal{C}}}}
\DeclareDocumentCommand{\QueueP}{}{\ensuremath{\mathcal{Q}_{\mathcal{P}}}}
\DeclareDocumentCommand{\UsedPaths}{}{\ensuremath{\mathcal{P}}}
\DeclareDocumentCommand{\Variables}{}{\ensuremath{\mathcal{V}}}

\DeclareDocumentCommand{\VGextcycleFlow}{O{\req} O{k}}{\ensuremath{G^{C_{#2}}_{#1,\textnormal{ext},f}}}

\DeclareDocumentCommand{\VGextcycleFlowBranchR}{O{\req} O{k}}{\ensuremath{G^{C_{#2},{B}_1}_{#1,\textnormal{ext},f}}}
\DeclareDocumentCommand{\VVextcycleFlowBranchR}{O{\req} O{k}}{\ensuremath{V^{C_{#2},{B}_1}_{#1,\textnormal{ext},f}}}
\DeclareDocumentCommand{\VEextcycleFlowBranchR}{O{\req} O{k}}{\ensuremath{E^{C_{#2},{B}_1}_{#1,\textnormal{ext},f}}}

\DeclareDocumentCommand{\VGextcycleFlowBranchL}{O{\req} O{k}}{\ensuremath{G^{C_{#2},{B}_2}_{#1,\textnormal{ext},f}}}
\DeclareDocumentCommand{\VVextcycleFlowBranchL}{O{\req} O{k}}{\ensuremath{V^{C_{#2},{B}_2}_{#1,\textnormal{ext},f}}}
\DeclareDocumentCommand{\VEextcycleFlowBranchL}{O{\req} O{k}}{\ensuremath{E^{C_{#2},{B}_2}_{#1,\textnormal{ext},f}}}

\DeclareDocumentCommand{\VGextpathFlow}{O{\req} O{k}}{\ensuremath{G^{P_{#2}}_{#1,\textnormal{ext},f}}}
\DeclareDocumentCommand{\VVextpathFlow}{O{\req} O{k}}{\ensuremath{V^{P_{#2}}_{#1,\textnormal{ext},f}}}
\DeclareDocumentCommand{\VEextpathFlow}{O{\req} O{k}}{\ensuremath{E^{P_{#2}}_{#1,\textnormal{ext},f}}}

\DeclareDocumentCommand{\VVKSource}{O{\req} O{K}}{\ensuremath{s^{K}_{#1}}}
\DeclareDocumentCommand{\VVKTarget}{O{\req} O{K}}{\ensuremath{t^{K}_{#1}}}
\DeclareDocumentCommand{\VVKSourcesTargets}{O{\req}}{\ensuremath{V^{K,\pm}_{#1}}}


\maketitle

\begin{abstract}
The SDN and NFV paradigms enable novel 
network services which can be realized and embedded in a flexible
and rapid manner.
For example, SDN can be used to flexibly steer traffic
from a source to a destination through a sequence of virtualized middleboxes, 
in order to realize so-called service chains. 
The service chain embedding problem consists of
three tasks: admission control, finding suitable locations to allocate the virtualized
middleboxes and computing corresponding routing paths.
This paper considers the offline batch embedding of multiple service chains. Concretely, we consider the objectives of maximizing the profit by embedding an optimal subset of requests or minimizing the costs when all requests need to be embedded. 
Interestingly, while the service chain embedding problem
has recently received much attention, so far, only
non-polynomial time algorithms~(based on integer programming)
as well as heuristics~(which do not provide any formal guarantees)
are known. 
This paper presents the first polynomial time service chain
 approximation algorithms both for the case with admission and without admission control.
Our algorithm is based on a novel extension of the classic linear programming and randomized rounding 
technique, which may be of independent interest. 
In particular, we show that our approach can also be extended to more complex service graphs, containing cycles or sub-chains, hence also
  providing new insights into the classic virtual
  network embedding problem.
\end{abstract}

\section{Introduction}

Computer networks are currently undergoing a phase transition, and
especially the Software-Defined Networking~(SDN)~and
Netwok Function Virtualization~(NFV)~paradigms 
have the potential to overcome the ossification of computer
networks and to introduce interesting new flexiblities
and novel service abstractions such as service chaining.

In a nutshell, in a Software-Defined Network~(SDN), the control over the forwarding
switches in the data plane
is outsourced and consolidated to a logically centralized software in the so-called
control plane. This separation enables faster innovations, as the control plane
can evolve independently from the data plane: software often trumps hardware in terms 
of supported innovation speed.
Moreover, the logically centralized perspective introduced by SDN is natural and attractive, as
many networking tasks~(e.g., routing, spanning tree constructions)~are inherently non-local.
Indeed, a more flexible traffic engineering is considered one of the key benefits
of SDN~\cite{b4,sdx}.
Such routes are not necessarily shortest paths or destination-based,
or
not even loop-free~\cite{flowtags}.
 In particular,  
OpenFlow~\cite{OpenFlow}, the standard SDN protocol today, allows to define routing paths 
which depend on Layer-2, Layer-3 and even Layer-4 header fields. 

Network Function Virtualization~(NFV)~introduces flexibilities in terms
of function and service deployments. Today's computer networks 
rely on a large number of middleboxes~(e.g., NATs, firewalls, WAN optimizers),
typically realized using expensive hardware appliances which are cumbersome to manage.
For example, it is known that the number of middleboxes in enterprise 
networks can be of the same order of magnitude as the number
of routers ~\cite{someone}.
The virtualization of these functions renders the network management
more flexible, and allows to define and quickly deploy novel in-network services~\cite{routebricks,opennf,sosr15,modeling-middleboxes,clickos}.
Virtualized network functions can easily be instantiated 
on the most suitable
network nodes, e.g., running in a virtual machine
on a commodity x86 server.
The transition to NFV is discussed
within standardization groups such as ETSI, and we currently also
witness first
deployments, e.g., TeraStream~\cite{terastream}.

Service chaining~\cite{ewsdn14,stefano-sigc,merlin} is a particularly interesting new service model,
that combines the flexibilities from SDN and NFV.
In a nutshell, a service chain describes a sequence of
network functions which need to be traversed on the way
from a given source~$s$ to a given destination~$t$. 
For example, a service chain could define that
 traffic originating at the source is first steered through
an intrusion detection system
for security,
next through a traffic optimizer,
and only then is routed towards the destination.
While NFV can be used to flexibly allocate network functions,
SDN can be used to steer traffic through them.

\subsection{The Scope and Problem} 

This paper studies the problem of how to algorithmically exploit the flexibilities 
introduced by the SDN+NFV paradigm. 
We attend the service chain embedding problem, which has recently
received much attention. The problem generally consists of three tasks:
(1)~(if possible)~admission control, i.e. selecting and serving
only the most valuable requests,
(2)~the allocation of the virtualized
middleboxes at the optimal locations 
and~(3)~the computation of routing paths via them. 
Assuming that one is allowed to exert admission control, the objective is to maximize the profit, i.e., the prizes collected for embedding service chains. We also study the problem variant, in which a given set of requests must be embedded, i.e. when admission control cannot be exerted. In this variant we consider the natural objective of minimizing the cumulative allocation costs.
The service chain embedding algorithms presented
so far in the literature either have a
non-polynomial runtime~(e.g., are based on integer programming~\cite{mehraghdam2014specifying,stefano-sigc,merlin}),
do not provide any approximation guarantees~\cite{karl-chains},
or ignore important aspects of the problem (such as 
link capacity constraints~\cite{sirocco15}). 

More generally, we also attend to the current trend towards more complex service chains, 
connecting network functions not only in a linear order 
but as arbitrary graphs, i.e. as a kind of \emph{virtual network}. 

\subsection{Our Contributions}
 
This paper makes the following contributions.
We present the first polynomial time algorithms for the (offline)~service chain
embedding problem with and without admission control, which provide provable approximation guarantees.

We also initate the study of approximation algorithms 
for more general service graphs (or ``virtual networks''). In particular, 
we present polynomial time
approximation algorithms for the embedding of service 
\emph{cactus graphs}, which may contain branch sub-chains and even cycles. To this end, we develop a novel Integer Program formulation together with a novel decomposition algorithm, 
enabling the randomized rounding: we prove that known Integer Programming formulations are not applicable. 

\subsection{Technical Novelty}

Our algorithms are based on the well-established  
randomized rounding approach~\cite{Raghavan-Thompson}:
the algorithms use an exact Integer Program, 
for which however we only compute \emph{relaxed, i.e. linear, solutions}, 
in polynomial time.
Given the resulting fractional solution, an approximate
integer solution is derived using randomized rounding, 
in the usual resource augmentation model.

However, while randomized rounding has been studied
intensively and applied successfully in the context of path embeddings~\cite{Raghavan-Thompson},
to the best of our knowledge, besides our own work,
the question of how to extend this approach to service chains~(where
paths need to traverse certain flexible waypoints)~or even more
complex graphs~(such as virtual networks),  has not been explored
yet. 
Moreover, we are not aware of any extensions of the randomized rounding
approach to problems allowing for admission control. 

Indeed, the randomized rounding of more complex graph requests
and the admission control pose some interesting new challenges. 
In particular, the more general setting requires both a novel Integer Programming formulation as well as a novel decomposition approach. 
Indeed, we show that solutions obtained using the standard formulation \cite{vnep,rostSchmidFeldmann2014} may not be decomposable at all,
as the relaxed embedding solutions are \emph{not} 
a linear combination of elementary solutions. Besides the fact that the randomized rounding approach can therefore not be applied, we prove that the relaxation of our novel formulation is indeed provably stronger than the well-known formulation.

\subsection{Organization}

The remainder of this paper is organized as follows.
Section~\ref{sec:model} formally introduces our model.
Section~\ref{sec:decompo} presents the 
Integer Programs and our decomposition method.
Section~\ref{sec:randround} presents our randomized approximation algorithm for the service chain embedding problem with admission control and Section~\ref{sec:approximation-without-admission-control} extends the approximation for the case without admission control.
In Section~\ref{sec:cactus} we derive a novel Integer Program and decomposition approach for approximating service graphs, and show why classic formulations are not sufficient.
Section~\ref{sec:relwork} reviews related work and Section~\ref{sec:conclusion} concludes our work.

\section{Offline Service Chain Embedding Problem}\label{sec:model}

This paper studies the Service Chain Embedding Problem, short SCEP.
Intuitively, a service chain consists of a set of Network Functions~(NFs), 
such as a firewall or a NAT, and routes between these functions. 
We consider the offline setting, where batches of service chains
have to be embedded simultaneously. 
Concretely, we study two problem variants: (1)~SCEP-P where the task is to embed a subset of service chains to maximize the profit and (2)~SCEP-C where \emph{all} given service chains need to be embedded and the objective is to minimize the resource costs. Hence, service chain requests might be attributed with prizes and resources as e.g., link bandwidth 
or processing~(e.g., of a firewall network function)~may come at a certain cost.


\subsection{Definitions \& Formal Model}
\label{sec:problem-definitions}

Given is a substrate network~(the physical network representing
the physical resources)~which is modeled as directed network 
$\SG =(\SV, \SE)$. We assume that the substrate network 
offers a finite set~$\types$ of different network 
functions~(NFs)~at nodes. The set of network function types may contain e.g., `FW'~(firewall), `DPI'~(deep packet inspection), etc. For each such type $\type \in \types$, we use the set $\SVTypes \subseteq \SV$ to denote the subset of substrate nodes that can host this type of 
network function. To simplify notation we introduce the set $\SRV = \{(\type, u)~| \type \in \types, u \in \SVTypes\}$ to denote all \emph{node} resources and denote by $\SR = \SRV \cup \SE$ the set of all \emph{substrate} resources.
Accordingly, the processing capabilities of substrate nodes and the available bandwidth on substrate edges are given by the function $\Scap: \SR \to \preals$. Hence, for each type and substrate node we use a single numerical value to describe the node's processing capability, e.g. given as the maximal throughput in Mbps. 
Additionally, we also allow to reference substrate node locations via types. To this end we introduce for each substrate node $u \in \SV$ the \emph{abstract} type $\textnormal{loc\_u} \in \types$, such that $\SVTypes[\textnormal{loc\_u}] = \{u\}$ and~$\Scap(\textnormal{loc\_u},u)~= \infty$ and $\Scap(\textnormal{loc\_u},v)~= 0$ for nodes $v \in \SV \setminus \{u\}$.

The set of service chain requests is denoted by~$\requests$. 
A request~$\req \in \requests$ is a directed chain graph~$\VG =~(\VV,\VE)$ 
with start node~$\Vstart \in \VV$ and end node~$\Vend \in \VV$. 
Each of these virtual nodes corresponds to a specific network function type 
which is given via the function~$\Vtype : \VV \to \types$. We assume 
that the types of~$\Vstart$ and~$\Vend$ denote specific nodes in the substrate. 
Edges of the service chain represent forwarding paths. 
Since the type for each node is well-defined, we again use consolidated 
capacities or demands~$\Vcap : \VV \cup \VE \to \preals$ for both edges and functions of any type. 
Note that capacities on edges may differ, as for instance, a 
function `WAN optimizier' can compress traffic. 

In the problem variant with admission control, requests~$\req \in \requests$ are attributed with a 
certain profit or benefit~$\Vprofit \in \preals$. On the other hand, costs are defined via $\Scost: \SR \to \preals$. Note that this definition allows to assign different costs for using \emph{the same} network function \emph{on different substrate nodes}. This allows us to model 
scenarios where, e.g., 
a firewall~(`FW')~costs more in terms of management overhead if implemented using a particular hardware appliance, than 
if it is implemented as a virtual machine~(`VM')
on commodity hardware.

We first define the notion of valid mappings, i.e. embeddings that obey the request's function types and connection requirements:

\begin{definition}[Valid Mapping]
\label{def:valid-mapping}
A valid mapping~$\map$ of request~$\req \in \requests$ is a tuple~$(\mapV, \mapE)$ of 
functions. The function~$\mapV : \VV \to \SV$ maps each virtual network functions to 
a single substrate node. The function~$\mapE : \VE \to \mathcal{P}(\SE)$ maps edges between network functions onto paths in the substrate network, such that:
\begin{itemize}
\item All network functions $i \in \VV$ are mapped onto nodes that can host the particular function type. Formally, $\mapV(i) \in  \SVTypes[\Vtype(i)]$ holds for all~$i \in \VV$.
\item The edge mapping $\mapE$ connects the respective network functions using simple paths, i.e. 
given a virtual edge~$(i,j) \in  \VE$ the embedding~$\mapE(i,j)$ is an edge-path~$\langle (v_1,v_2), \dots, (v_{k-1},v_k)~\rangle$ such that~$(v_l,v_{l+1}) \in  \SE$ for~$1 \leq l < k$ and~$v_1 = \mapV(i)$ and~$v_k = \mapV(j)$.
\end{itemize}
\end{definition}

Next we define the notion of a feasible embedding for a set of requests, i.e. an embedding that obeys the network function and edge capacities.

\begin{definition}[Feasible Embedding]
\label{def:feasible-embedding}
A feasible embedding of a subset of requests~$\requestsP \subseteq \requests$ is given by 
valid mappings~$\map =~(\mapV, \mapE)$ for~$\req \in \requestsP$, such that network function and edge capacities are obeyed:
\begin{itemize}
\item For all types~$\type \in \types$ and nodes~$u \in \SVTypes$ holds: 
$\sum_{\req \in \requestsP} \sum_{i \in \VV, \mapV(i)~= u} \Vcap(i)~\leq \Scap(\type, u)$~.\vspace{6pt}
\item For all edges~$(u,v) \in  \SE$ holds: \\$\sum_{\req \in \requestsP} \sum_{(i,j) \in  \VE: (u,v) \in  \mapE(i,j)~} \Vcap(i,j)~\leq \Scap(u,v)$~.
\end{itemize}
\end{definition}

We first define the SCEP variant with admission control whose objective is to maximize the net profit (SCEP-P), i.e. the achieved profit for embedding a subset of requests.  

\begin{definition}[SCEP for Profit Maximization: SCEP-P]~\\[-12pt]
\begin{description}
\label{def:scep-with-admission-control}
\item[Given:] A substrate network~$\SG =~(\SV,\SE)$ and a set of requests~$\requests$ as described above.
\item[Task:] Find a subset~$\requestsP \subseteq \requests$ of requests to embed and a feasible embedding, given by a mapping $\map$ for each request~$\req \in \requestsP$, maximizing the net profit $\sum_{\req \in \requestsP} \Vprofit$.
\end{description}
\end{definition}

In the variant without admission control, i.e. when all given requests must be embedded, we consider the natural objective of minimizing the cumulative cost of all embeddings. Concretely, the cost of the mapping $\map$ of request $\req \in \requests$ is defined as the sum of costs for placing network functions
 plus the number of substrate links along which network bandwidth 
 needs to be reserved, times the (processing or bandwidth)~demand:
\begin{alignat}{1}
\label{eq:cost-definition}
\begin{array}{rl}
c(\map)~= & \sum_{i \in \VV} \Vcap(i)~\cdot \Scost(\Vtype(i), \mapV(i))~+ \\
          & \sum_{(i,j) \in  \VE} \Vcap(i,j)~\sum_{(u,v) \in  \mapE(i,j)} \Scost(u,v)
\end{array}
\end{alignat}

The variant SCEP-C without admission control which asks for minimizing the costs is hence defined as follows.

\begin{definition}[SCEP for Cost Minimization: SCEP-C]~\\[-12pt]
\label{def:scep-without-admission-control}
\begin{description}
\item[Given:] A substrate network~$\SG =~(\SV,\SE)$ and a set of requests~$\requests$ as described above.
\item[Task:] Find a feasible embedding $\map$ for all requests~$\req \in \requests$ of minimal cost~$\sum_{\req \in \requests} c(\map)$.
\end{description}
\end{definition}

\subsection{NP-Hardness}
\label{sec:np-hardness}

Both introduced SCEP variants are strongly NP-hard, i.e. they are hard independently of the parameters as e.g. the capacities. We prove the NP-hardness 
by establishing a connection to multi-commodity flow problems.
Concretely, we present a polynomial time reduction from the Unsplittable Flow (USF)~and the Edge-Disjoint Paths (EDP)~problems~\cite{Guruswami2003473} to the  the respective SCEP variants. Both USF and EDP are defined on a (directed)~graph $G=(V,E)$ with capacities $d : E \to \preals$ on the edges. The task is to route a set of $K$ commodities $(s_k,t_k)$ with demands $d_k \in \preals$ for $1 \leq k \leq K$ from  $s_k \in V$ to $t_k \in V$ along simple paths inside $G$. Concretely, EDP considers the decision problem in which both the edge capacities and the demands are $1$ and the task is to find a feasible routing. The variant of EDP asking for the maximum number of routable commodities was one of Karp's original 21 NP-complete problems and the decision variant was shown to be NP-complete even on series-parallel graphs~\cite{Nishizeki2001177}. In the USF problem, for each commodity an additional benefit $b_k \in \preals$, $1 \leq k \leq K$, is given and the task is to find a selection of commodities to route, such that capacities are not violated and the sum of benefits of the selected commodities is maximized. Solving the USF problem is NP-hard and proven to be hard to approximate within a factor of $|E|^{1/2-\varepsilon}$ for any $\varepsilon > 0$~\cite{Guruswami2003473}.

We will argue in the following that EDP can be reduced to SCEP-C and USF can be reduced to SCEP-P. Both reductions use the same principal idea of expressing the given commodities as requests. Hence, we first describe this construction before discussing the respective reductions. For commodities $(s_k,t_k)$ with $1\leq k \leq K$ a request $\req_k$ consisting only of the two virtual nodes $i_k$ and $j_k$ and the edge $(i_k,j_k)$ is introduced. By setting $s_{\req_k} = i_k$ and $t_{\req_k} = j_k$ and $\Vtype(i_k)~= \textnormal{loc\_}s_{k}$ and $\Vtype(j_k)~= \textnormal{loc\_}t_{k}$, we can enforce that flow of request $\req_k$ originates at $s_k \in V$ and terminates at $t_k \in V$, hence modeling the original commodities. In both reductions presented below, we do not make use of network functions, i.e. $\types = \{\textnormal{loc\_u} | u \in \SV\}$, and accordingly we do not need to specify network function capacities.

Regarding the polynomial time reduction from EDP to SCEP-C, we simply use unitary virtual demands and substrate capacities. As this yields an equivalent formulation of EDP, which is NP-hard, \emph{finding} a feasible solution for SCEP-C is NP-hard. Hence, there cannot exist an approximation algorithm that (always)~finds a feasible solutions within polynomial time unless $P = \mathit{NP}$ or unless capacity violations are allowed. 

Regarding the reduction from USF to SCEP-P, we adopt the demands by setting $\Vcap[{\req_k}](i_k,j_k)~\triangleq d_k$ for $1 \leq k \leq K$, adopt the network capacities via $\Scap(u,v)~\triangleq d(u,v)$ for $(u,v) \in  E$, and setting the profits accordingly $b_{\req_k} \triangleq b_k$ for $1 \leq k \leq K$. It is easy to see, that any solution to this SCEP-P instance also induces a solution to the original USF instance. It follows that SCEP-P is strongly NP-hard.

\subsection{Further Notation}

We generally denote directed graphs by~$G=(V,E)$. We use~$\delta^+_E(u)~:= \{(u,v) \in  E \}$ to denote the outgoing edges of node~$u \in V$ with respect to~$E$ and similarly define~$\delta^-_E(u)~:= \{(v,u) \in  E \}$ to denote the incoming edges. If the set of 
edges~$E$ can be derived from the the context, we often omit stating~$E$ explicitly.
When considering functions on tuples, we often omit the (implicit)~braces around a tuple and write e.g. $f(x,y)$ instead of $f((x,y))$.
Furthermore, when only some specific elements of a tuple are of importance, we write $(x,\cdot) \in  Z$ in favor of $(x,y) \in  Z$.

\section{Decomposing Linear Solutions}\label{sec:decompo}

In this section, we lay the foundation for the approximation algorithms for both SCEP variants by introducing Integer Programming (IP)~formulations to compute optimal embeddings (see Section~\ref{sec:integer-linear-program}). Given the NP-hardness of the respective problems, solving any of the IPs to optimality is not possible within polynomial time (unless $P = \mathit{NP}$). Hence, we consider the linear relaxations of the respective formulations instead, as these naturally represent a conical or convex combination of valid mappings.  We formally show that linear solutions can be decomposed into valid mappings in Section~\ref{sec:decomp-algorithm-synopsis}. Given the ability to decompose solutions, we apply randomized rounding techniques in Sections~\ref{sec:randround} and \ref{sec:approximation-without-admission-control} to obtain tri-criteria approximation algorithms for the respective SCEP variants.

\begin{figure}[t!]
\includegraphics[width=1\columnwidth]{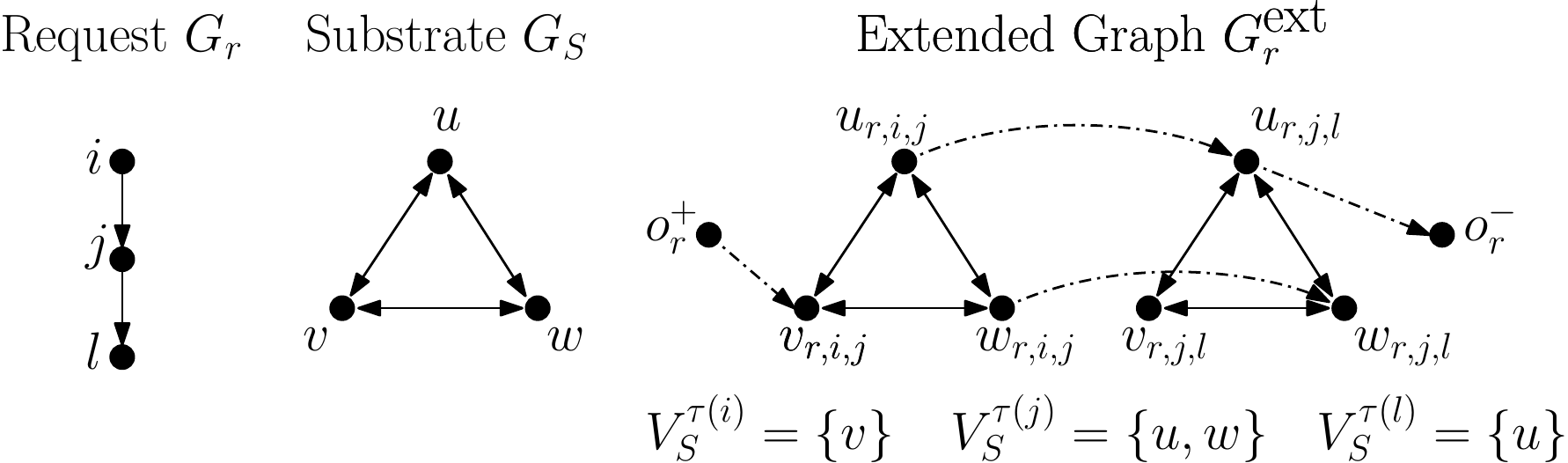}
\caption{Example for the extended graph construction (cf. Definition~\ref{def:extended-graph}).}
\label{fig:extended-graph-linear}
\end{figure}

\subsection{Integer Programming}
\label{sec:integer-linear-program}

To formulate the service chain embedding problems as Integer Programs we employ a flow formulation on a graph construction reminiscent of the one used by Merlin~\cite{merlin}. Concretely, we construct an extended and layered graph consisting of copies of the substrate network together with super sources and sinks. The underlying idea is to model the usage (and potentially the placement)~of network functions by traversing inter-layer edges while intra-layer edges will be used for connecting the respective network functions. Figure~\ref{fig:extended-graph-linear} depicts a simple example of the used graph construction. The request $\req$ consists of the three nodes $i$, $j$, and $l$. Recall that we assume that the start and the end node specify locations in the substrate network (cf. Section~\ref{sec:model}). Hence, in the example the start node $s_\req = i$ and the end node $t_\req = l$ can only be mapped onto the substrate nodes $v$ and $u$ respectively, while the virtual node $j$ may be placed on the substrate nodes $u$ and $w$. Since for each connection of network functions a copy of the substrate network is introduced, the edges between these \emph{layers} naturally represent the utilization of a network function.
Additionally, the extended graph $\VGext$ contains a single super source $\Vsource$ and a super sink $\Vsink$, such that any path from $\Vsource$ to $\Vsink$ represents a valid mapping of the request (cf. Discussion in Section~\ref{sec:decomp-algorithm-synopsis}). Formally, the extended graph \emph{for each request} $r \in \requests$ is introduced as follows.
\begin{definition}[Extended Graph]
\label{def:extended-graph}
Let~$\req \in \requests$ be a request. The extended graph~$\VGext =~(\VVext, \VEext)$ is defined as follows:
\begin{alignat}{1}
\VVext = &  \,\{ \Vsource,\Vsink \} \cup \{u^{i,j}_{\req} | (i,j) \in \VE, u \in \SV \} \\[6pt]
 \VEext = & 
\begin{array}{l}
\{ (u^{i,j}_{\req}, v^{i,j}_{\req}) | (i,j) \in \VE, (u,v) \in \SE \} \cup \\
\{ (\Vsource,u^{\Vstart, j}_{\req}) | (\Vstart,j) \in \VE, u \in \SVTypes[\Vtype(\Vstart)] \} \cup \\
\{ (u^{i,\Vend}_{\req}, \Vsink) | (i,\Vend) \in \VE, u \in \SVTypes[\Vtype(\Vend)] \} \cup \\
 \{ (u^{i,j}_{\req}, u^{j,k}_{\req}) | (i,j), (j,k) \in \VE, u \in \SVTypes[\Vtype(j)] \}
\end{array}
\end{alignat}
We denote by $\VEextHorizontal = \{((u^{i,j}_{\req}, v^{i,j}_{\req}),(i,j)) | (i,j) \in \VE\}$ all copies of the substrate edge $(u,v) \in \SE$ together with the respective virtual edge $(i,j) \in \VE$. Similarly, we denote by $\VEextVertical = \{((u^{i,j}_{\req}, u^{j,k}_{\req}),j) | j \in \VV, \Vtype(j) = \type, (i,j), (j,k) \in \VE\}$ the edges that indicate that node $u \in \SV$ processes flow of network function $j \in \VV$ having type $\type \in \types$.
\end{definition}

Having defined the extended graph as above, we will first discuss our Integer Program~\ref{alg:SCEP-IP} for SCEP-P. We use a single 
variable~$x_{\req} \in \{0,1\}$ per request~$\req$ to indicate whether
the request is to be embedded or not. If~$x_{\req} = 1$, then Constraint~(\ref{alg:SCEP:FlowInduction})~induces a unit flow 
from~$\Vsource$ to $\Vsink$ in the extended graph $\VGext$ using the flow variables 
$f_{\req,e} \in \{0,1\}$ for~$e \in \VEext$. Particularly, Constraint~(\ref{alg:SCEP:FlowPreservation})~states flow preservation at each node, except at the source and the sink.

Constraints~(\ref{alg:SCEP:nodeLoad})~and~(\ref{alg:SCEP:edgeLoad})~
compute the effective load \emph{per request} on the network functions 
and the substrate edges. Towards this end, variables~$l_{\req,x,y} \geq 0$ indicate the load induced by request $r \in \requests$ on resource $(x,y) \in  \SR$. By the construction of the extended graph (see Definition~\ref{def:extended-graph}), 
the sets~$\VEextHorizontal$ and~$\VEextVertical$  actually represent 
a partition of all edges in the extended graph for $\req \in \requests$. Since each 
layer represents a virtual connection~$(i,j) \in \VE$ with a specific load 
$\Vcap(i,j)$ and each edge between layers~$(i,j) \in  \VE$ and~$(j,k) \in  \VE$ 
represents the usage of the network function~$j$ with demand 
$\Vcap(j)$, the unit flow is scaled by the respective demand. Constraint~(\ref{alg:SCEP:capacities})~ensures the feasibility of the embedding (cf.~Definition~\ref{def:feasible-embedding}), i.e., 
the overall amount of used resources does not exceed the offered capacities 
(on network functions as well as on the edges). 

{
 \LinesNotNumbered
 \renewcommand{\arraystretch}{1.5}

 \begin{IPFormulation}{t!}

 \SetAlgorithmName{Integer Program}{}{{}}

 \scalebox{0.93}{

 \newcommand{\spaceIt}{\qquad\quad\quad}
 \newcommand{\miniSpace}{\hspace{1.5pt}}

 \centering
 \hspace{-32pt}
 \noindent
 \begin{tabular}{FRLQB}
  \multicolumn{4}{C}{\textnormal{max~}  \sum \limits_{\req \in \requests} \Vprofit \cdot x_{\req} } & 
  \tagIt{alg:SCEPObj}\\[10pt]
  \sum \limits_{e \in \delta^+(\Vsource)} f_{\req, e} & = & x_{\req} & \forall \req \in \requests, i \in \VV & \tagIt{alg:SCEP:FlowInduction} \\
	\sum \limits_{(u,v)\in \delta^+(u)} f_{\req, e}  & = & \hspace{-12pt} \sum \limits_{(v,u) \in \delta^-(u)} \hspace{-8pt} f_{\req, e}   &  \forall \req \in \requests, u \in \VVext \setminus \{ \Vsource, \Vsink \} & \tagIt{alg:SCEP:FlowPreservation} \\

	\sum \limits_{(e,i)\in \VEextVertical} \hspace{-16pt}\Vcap(i)\cdot f_{\req,e}  & = & l_{\req,\type,u} & \forall \req \in \requests,  (\type,u)\in \SRV & \tagIt{alg:SCEP:nodeLoad} \\
	
	\sum \limits_{	(e,i,j)\in \VEextHorizontal} \hspace{-16pt} \Vcap(i,j)\cdot f_{\req,e} & = & l_{\req,u,v} & \forall \req \in \requests, (u,v)\in \SE & \tagIt{alg:SCEP:edgeLoad} \\

\sum \limits_{\req \in \requests } l_{\req,x,y}& \leq & \Scap(x,y)& \forall (x,y)\in \SR & \tagIt{alg:SCEP:capacities} \\

  x_{\req} & \in & \{0,1\} &  \forall \req \in \requests & \tagIt{alg:SCEP:embedding_variable} \\
  f_{\req,e} & \in & \{0,1\} &  \forall \req \in \requests, e \in \VEext     & \tagIt{alg:SCEP:pathVariable} \\
  l_{\req,x,y} & \geq & 0 &  \forall \req \in \requests, (x,y)\in \SR     & \tagIt{alg:SCEP:VariableLoad} \\
 \end{tabular}
 }
 \caption{SCEP-P}
 \label{alg:SCEP-IP}
 \end{IPFormulation}
 }
 
{

 \newcolumntype{F}{>{$\displaystyle\,}r<{$}@{\hspace{0.0em}}}
 \newcolumntype{C}{>{$\displaystyle\,}c<{$}@{\hspace{0.0em}}}
 \newcolumntype{B}{>{$\displaystyle\,}r<{$}@{\hspace{0.0em}}}
 \newcolumntype{R}{>{$\displaystyle}r<{$}@{\hspace{0.2em}}}
 \newcolumntype{S}{>{$\displaystyle}r<{$}@{\hspace{0.2em}}}
 \newcolumntype{L}{>{$\displaystyle}l<{$}@{\hspace{0.2em}}}
 \newcolumntype{Q}{>{$\displaystyle}l<{$}@{\hspace{0.3em}}}

\LinesNotNumbered
\renewcommand{\arraystretch}{1.5}

\begin{IPFormulation}{t!}

\SetAlgorithmName{Integer Program}{}{{}}

\scalebox{0.94}{

\newcommand{\spaceIt}{\qquad\quad\quad}
\newcommand{\miniSpace}{\hspace{1.5pt}}

\centering
\hspace{-32pt}
\begin{tabular}{FFCLLBB}
\multicolumn{5}{C}{\textnormal{min~}   \sum \limits_{\req \in \requests} \sum \limits_{(x,y) \in   \SR} \Scost(x,y)~\cdot l_{\req,x,y} } & ~ & \tagIt{alg:SCEPObj_without}\\[10pt]
                  
\qquad \qquad \qquad \qquad \qquad &  \textnormal{\ref{alg:SCEP:FlowInduction} - \ref{alg:SCEP:capacities}} & ~\textnormal{and}~  & \textnormal{\ref{alg:SCEP:pathVariable} - \ref{alg:SCEP:VariableLoad}} & & & \\
&  x_{\req}  &=&  1 & \qquad \forall \req \in \requests & & \tagIt{alg:SCEP:embedding_variable_without} \\
\end{tabular}
}
\caption{SCEP-C}
\label{alg:SCEP-IP_without}
\end{IPFormulation}
}

Lastly, the objective function sums up the benefits of embedded requests $\req \in \requests$ for which $x_{\req} = 1$ holds (cf. Definition~\ref{def:scep-with-admission-control}).

In the following, we shortly argue that any feasible solution to IP~\ref{alg:SCEP-IP} induces a feasible solution to SCEP-P (and vice versa). If a request $r \in \requests$ is not embedded, i.e. $x_\req = 0$ holds, then no flow and hence no resource reservations are induced. If on the other hand $x_\req = 1$ holds for $\req \in \requests$, then the flow variables $\{f_{\req,e} | e \in \VEext\}$ induce a unit $\Vsource$-$\Vsink$ flow. By construction, this unit flow must pass through all \emph{layers}, i.e. copies of the substrate network. As previous layers are not reachable from subsequent ones, cycles may only be contained inside a single layer. Hence, network function mappings are uniquely identified by considering the inter-layer flow variables. Thus, there must exist unique nodes at which flow enters and through which the flow leaves each layer. Together with the flow preservation this implies that the respective network functions are connected by the edges \emph{inside} the layers, therefore representing valid mappings.

Considering SCEP-C, we adapt the IP~\ref{alg:SCEP-IP} slightly to obtain the Integer Program~\ref{alg:SCEP-IP_without} for the variant minimizing the costs: (i)~all requests must be embedded by enforcing that $x_\req = 1$ holds for all requests $\req \in \requests$ and (ii)~the objective is changed to minimize the overall resource costs (cf. Equation~\ref{eq:cost-definition}). As the constraints safeguarding the feasibility of solutions are reused, the IP~\ref{alg:SCEP-IP_without} indeed computes optimal solutions for SCEP-C.

While solving Integer Programs~\ref{alg:SCEP-IP} and \ref{alg:SCEP-IP_without} with binary 
variables is computationally hard (cf. Section~\ref{sec:np-hardness}), 
the respective linear relaxations can be computed in polynomial time~\cite{matousek2007understanding}.
Concretely, the linear relaxation is obtained by simply replacing~$\{0,1\}$ with~$[0,1]$ in Constraints~(\ref{alg:SCEP:embedding_variable})~and 
(\ref{alg:SCEP:pathVariable})~respectively. We generally denote the set of feasible solutions to the linear relaxation by~$\FeasibleLP$ and the set of feasible solutions to the respective integer program by~$\FeasibleIP$. We omit the reference to any particular formulation here as it will be clear from the context. We recall the following well-known fact:
\begin{fact}
\label{obs:IP-solutions-are-a-subset-of-LP-solutions}
$\FeasibleIP \subseteq \FeasibleLP$.
\end{fact}
The above fact will e.g. imply that the profit of the optimal linear solution will be higher than the one of the optimal integer solution.

\subsection{Decomposition Algorithm for Linear Solutions}
\label{sec:decomp-algorithm-synopsis}

As discussed above, any binary solution to the formulations~\ref{alg:SCEP-IP} and \ref{alg:SCEP-IP_without} represents a feasible solution to the respective problem variant. However, as we will consider solutions to the respective linear relaxations instead, we shortly discuss how relaxed solutions can be decomposed into conical (SCEP-P)~or convex combinations (SCEP-C)~of valid mappings. Concretely, Algorithm~\ref{alg:decompositionAlgorithm} computes a set of triples $\PotEmbeddings = \{\decomp = (\prob,\mapping,\load)\}_k$, where $\prob \in [0,1]$ denotes the (fractional)~embedding value of the $k$-th decomposition, and $\mapping$ and $\load$ represent the (valid)~mapping and the induced loads on network functions and edges respectively. Importantly, the load function $\load : \SR \to \preals$ represents the cumulative loads, when embedding the request $\req \in \requests $ \emph{fully} according to the $k$-th decomposition.

The pseudocode for our decomposition scheme 
is given in Algorithm~\ref{alg:decompositionAlgorithm}. 
For each request~$\req \in \requests$, a path 
decomposition is performed from $\Vsource$ to $\Vsink$ as long as the outgoing flow from the source $\Vsource$ is larger than 0. Note that the flow variables are an input and originate from the  linear program for SCEP-C or SCEP-P respectively.
We use~$\VGextFlow$ to denote the graph in which an edge $e \in \VEext$ is contained, iff. the flow value along it is greater 0, i.e. for which 
$f_{\req,e} > 0$ holds. Within this graph 
an arbitrary~$\Vsource$-$\Vsink$ path~$P$ is chosen and the minimum available 
`capacity' is stored in~$f^k_{\req}$.

\begin{figure}[tbhp]

\scalebox{0.95}{
\begin{minipage}{1.03\columnwidth}

\begingroup
\removelatexerror

\begin{algorithm*}[H]

\SetKwInOut{Input}{Input}\SetKwInOut{Output}{Output}
\SetKwFunction{LP}{LP}
\SetKwFunction{LP}{LP}
\SetKwFunction{BFS}{BFS}

\newcommand{\SET}{\textbf{set~}}
\newcommand{\APPEND}{\textbf{append~}}
\newcommand{\DEFINE}{\textbf{define~}}
\newcommand{\AND}{\textbf{and~}}
\newcommand{\LET}{\textbf{let~}}
\newcommand{\WITH}{\textbf{with~}}
\newcommand{\COMPUTE}{\textbf{compute~}}
\newcommand{\CHOOSE}{\textbf{choose~}}
\newcommand{\DECOMPOSE}{\textbf{decompose~}}
\newcommand{\FORALL}{\textbf{for all~}}
\newcommand{\FOR}{\textbf{for~}}
\newcommand{\OBTAIN}{\textbf{obtain~}}
\newcommand{\WITHPROBABILITY}{\textbf{with probability~}}

\Input{Substrate~$\SG=(\SV,\SE)$, set of requests~$\requests$,\\ \,~solution~$(\vec{x},\vec{f},\vec{l}) \in  \FeasibleLP$}
\Output{Fractional embeddings~$\PotEmbeddings = \{(\prob,\mapping,\load)\}_k$ \\ \,~for each $\req \in R$}

\For{$\req \in \requests$}
{
	\SET~$\PotEmbeddings  \gets \emptyset $ \AND $k \gets 1$ \\
	\While{$\sum_{e \in \delta^+(\Vsource)} f_{\req,e} > 0$ }
	{
		\CHOOSE $P = \langle \Vsource, \dots, \Vsink \rangle \in \VGextFlow$ \label{alg:decomposition:chooseP}\\
		\SET $f^k_{\req} \gets \min_{e \in P} f_{\req,e}$\\
		\SET $m^k_{\req} = (\mapV,\mapE)~\gets (\emptyset,\emptyset)$ \label{alg:decomposition:startMappingCreation}\\
		\SET $l^k_{\req}(x,y)~\gets 0$ \FORALL $(x,y) \in  \SR $\\
		\For{$i \in \VV$}{ \label{alg:decomp:begin-node-mapping}
			\For{$u \in \SVTypes[\Vtype(i)]$}{
				\uIf{$i = \Vstart$ \textnormal{\textbf{and}} $(\Vsource,u^{\Vstart,\cdot}_{\req}) \in  P$}{
					\SET $\mapV(i)~\gets u$\\
				}
				\uElseIf{$i = \Vend$ \textnormal{\textbf{and}} $(u^{\cdot,i}_{\req},\Vsink) \in  P$}{
					\SET $\mapV(i)~\gets u$ \\
				}
				\uElseIf{$(u^{\cdot,i}_{\req},u^{i, \cdot}_{\req}) \in  P$}{
					\SET $\mapV(i)~\gets u$ \label{alg:decomp:end-node-mapping} \\
				}
			}
			\SET $l^k_{\req}(\Vtype(i),\mapV(i)) \gets l^k_{\req}(\Vtype(i),\mapV(i)) + \Vcap(i)$  \label{alg:decomp:setting-node-load} \\
		}
		\For{$(i,j) \in  \VE$}{
			\SET $\mapE(i,j) \gets \langle (u,v) \in \SE | (u^{i,j}_{\req}, v^{i,j}_{\req}) \in P \rangle$ \label{alg:decomp:extraction-edge-paths}\\
			\For{$(u,v) \in  \mapE(i,j)$}{
				\SET $l^k_{\req}(u,v)~\gets l^k_{\req}(u,v) + \Vcap(i,j)$ \label{alg:decomp:setting-edge-load}\\
			}
		}
		\SET $\PotEmbeddings \gets \PotEmbeddings \cup \{D^k_{\req}\}$ \WITH $D^k_{\req} = (f^k_{\req},m^k_{\req},l^k_{\req})$\\
		\SET $f_{\req,e} \gets f_{\req,e} - f^k_{\req}$ \FORALL $e \in P$  \AND $k \gets k +1$\\
	}
}

\KwRet{$\{\PotEmbeddings | r \in \requests\}$}
\caption{Decomposition Algorithm}
\label{alg:decompositionAlgorithm}
\end{algorithm*}
\endgroup
\end{minipage}
}
\end{figure}

In Lines~\ref{alg:decomp:begin-node-mapping}-\ref{alg:decomp:end-node-mapping} the node mappings are set. 
For all virtual network functions~$i \in \VV$ and all potential substrate nodes $u \in \SVTypes[\Vtype(i)]$, we check whether $u$ hosts $i$ by considering the inter-layer connections contained in $P$.
Besides the trivial cases when $i = \Vsource$ or $i = \Vsink$ holds, the network function $i$ is mapped onto node $u$ iff. edge~$(u_{\req,\cdot,i},u_{\req,i,\cdot})$ is contained 
in~$P$. As $P$ is a directed path from $\Vsource$ to $\Vsink$ and the extended graph does not contain any inter-layer cycles, this mapping is uniquely defined for each found path $P$. 
For the start node~$\Vstart$ of the request~$\req \in \requests$ and the end 
node~$\Vend$ connections from~$\Vsource$ or to~$\Vsink$ are checked respectively.

Concerning the mapping of the virtual edge $(i,j) \in  \VE$, the edges of $P$ used in the substrate edge layer corresponding to the (virtual)~connection $(i,j)$ are extracted in Line~\ref{alg:decomp:extraction-edge-paths}. Note that $P$ is by construction a simple path and hence the constructed edge paths will be simple as well. In Lines~\ref{alg:decomp:setting-node-load} and \ref{alg:decomp:setting-edge-load}, the cumulative load on all physical network functions and edges are computed, that would arise if request $\req \in \requests$ is \emph{fully} embedded according to the $k$-th decomposition.

Lastly, the $k$-th decomposition $D^k_{\req}$, a triple consisting of the fractional embedding value $\prob$, the mapping $\mapping$, and the load $\load$, is added to 
the set of potential embeddings~$\PotEmbeddings$ and the flow variables along $P$ are decreased by $f^k_{\req}$. By decreasing the flow uniformly along $P$, flow preservation with respect to the adapted flow variables is preserved and the next iteration is started by incrementing $k$.

By construction, we obtain the following lemma:
\begin{lemma}
\label{lem:validity-of-kth-decomposition}
Each mapping~$m^k_{\req}$  constructed by Algorithm~\ref{alg:decompositionAlgorithm} in the $k$-th iteration is \emph{valid}.
\end{lemma}

As initially the outgoing flow equals the embedding variable and as flow preservation is preserved after each iteration, the flow in the extended network is fully decomposed by Algorithm~\ref{alg:decompositionAlgorithm}:

\begin{lemma}
\label{lem:sum-of-f-mu}
The decomposition $\PotEmbeddings$ computed in Algorithm~\ref{alg:decompositionAlgorithm} is complete, i.e. $\sum_{D^k_{\req} \in \PotEmbeddings} f^k_{\req} = x_\req$ holds, for $\req \in \requests$.
\end{lemma}

Note that the above lemmas hold independently of whether the linear solutions are computed using IP~\ref{alg:SCEP-IP} or IP~\ref{alg:SCEP-IP_without}. We give two lemmas relating the net profit (for SCEP-P)~and the costs (for SCEP-C)~of the decomposed mappings to the ones computed using the linear relaxations. We state the first without proof as it is a direct corollary of Lemma~\ref{lem:sum-of-f-mu}.

\begin{lemma}
\label{lem:relation-of-net-profit-in-decomposition-and-the-LP-profit}
Let $(\vec{x}, \vec{f}, \vec{l}) \in  \FeasibleLP$ denote a \emph{feasible} solution to the linear relaxation of Integer Program~\ref{alg:SCEP-IP} achieving a net profit of $\hat{B}$ and let $\PotEmbeddings$ denote the respective decompositions of this linear solution for requests $\req \in \requests$ computed by Algorithm~\ref{alg:decompositionAlgorithm}, then the following holds:
\begin{align}
\sum_{\req \in \requests} \sum_{\decomp \in \PotEmbeddings} \prob \cdot \Vprofit = \hat{B}~.
\end{align}
\end{lemma}

While the above shows that for SCEP-P the decomposition always achieves the same profit as the solution to the linear relaxation of IP~\ref{alg:SCEP-IP}, a similar statement holds for SCEP-C and IP~\ref{alg:SCEP-IP_without}:

\begin{lemma}
\label{lem:relation-of-net-profit-in-decomposition-and-the-LP-cost}
Let $(\vec{x}, \vec{f}, \vec{l}) \in  \FeasibleLP$ denote a \emph{feasible} solution to the linear relaxation of Integer Program~\ref{alg:SCEP-IP_without} having a cost of $\hat{C}$ and let $\PotEmbeddings$ denote the respective decompositions of this linear solution computed by Algorithm~\ref{alg:decompositionAlgorithm} for requests $\req \in \requests$, then the following holds:
\begin{align}
\sum_{\req \in \requests} \sum_{\decomp \in \PotEmbeddings} \prob \cdot c(\mapping)~\leq \hat{C}~.
\end{align}
Additionally, equality holds, if the solution $(\vec{x}, \vec{f}, \vec{l}) \in  \FeasibleLP$, respectively the objective $\hat{C}$, is optimal.
\begin{proof}
We consider a single request $\req \in \requests$ and show that $\sum_{\decomp \in \PotEmbeddings} \prob \cdot c(\mapping)~\leq  \sum \limits_{(x,y) \in   \SR} \Scost(x,y)~\cdot l_{\req,x,y}$ holds.
The Integer Program~\ref{alg:SCEP-IP_without} computes the loads on resources $(x,y) \in  \SR$ in Constraints~\ref{alg:SCEP:nodeLoad} and \ref{alg:SCEP:edgeLoad} based on the flow variables, which then drive the costs inside the objective (cf. Constraint~\ref{alg:SCEPObj_without}). Within the decomposition algorithm, only paths $P \in \VGext$ are selected such that $f_{\req, e} > 0$ holds for all $e \in P$. Hence, the resulting mapping $\map$, obtained by extracting the mapping information from $P$, uses only resources previously accounted for in the Integer Program~\ref{alg:SCEP-IP_without}. Since the computation of costs within the IP agrees with the definition of costs applied to the costs of a single mapping (cf. Equation~\ref{eq:cost-definition}), the reduction of flow variables along $P$ by $\prob$ (and the corresponding reduction of the loads)~reduces the cost component of the objective by exactly $\prob \cdot c(\map)$. Thus, the costs accumulated inside the decompositions $\decomp \in \PotEmbeddings$ are covered by the respective costs of the Integer Program~\ref{alg:SCEP-IP_without}.
This proves the inequality. To prove equality, given an optimal solution, consider the following. According to the above argumentation, all costs accumulated within the resulting decomposition $\PotEmbeddings$ are (at least)~accounted for in the IP~\ref{alg:SCEP-IP_without}. Thus, the only possibility that the costs accounted for in the linear programming solution $(\vec{x}, \vec{f}, \vec{l} )$ are greater than the costs accounted for in the decomposition $\PotEmbeddings$ is that the linear programming solution still contains (cyclic)~flows \emph{after} having fully decomposed the request $\req$. As these (cyclic)~flows can be removed without violating any of the constraints while reducing the costs, the given solution cannot have been optimal.
\end{proof}
\end{lemma}

By the above argumentation, it is easy to check that the (fractionally)~accounted resources of the returned decompositions $\PotEmbeddings$ are upper bounded by the resources allocations of the relaxations of Integer Programs~\ref{alg:SCEP-IP} and \ref{alg:SCEP-IP_without}, and hence are upper bounded by the respective capacities (cf. Constraint~\ref{alg:SCEP:capacities}). 

\begin{lemma}
\label{lem:allocations}
The cumulative load induced by the fractional mappings obtained by Algorithm~\ref{alg:decompositionAlgorithm} is less than the cumulative computed in the respective integer program and hence less than the offered capacity, i.e. for all resources $(x,y) \in  \SR$ holds
\begin{align}
\sum_{\req \in \requests} \sum_{D^k_{\req} \in \mathcal{D}_\req} f^k_\req \cdot l^k_{\req}(x,y)~\leq \sum_{\req \in \requests} l_{\req, x,y} \leq  \Scap(x,y)\,,
\label{eq:lem:allocations}
\end{align}
where $l_{\req,x,y}$ refers to the respective variables of the respective Integer Program.
\end{lemma}

\section{Approximating SCEP-P}\label{sec:randround}

This section presents our approximation algorithm for SCEP-P which is based on the randomized rounding of the decomposed fractional solutions of Integer Program~\ref{alg:SCEP-IP}. 
In particular, our algorithm 
provides a tri-criteria approximation with high probability,
that is, 
it computes approximate solutions with performance guarantees for the 
profit and for the maximal violation of capacities of both network functions and
and edges with an arbitrarily high probability. 
We discuss the algorithm in Section~\ref{sec:synopsis-admission-control} and then derive probabilistic bounds for the profit (see Section~\ref{sec:performance-guarantee-obj-admission-control})~and the violation of capacities (see Section~\ref{sec:performance-guarantee-cap-violation-admission-control}). In Section~\ref{sec:main-results-admission-control} the results are condensed using a simple union-bound argument to prove our main theorem, namely that the presented algorithm is indeed a tri-criteria approximation for SCEP-P.

\begin{figure}


\begingroup
\removelatexerror

\begin{algorithm*}[H]

\let\oldnl\nl
\newcommand{\nonl}{\renewcommand{\nl}{\let\nl\oldnl}}

\SetKwInOut{Input}{Input}\SetKwInOut{Output}{Output}
\SetKwFunction{LP}{LP}
\SetKwFunction{LP}{LP}
\SetKwFunction{BFS}{BFS}{}{}

\newcommand{\SET}{\textbf{set~}}
\newcommand{\DEFINE}{\textbf{define~}}
\newcommand{\AND}{\textbf{and~}}
\newcommand{\ST}{\textbf{s.t.~}}
\newcommand{\LET}{\textbf{let~}}
\newcommand{\WITH}{\textbf{with~}}
\newcommand{\COMPUTE}{\textbf{compute~}}
\newcommand{\CHOOSE}{\textbf{choose~}}
\newcommand{\SAMPLE}{\textbf{sample~}}
\newcommand{\SELECT}{\textbf{select~}}
\newcommand{\DECOMPOSE}{\textbf{decompose~}}
\newcommand{\FORALL}{\textbf{for all~}}
\newcommand{\OBTAIN}{\textbf{obtain~}}
\newcommand{\WITHPROBABILITY}{\textbf{with probability~}}

\Input{~Substrate~$\SG=(\SV,\SE)$, set of requests~$\requests$, \\
\,~approximation factors~$\alpha, \beta,\gamma \geq 0$,\\
\,~maximal number of rounding tries~$Q$}
\Output{~Approximate solution for SCEP-P}

\COMPUTE solution~$(\vec{x}, \vec{f}, \vec{l})$ of Linear Program~\ref{alg:SCEP-IP}\\
\COMPUTE $\{\PotEmbeddings  | \req \in \requests \}$  using Algorithm~\ref{alg:decompositionAlgorithm}\\

\SET $q \gets 1$\\
\While{$q \leq Q$}{
	\SET $R' \gets \emptyset$ \AND $\hat{m}_{\req} \gets \emptyset$  \FORALL $\req \in \requests$ \label{alg:approx-admission:start-while}\\
	\SET $B \gets 0$\\
	\SET $L[x,y] = 0$ \FORALL $(x,y) \in  \SR$\\
	\For{$\req \in R$}
	{
		\CHOOSE $p \in [0,1]$ uniformly at random \label{alg:approx-admission:choose-p}\\
		\If{$p \leq \sum_{D^k_{\req} \in \PotEmbeddings} \prob$ }{
			\SET $\hat{k} \gets \min \{ k \in \{1, \dots, |\PotEmbeddings|\}| \sum^k_{l=1} \prob[\req][l] \geq p\}$\\
			\SET $\hat{m}_{\req} \gets \mapping[\req][\hat{k}]$ \AND $\hat{l}_{\req} \gets \load[\req][\hat{k}]$\\
			\SET $R' \gets R' \cup \{\req\}$\\
			\SET $B \gets B + \Vprofit $\\
			\For{$(x,y) \in  \SR$}{
				\SET $L[x,y] \gets L[x,y] + \hat{l}_{\req}(x,y)$ \label{alg:approx-admission:end-while}\\
			}
			
		}
	

	}
		
	\If{\scalebox{0.89}{$\left(
		\begin{array}{l}
			B \geq~ \alpha \cdot \optLP  \\
			\textnormal{\AND}L[\type,u] \leq (1+\beta)~\cdot \Scap(\type,u)~~\textnormal{\textbf{for~}} (\type,u) \in  \SRV \\
			\textnormal{\AND} 	L[u,v] \leq (1+\gamma)~\cdot \Scap(u,v)~~\textnormal{\textbf{for~}}(u,v) \in  \SE
			\end{array}
		\right)$}}{
		\KwRet{$(R',\{\hat{m}_\req | r \in R'\})$}\\
	}
	$q \gets q +1$\\
}
\KwRet{\NULL}
\caption{Approximation Algorithm for SCEP-P}
\label{alg:approxAdmissionControl}
\end{algorithm*}

\endgroup
\end{figure}

\subsection{Synopsis of the Approximation Algorithm~\ref{alg:approxAdmissionControl}}
\label{sec:synopsis-admission-control}
The approximation scheme for SCEP-P is given as Algorithm~\ref{alg:approxAdmissionControl}. Besides the problem specification, the approximation algorithm is handed four additional parameters: the parameters $\alpha$, $\beta$, and $\gamma$ will bound the quality of the found solution with respect to the optimal solution in terms of profit achieved ($0 \leq \alpha \leq 1$), the maximal violation of network function ($0 \leq \beta$)~and edge capacities ($0 \leq \gamma$). As Algorithm~\ref{alg:approxAdmissionControl} is randomized and as we will only show that the algorithm has a constant success probability, the parameter $Q$ controls the number of rounding tries to obtain a solution within the approximation factors $\alpha$, $\beta$, and $\gamma$.

Algorithm~\ref{alg:approxAdmissionControl} first 
uses the relaxation of Integer Program~\ref{alg:SCEP-IP} to compute a 
fractional solution~$(\vec{x}, \vec{f}, \vec{l})$. This solution is then 
decomposed according to Algorithm~\ref{alg:decompositionAlgorithm}, obtaining decompositions~$\PotEmbeddings$ for requests $r \in \requests$. 
The while-loop~(see Lines~4-19)~attempts to construct 
a solution~$(R', \{m_\req | r \in R'\})$ for SCEP-P (cf. Definition~\ref{def:scep-with-admission-control})~according to the following scheme.
Essentially, for every request~$\req \in \requests$ a dice with $|\PotEmbeddings|+1$ many faces is cast, such that $D^k_{\req} \in \PotEmbeddings$ is chosen with probability $\prob$ and none of the embeddings is selected with probability $1-\sum_{D^k_{\req} \in \PotEmbeddings} \prob$. Within the algorithm, casting the dice is done by uniformly selecting a value $p$
 in the range of $[0,1]$ such that the $\hat{k}$-th
decomposition~$D^{\hat{k}}_{\req} = (\prob[\req][\hat{k}], \mapping[\req][{\hat{k}}],  \load[\req][{\hat{k}}]) \in  \PotEmbeddings$ is chosen iff. $\sum^{\hat{k}}_{l=1} \prob[\req][l] \leq p < \sum^{{\hat{k}}+1}_{l=1} \prob[\req][l] $ holds. In case that a mapping was selected, the corresponding mapping and load informations are stored in the (globally visible)~variables $\hat{m}_{\req}$ and $\hat{l}_{\req}$. In Lines~13-16 the request $\req \in \requests$ is added to the set of embedded requests $R'$, the currently achieved objective ($B$)~and the cumulative loads on the physical network functions and edges are adapted accordingly. Note that the load information $\hat{l}_{\req}: \SR \to \preals$ of the decomposition stores the total allocations for each network resource of mapping $\hat{m}_{\req}$.

After having iterated over all requests $r \in \requests$, the obtained solution is returned only if 
the constructed solution achieves at least an $\alpha$-fraction of the objective of the linear program and 
violates node and edge capacities by factors less than~$1+\beta$ and~$1+\gamma$ respectively 
(see Lines~17 and 18). If after~$Q$ iterations no solution within the respective approximation bounds was found, the algorithm returns~$\NULL$. 

In the upcoming sections, the probabilities for finding solutions subject to the parameters $\alpha$, $\beta$, $\gamma$, and $Q$ will be analyzed. Concretely, the analysis of the performance with respect to the objective is contained in Section~\ref{sec:performance-guarantee-obj-admission-control}, while Section~\ref{sec:performance-guarantee-cap-violation-admission-control} proves bounds for capacity violations and  Section~\ref{sec:main-results-admission-control} consolidates the results.

\subsection{Probabilistic Guarantee for the Profit} 
\label{sec:performance-guarantee-obj-admission-control}

To analyze the performance of Algorithm~\ref{alg:approxAdmissionControl} with respect to the achieved profit, we recast the algorithm in terms of random variables. For bounding the profit achieved by the algorithm we introduce the \emph{discrete} random variable $\randVarY \in \{0,\Vprofit\}$, which models the profit achieved by (potentially)~embedding request $\req \in \requests$.
According to Algorithm~\ref{alg:approxAdmissionControl}, request $\req \in \requests$ is embedded as long as the random variable $p$ in Line~\ref{alg:approx-admission:choose-p} was less than $\sum_{\decomp \in \PotEmbeddings} \prob$. Hence, we have that $\ProbVarY[0] = 1 - \sum_{\decomp \in \PotEmbeddings} \prob$ holds, i.e. that the probability to achieve no profit for request $\req \in \requests$ is $1-\sum_{\decomp \in \PotEmbeddings} \prob$. On the other hand, the probability to embed request $\req \in \requests$ equals $\sum_{\decomp \in \PotEmbeddings} \prob$ as in this case some decomposition will be chosen. Hence, we obtain $\ProbVarY[\Vprofit] = \sum_{\decomp \in \PotEmbeddings} \prob$. Given these random variables, we can model the achieved net profit of Algorithm~\ref{alg:decompositionAlgorithm} as: 
\begin{align}
\label{eq:modeling-B-as-random-variables}
	B = \sum_{\req \in \requests} \randVarY ~.
\end{align}

The expectation of the random variable $B$ computes to $\sum_{\req \in \requests} \sum_{\decomp \in \PotEmbeddings} \prob \cdot \Vprofit$ and by Lemma~\ref{lem:relation-of-net-profit-in-decomposition-and-the-LP-profit} we obtain the following corollary:

\begin{corollary}
\label{cor:RelatingObjectiveLPAndExpectationOfRounding}
Given an \emph{optimal} solution~$(\vec{x}, \vec{f}, \vec{l}) \in  \FeasibleLP$ for the Linear Program~\ref{alg:SCEP-IP} and denoting the objective value of this solution as~$\optLP$, we have:
\[
	\optLP = \sum_{\req \in \requests}\sum_{\decomp \in \PotEmbeddings} \prob \cdot \Vprofit = \Exp(B)~~,
\]
where~$\PotEmbeddings$ denotes the decomposition of~$(\vec{x}, \vec{f}, \vec{l})$ obtained by Algorithm~\ref{alg:decompositionAlgorithm} for requests $\req \in \requests$.
\end{corollary}

To bound the probability of achieving a fraction of the profit of the optimal solution, we will make use of the following Chernoff-Bound over \emph{continuous} random variables.

\begin{theorem}[Chernoff-Bound \cite{dubhashi2009concentration}]
\label{thm:chernoff}
Let $X = \sum_{i = 1}^n X_i$ be a sum of $n$ independent random variables $X_i \in [0,1]$, $1 \leq i \leq n$. Then the following holds for any $0 < \varepsilon < 1$:
\begin{align}
\mathbb{P} \big( X \leq (1-\varepsilon)\cdot \mathbb{E}(X)~\big)~\leq exp(-\varepsilon^2\cdot \mathbb{E}(X)/2)
\end{align}
\end{theorem}

Note that the above theorem only considers sums of random variables which are contained within the interval $[0,1]$. In the following, we lay the foundation for appropriately rescaling the random variables $\randVarY$ such that these are contained in the interval $[0,1]$, while still allowing to bound the expected profit. To this end, we first show that we may assume that all requests can be \emph{fully} fractionally embedded \emph{in the absence of other requests} as otherwise the respective requests cannot be embedded given additional requests. Then we will effectively rescale the random variables $\randVarY$ by dividing by the maximal net profit that can be obtained by embedding a \emph{single} request in the \emph{absence of other requests}.

\begin{lemma}
\label{lem:all-requests-can-be-embedded}
We may assume without loss of generality that all requests can be fully fractionally embedded in the absence of other requests.
\begin{proof}
To show the claim we argue that we can extend the approximation scheme by a simple preprocessing step that filters out any request that cannot be fractionally embedded \emph{alone}, i.e. having all the substrate's available resources at its availability.

Practically, we can compute for each request $\req \in \requests$ the linear relaxation of Integer Program~\ref{alg:SCEP-IP} in the absence of any other requests, i.e. for $\requests' = \{\req\}$. As the profit $\Vprofit$ of request $\req$ is positive, the variable $x_{\req}$ is maximized effectively. If for the optimal (linear)~solution $x_{\req} < 1$ holds, then the substrate capacities are not sufficient to (fully)~fractionally embed the request. Hence, in this case the request $\req$ cannot be embedded (fully)~under any circumstances in the original SCEP-P instance as any valid and feasible mapping $\map$ \emph{for the original problem} would induce a feasible solution with $x_{\req} = 1$ when request $\req$ is embedded alone. Note that this preprocessing step can be implemented in polynomial time, as the relaxation of IP~\ref{alg:SCEP-IP_without} can be computed in polynomial time.
\end{proof}
\end{lemma}

According to the above lemma we can assume that all requests can be fully fractionally embedded.
Let $\VprofitMax = \max_{\req \in \requests} \Vprofit$ denote the maximal profit of any single request. The following lemma shows that any fractional solution to the IP~\ref{alg:SCEP-IP} will achieve at least a profit of $\VprofitMax$.

\begin{lemma}
\label{lem:expected-profit-larger-max}
$\optLP \geq \VprofitMax$ holds, where $\optLP$ denotes the optimal profit of the relaxation of IP~\ref{alg:SCEP-IP}.
\begin{proof}
Let $r' \in \requests$ denote any of the requests having the maximal profit $\VprofitMax$ and let $(\vec{x}_{r'}, \vec{f}_{r'}, \vec{l}_{r'})$ denote the linear solution obtained by embedding the request $r'$ according to IP~\ref{alg:SCEP-IP} in the absence of other requests. Considering the set of original requests we now construct a linear solution $(\vec{x}_{\requests}, \vec{f}_{\requests}, \vec{l}_{\requests})$ over the variables corresponding to the original set of requests $\requests$. Concretely, for $(\vec{x}_{\requests}, \vec{f}_{\requests}, \vec{l}_{\requests})$ we set all variables related to request $r'$  according to the solution $(\vec{x}_{r'}, \vec{f}_{r'}, \vec{l}_{r'})$ and set all other variables to $0$. This is a feasible solution, i.e. $(\vec{x}_{\requests}, \vec{f}_{\requests}, \vec{l}_{\requests}) \in  \FeasibleLP$ holds, which achieves the same profit as the solution $(\vec{x}_{r'}, \vec{f}_{r'}, \vec{l}_{r'})$, namely $\VprofitMax$ by Lemma~\ref{lem:all-requests-can-be-embedded}. Hence, the profit of the optimal linear solution is lower bounded by this particular solution's objective and the claim follows.
\end{proof}
\end{lemma}

The above lemma is instrumental in proving the following probabilistic bound on the profit achieved by Algorithm~\ref{alg:approxAdmissionControl}:

\begin{theorem}
\label{thm:probability-of-not-succeeding-in-objective}
The probability of achieving less than $1/3$ of the profit of an optimal solution is upper bounded by $exp(-2/9)$.
\begin{proof}

Instead of considering $B = \sum_{\req \in \requests} \randVarY$, we consider the sum $B' = \sum_{\req \in \requests} \randVarY'$ over the rescaled variables $\randVarY' = \randVarY / \VprofitMax \in [0,1]$. Obviously $\randVarY' \in [0,1]$ holds.
Choosing $\varepsilon = 2/3$ and applying Theorem~\ref{thm:chernoff} on $B'$ we obtain:
\begin{align}
\mathbb{P} \big( B' \leq (1/3)\cdot \mathbb{E}(B')~\big)~\leq exp(-2\cdot \mathbb{E}(B')/9)\,.
\end{align}
By Lemma~\ref{lem:expected-profit-larger-max} and as $B' \cdot \VprofitMax = B$ holds, we obtain that $\Exp(B')~= \Exp(B)~/ \VprofitMax \geq 1$ holds. Plugging in the \emph{minimal} value of $\Exp(B')$, i.e. $1$, into the equation, we maximize the term $exp(-2\cdot \mathbb{E}(B')/9)$ and hence get \begin{align}
\mathbb{P} \big( B' \leq (1/3)\cdot \mathbb{E}(B')~\big)~\leq exp(-2/9)\,.
\end{align}
By using $B' \cdot \VprofitMax = B$, we obtain
\begin{align}
\mathbb{P} \big( B \leq (1/3)\cdot \mathbb{E}(B)~\big)~\leq exp(-2/9)\,.
\label{eq:probability-of-one-thirds-of-the-profit}
\end{align}
Denoting the optimal profit of Integer Program~\ref{alg:SCEP-IP} by $\optIP$ and the optimal profit of the relaxation by $\optLP$, we note that $\Exp(B)~= \optLP$ holds by Lemma~\ref{lem:relation-of-net-profit-in-decomposition-and-the-LP-profit}. Furthermore, $\optIP \leq \optLP$ follows from Fact~\ref{obs:IP-solutions-are-a-subset-of-LP-solutions}. Thus, 
\begin{align}
\frac{\optIP}{3} \leq 
\frac{\optLP}{3} = 
\frac{\Exp(B)}{3 }
\end{align}
holds, completing the proof together with Equation~\ref{eq:probability-of-one-thirds-of-the-profit}.
\end{proof}
\end{theorem}

\subsection{Probabilistic Guarantees for Capacity Violations}
\label{sec:performance-guarantee-cap-violation-admission-control}
The input to the Algorithm~\ref{alg:approxAdmissionControl} 
encompasses the factors~$\beta, \gamma \geq 0$, such that accepted solutions 
must satisfy~$L[\type,u] \leq (1+\beta)~\cdot \Scap(\type, u)$ for~$(\type,u) \in  \SRV$ and~$L[u,v] \leq (1+\gamma)\cdot \Scap(u,v)$ for $(u,v) \in  \SE$. In the following, we will analyze with which probability the above will hold. Our probabilistic bounds rely on Hoeffding's inequality:
\begin{fact}[Hoeffding's Inequality]
Let~$\{X_i\}$ be independent random variables, such that~$X_i \in [a_i,b_i]$, then the following holds:
\[
\mathbb{P} \Big(\sum_i X_i - \mathbb{E}(\sum_i X_i))~\geq t\Big)~\leq exp \big(-2t^2 /~(\sum_i~(b_i - a_i)^2)\big)
\]
\end{fact}

For our analysis to work, assumptions on the maximal allocation of a single request on substrate 
nodes and edges must be made. We define the maximal load as follows:
\begin{definition}[Maximal Load]
\label{def:maximal-loads}
We define the maximal load per network function and edge for each request as follows:
\begin{alignat}{5}
\maxLoadV && = & \max_{i \in \VV: \Vtype(i)~= \type} \Vcap(i)~& & \quad \textnormal{for}~ \req \in \requests, (\type,u) \in  \SRV & \\
\maxLoadE && = & \max_{~(i,j) \in  \VE} \Vcap(i,j)~& &\quad  \textnormal{for}~ \req \in \requests,~(u,v) \in  \SE & 
\end{alignat}
Additionally, we define the maximal load that a whole request may induce on a network function of type $\type$ on substrate node $u \in \SVTypes$ as the following sum:
\begin{alignat}{3}
\maxLoadVSum && = & \sum_{i \in \VV: \Vtype(i)~= \type} \Vcap(i)~& \quad \textnormal{for}~ \req \in \requests, \type \in \types, u \in \SV
\end{alignat}
\end{definition}

The following lemma shows that we may assume~$\maxLoadV[\req][x][y] \leq \Scap(x,y)$ for all network resources $(x,y) \in  \SR$.

\begin{lemma}
\label{lem:assumptions}
We may assume the following without loss of generality.
\begin{alignat}{3}
\maxLoadV[\req][x][y] && \leq & \Scap(x,y)~& \quad \forall (x,y) \in  \SR
\end{alignat}
\begin{proof}
Assume that $\maxLoadV > \Scap(u)$ holds for some~$\req \in \requests, (\type,u) \in  \SRV$. If this is the case, then there exists a virtual network function~$i \in \VV$ with~$\Vtype(i)~= \type$, such that~$\Vcap(i)~> \Scap(\type,u)$ holds. However, if this is the case, there cannot exist a feasible solution in which~$i$ is mapped onto~$u$, as this mapping would clearly exceed the capacity. Hence, we can remove the edges that indicate the mapping of the respective virtual function $i$ on substrate node $u$ in the extended graph construction a priori before computing the linear relaxation. The same argument holds true if~$\maxLoadE > \Scap(u,v)$ holds for some~$(u,v) \in  \SE$. 
\end{proof}

\end{lemma}
%

We now model the load on the different network functions $(\type,u) \in  \SRV$ and the substrate edges $(u,v) \in  \SE$ induced by each of the requests $\req \in \requests$ as random variables $L_{\req,\type,u} \in [0,\maxLoadVSum]$ and $L_{\req,u,v} \in [0,\maxLoadE \cdot |\VE|]$ respectively. To this end, we note that Algorithm~\ref{alg:approxAdmissionControl} chooses the $k$-th decomposition $\decomp = (\prob, \load,\mapping)$ with probability $\prob$. Hence, with probability $\prob$ the load $\load$ is induced for request $\req \in \requests$. The respective variables can therefore be defined as $
\mathbb{P}(L_{\req,x,y} = \load(x,y))= \prob $ (assuming pairwise different loads)~and $\mathbb{P}(L_{\req,x,y} = 0)=  1 - \sum_{\decomp \in \PotEmbeddings } \prob $ for $(x,y) \in  \SRV$.
Additionally, we denote by $L_{x,y} = \sum_{\req \in \requests} L_{\req,x,y}$ the overall load induced on function resource $(x,y) \in  \SR$.
By definition, the expected load on the network nodes and the substrate edges $(x,y) \in  \SRV$ compute to
\begin{alignat}{3}
\mathbb{E}(L_{x,y})~&& = & \sum_{\req \in \requests} \sum_{\decomp \in \PotEmbeddings} \prob \cdot \load(x,y)~ \,.
\end{alignat}
Together with Lemma~\ref{lem:allocations} these equations yield:
\begin{alignat}{5}
\mathbb{E}(L_{x,y})~&& \leq & \Scap(x,y)~& \quad & \textnormal{for}~ (x,y) \in  \SRV \label{eq:load-vs-capacity}
\end{alignat}
Using the above, we can apply Hoeffding's Inequality:
\begin{lemma}
\label{lem:approximation-single-node}
Let~$\DeltaV = \sum_{\req \in \requests}~(\maxLoadVSum / \maxLoadV)^2$. The probability that the capacity of a single network function~$\type \in \types$ on node~$u \in \SVTypes$ is exceeded by more than a factor~$(1+\sqrt{2\cdot \log~(|\SV|\cdot|\types|)~\cdot \DeltaV})$ is upper bounded by~$(|\SV| \cdot |\types|)^{-4}$.
\begin{proof}
Each variable~$\randVarLNode$ is clearly contained in the interval~$[0, \maxLoadVSum]$ and hence $\sum_{\req \in \requests} (\maxLoadVSum)^2$ will be the denominator in Hoeffding's Inequality. We choose~$t = \sqrt{2\cdot \log~(|\SV|\cdot|\types|)~\cdot \DeltaV  } \cdot \Scap(\type,u)$ and obtain:
\noindent\begin{alignat*}{2}
& \mathbb{P} \Big(L_{\type,u} - \mathbb{E}(L_{\type,u})~\geq t \Big)~&&  \notag \\
& ~~~ \leq exp \Bigg(\frac{-2\cdot t^2}{\sum_{\req \in \requests}~(\maxLoadVSum)^2} \Bigg)~&& \\
& ~~~ \leq exp \Bigg(\frac{-4 \log~(|\SV|\cdot|\types|)~\cdot \DeltaV \cdot  \Scap(\type,u)^2}{\sum_{\req \in \requests}~(\Scap(\type,u)\cdot \maxLoadVSum / \maxLoadV)^2} \Bigg)~&& \numberthis \label{formula:proof:singleNodeViolation:upperCapacities}
\end{alignat*}
\noindent
\begin{alignat*}{2}
& ~~~ = exp \Bigg(\frac{-4 \log~(|\SV|\cdot|\types|)~\sum_{\req \in \requests}~(\maxLoadVSum / \maxLoadV)^2 }{\sum_{\req \in \requests}~( \maxLoadVSum / \maxLoadV)^2} \Bigg)~&& \\
& ~~~ = exp \Bigg(\frac{-4 \log~(|\SV|\cdot|\types|)~\sum_{\req \in \requests}~(\maxLoadVSum / \maxLoadV)^2 }{\sum_{\req \in \requests}~( \maxLoadVSum / \maxLoadV)^2} \Bigg)~&& \\
& ~~~ =~(|\SV|\cdot |\types|)^{-4} && 
\end{alignat*}

In Line~(\ref{formula:proof:singleNodeViolation:upperCapacities})~we have used 
\begin{align}
\maxLoadVSum \leq \Scap(\type,u)~\cdot \maxLoadVSum / \maxLoadV \label{formula:proof:singleNodeViolation:estimation}
\end{align}
and accordingly increased the denominator to increase the probability. It is easy to check, that Equation~(\ref{formula:proof:singleNodeViolation:estimation})~holds, as we may assume that~$\maxLoadV \leq \Scap(\type,u)$ holds~(see Lemma~\ref{lem:assumptions}).
By using Equation~(\ref{eq:load-vs-capacity})~and plugging in $t$ we obtain

{
\smaller[1]
\begin{align*}
\mathbb{P} \Big(L_{\type,u} \geq (1+\sqrt{2\cdot \log~(|\SV|\cdot|\types|)~\cdot \DeltaV})~\cdot \Scap(\type,u)~\Big)~\leq (|\SV|\cdot |\types|)^{-4}
\end{align*}
}
for all $(\type,u) \in  \SRV$, proving our claim.
\end{proof}
\end{lemma}

It should be noted that if network functions are unique within a request, then~$\DeltaV$ equals the number of requests~$|\requests|$, since in this case~$\maxLoadVSum = \maxLoadV$ holds. Next, we consider a very similar result on the capacity violations of substrate edges. However, as in the worst case each substrate edge $(u,v) \in  \SE$ is used $|\VE|$ many times, we have to choose a slightly differently defined $\Delta_E$.

\begin{lemma}
\label{lem:approximation-single-edge}
Let~$\DeltaE = \sum_{\req \in \requests} |\VE|^2$. The probability that the capacity of a single substrate edge~$(u,v) \in  \SE$ is exceeded by more than a factor~$(1+\sqrt{2\cdot \log~(|\SV|)~\cdot \DeltaE})$ is bounded by~$|\SV|^{-4}$.
\begin{proof}
Each variable~$\randVarLEdge$ is clearly contained in the interval~$[0, \maxLoadE \cdot |\VE|]$. We choose~$t = \sqrt{2\cdot \log |\SV| \cdot \DeltaE  } \cdot \Scap(u,v)$ and apply Hoeffdings Inequality:
{
\begin{alignat}{2}
& \mathbb{P} \Big(L_{u,v} - \mathbb{E}(L_{u,v})~\geq t \Big)~&&  \notag \\
& ~~~ \leq exp \Bigg(\frac{-2\cdot t^2}{\sum_{\req \in \requests}~(\maxLoadE \cdot |\VE|)^2} \Bigg)~&& \\
& ~~~ \leq exp \Bigg(\frac{-4 \log |\SV| \cdot \DeltaE \cdot  \Scap(u,v)^2}{\sum_{\req \in \requests} \Scap(u,v)^2\cdot|\VE|^2} \Bigg)~&& \label{formula:proof:singleEdgeViolation:upperCapacities}\\
& ~~~ \leq exp \Bigg(\frac{-4 \log |\SV| \cdot \sum_{\req \in \requests} |\VE|^2 \cdot  \Scap(u,v)^2}{\Scap(u,v)^2 \cdot \sum_{\req \in \requests} |\VE|^2} \Bigg)~&& \label{formula:proof:singleEdgeViolation:upperCapacities_2}\\
& ~~~ = exp \Bigg(\frac{-4 \log |\SV| \sum_{\req \in \requests} |\VE|^2 }{\sum_{\req \in \requests} |\VE|^2} \Bigg)~&& \\
& ~~~ = |\SV|^{-4} && 
\end{alignat}
}
The rest of the proof is analogous to the one of Lemma~\ref{lem:approximation-single-node}:
In Equation~(\ref{formula:proof:singleEdgeViolation:upperCapacities})~we have used again the fact that~$\maxLoadE \leq \Scap(u,v)$ holds~(see Lemma~\ref{lem:assumptions}). In the numerator of Equation~\ref{formula:proof:singleEdgeViolation:upperCapacities_2} we have replaced $\DeltaE$ by its definition and placed $\Scap(u,v)^2$ outside the sum in the denominator. Analogously to the proof of Lemma~\ref{lem:approximation-single-node}, the remaining part of the proof follows from Equation~(\ref{eq:load-vs-capacity}).
\end{proof}

\end{lemma}

We state the following corollaries  without proof, showing that the above shown bounds work nicely if we assume more strict bounds on the maximal loads.

\begin{corollary}
Assume that~$\maxLoadV \leq \varepsilon \cdot \Scap(\type,u)$ holds for~$0 < \varepsilon < 1$ and all~$(\type,u) \in  \SRV$. With~$\DeltaV$ as defined in Lemma~\ref{lem:approximation-single-node}, we obtain:
The probability that in Algorithm~\ref{alg:approxAdmissionControl} the capacity of a single network function~$\type \in \types$ on node~$u \in \SVTypes$ is exceeded by more than a factor~$(1+ \varepsilon \cdot \sqrt{2\cdot \log~(|\SV|\cdot|\types|)~\cdot \DeltaV})$ is upper bounded by~$(|\SV| \cdot |\types|)^{-4}$.
\end{corollary}

\begin{corollary}
Assume that~$\maxLoadE \leq \varepsilon \cdot \Scap(u,v)$ holds for~$0 < \varepsilon < 1$ and all~$(u,v) \in  \SE$. With~$\DeltaE$ as defined in Lemma~\ref{lem:approximation-single-edge}, we obtain:
The probability that in Algorithm~\ref{alg:approxAdmissionControl} the capacity of a single substrate node~$(u,v) \in  \SE$ is exceeded by more than a factor~$(1+\varepsilon \cdot \sqrt{2\cdot \log~(|\SV|)~\cdot \DeltaE})$ is upper bounded by~$|\SV|^{-4}$.
\end{corollary}

\subsection{Main Results}
\label{sec:main-results-admission-control}

We can now state the main tri-criteria approximation results obtained for SCEP-P. First, note that Algorithms~\ref{alg:decompositionAlgorithm} and \ref{alg:approxAdmissionControl} run in polynomial time. The runtime of Algorithm~\ref{alg:decompositionAlgorithm} is dominated by the search for paths and as in each iteration at least the flow of a single edge in the extended graph is set to~$0$, only~$\mathcal{O}(\sum_{\req \in \requests} |\VE| \cdot |\SE|)$ many graph searches are necessary. The runtime of the approximation itself is clearly dominated by the runtime to solve the Linear Program which has~$\mathcal{O}(\sum_{\req \in \requests} |\VE| \cdot |\SE|)$ many variables and constraints and can therefore be solved in polynomial time using e.g. the Ellipsoid algorithm~\cite{matousek2007understanding}.

The following lemma shows that Algorithm~\ref{alg:approxAdmissionControl} can produce solutions of high quality \emph{with high probability}:

\begin{lemma}
Let~$0 < \varepsilon \leq 1$ be chosen minimally, such that~$\maxLoadV[\req][x][y] \leq \varepsilon \cdot \Scap(\type,u)$ holds for all~$(x,y) \in  \SR$.  Setting $\alpha = 1/3$, $\beta = \varepsilon \cdot \sqrt{2\cdot \log~(|\SV|\cdot|\types|)~\cdot \DeltaV  }$ and~$\gamma = \varepsilon \cdot \sqrt{2\cdot \log |\SV|)~\cdot \DeltaE  }$ with~$\DeltaV,\DeltaE$ as defined in Lemmas~\ref{lem:approximation-single-node} and \ref{lem:approximation-single-edge}, the probability that a solution is found within~$Q \in \mathbb{N}$ rounds is lower bounded by~$1 -(19/20)^Q$ for $|\SV| \geq 3$.
\begin{proof}
We apply a union bound argument. By Lemma~\ref{lem:approximation-single-node} the probability that for a single network function of type~$\type \in \types$ on node~$u \in \SVTypes$ the allocations exceed the capacity by more than a factor $(1+\beta)$ is less than~$(|\SV| \cdot |\types|)^{-4}$. Given that there are maximally~$|\SV| \cdot |\types|$ many of network functions overall, the probability that any of these exceeds the capacity by a factor above $(1+\beta)$ is less than~$(|\SV| \cdot |\types|)^{-3} \leq |\SV|^{-3}$. Similarly, by Lemma~\ref{lem:approximation-single-edge} the probability that the edge capacity of a single edge is violated by more than a factor~$(1+\gamma)$ is less than~$|\SV|^{-4}$. As there are at most~$|\SV|^2$ edges, the union bound gives us that the probability that the capacity of any of the edges is violated by a factor larger than~$(1+\gamma)$ is upper bounded by~$|\SV|^{-2}$.  Lastly, by Theorem~\ref{thm:probability-of-not-succeeding-in-objective} the probability of \emph{not} finding a solution having an $\alpha$-fraction of the optimal objective is less or equal to~$exp(-2/9)~\approx 0.8074$. The probability to not find a suitable solution, satisfying the objective and the capacity criteria, within a single round is therefore upper bounded by~$exp(-2/9)~- 1/9 - 1/27 \leq 19/20$ if~$|\SV| \geq 3$ holds. The probability find a suitable solution within~$Q \in \mathbb{N}$ many rounds hence is $1-(19/20)^Q$ for $|\SV| \geq 3$.
\end{proof}
\end{lemma}

\begin{theorem}
\label{thm:result-for-admission-control}
Assuming that $|\SV| \geq 3$ holds, and that $\maxLoadV[\req][x][y] \leq \varepsilon \cdot \Scap(x,y)$ holds for all resources~$(x,y) \in  \SR$ with $0 < \varepsilon \leq 1$ and by setting~$\beta = \varepsilon \cdot \sqrt{2\cdot \log~(|\SV|\cdot|\types|)~\cdot \DeltaV  }$ and~$\gamma = \varepsilon \cdot \sqrt{2\cdot \log |\SV| \cdot \DeltaE  }$ with~$\DeltaV,\DeltaE$ as defined in Lemmas~\ref{lem:approximation-single-node} and \ref{lem:approximation-single-edge}, Algorithm~\ref{alg:approxAdmissionControl} is a~$(\alpha,1+\beta,1+\gamma)$ tri-criteria approximation algorithm for  SCEP-P, such that it finds a solution \emph{with high probability}, that achieves at least an~$\alpha = 1/3$ fraction of the optimal profit and violates network function and edge capacities only within the factors~$1+\beta$ and~$1+\gamma$ respectively.
\end{theorem}

\section{Approximating SCEP-C}

\label{sec:approximation-without-admission-control}
In the previous section we have derived a tri-criteria approximation for the SCEP-P variant that maximizes the profit of embedding requests while only exceeding capacities within certain bounds. We show in this section that the approximation scheme for SCEP-P can be adapted for the cost minimization variant SCEP-C by introducing an additional preprocessing step. 

Recall that the cost variant SCEP-C (see Definition~\ref{def:scep-without-admission-control})~asks for finding a feasible embedding $\map$ for all given requests $\req \in \requests$, such that the sum of costs $\sum_{\req \in \requests} c(\map)$ is minimized. We propose Algorithm~\ref{alg:approxWithoutAdmissionControl} to approximate SCEP-P. After shortly discussing the adaptions necessitated with respect to Algorithm~\ref{alg:approxAdmissionControl}, we proceed to prove the respective probabilistic guarantees analogously to the previous section with the main results contained in Section~\ref{sec:main-results-without-admission-control}.

\subsection{Synopsis of the Approximation Algorithm~\ref{alg:approxWithoutAdmissionControl}}

The approximation for SCEP-C given in Algorithm~\ref{alg:approxWithoutAdmissionControl} is based on Algorithm~\ref{alg:approxAdmissionControl}. Algorithm~\ref{alg:approxWithoutAdmissionControl} first computes a solution to the linear relaxation of Integer Program~\ref{alg:SCEP-IP_without} which only differs from the previously used Integer Program~\ref{alg:SCEP-IP} by requiring to embed all requests and adopting the objective to minimize the overall induced costs (cf. Section~\ref{sec:integer-linear-program}). While for the relaxation of Integer Program~\ref{alg:SCEP-IP} a feasible solution always exists -- namely, not embedding any requests -- this is not the case for the relaxation of Integer Program~\ref{alg:SCEP-IP_without}, i.e. $\FeasibleLP = \emptyset$ might hold. Hence, if solving the formulation was determined to be infeasible, the algorithm returns that no solution exists. This is valid, as $\FeasibleIP \subseteq \FeasibleLP$ holds (cf. Fact~\ref{obs:IP-solutions-are-a-subset-of-LP-solutions})~and, if $\FeasibleLP = \emptyset$ holds, $\FeasibleIP = \emptyset$ must follow. 

Having found a linear programming solution $(\vec{x}, \vec{f}, \vec{l}) \in  \FeasibleLP$, we apply the decomposition algorithm presented in Section~\ref{sec:decomp-algorithm-synopsis} to obtain the set of decomposed embeddings $\PotEmbeddings = \{(\prob,\mapping,\load)\}_k$ for all requests $\req \in \requests$. As the Integer Program~\ref{alg:SCEP-IP_without} enforces that $x_{\req} = 1$ holds for all requests $\req \in \requests$, we derive  the following corollary  from Lemma~\ref{lem:sum-of-f-mu}, stating that the sum of fractional embedding values is one for all requests.

\begin{corollary}
\label{cor:full-fractional-embeddings-cost}
$\sum_{\decomp \in \PotEmbeddings} \prob = 1$ holds for all requests.
\end{corollary}

Lines~5-14 is the core addition of Algorithm~\ref{alg:approxWithoutAdmissionControl} when compared to Algorithm~\ref{alg:approxAdmissionControl}. This preprocessing step effectively removes fractional mappings that are too costly from the set of decompositions $\PotEmbeddings$ by setting their fractional embedding values to zero and rescaling the remaining ones. Concretely, given a request $\req \in \requests$, first the \emph{weighted (averaged)~cost} $\WAC = \sum_{\decomp \in \PotEmbeddings} \prob \cdot c(\mapping)$ is computed. In the next step, the fractional embedding values of those decompositions costing less than two times the weighted cost $\WAC$ is computed and assigned to $\lambda_{\req}$. Then, for each decomposition $\decomp = (\prob,\mapping,\load) \in  \PotEmbeddings$ a new fractional embedding value $\probHat$ is defined: either $\probHat$ is set to zero or is set to $\prob$ rescaled by dividing by $\lambda_{\req}$. The rescaling guarantees, that also for  $\PotEmbeddingsHat$ the sum of newly defined fractional embedding values $\probHat$ equals one, i.e. $\sum_{k=1}^{|\PotEmbeddingsHat|} \probHat = 1$ holds. As there will exist at least a single decomposition $\mapping$ such that $c(\mapping)~\leq 2\cdot \WAC$ holds, the set of decompositions will not be empty and the rest of the algorithm is well-defined. We formally prove this in Lemma~\ref{lem:wac-pruning} in the next section.

\begin{figure}[tbhp]

\scalebox{0.93}{
\begin{minipage}{1.07\columnwidth}

\begingroup
\removelatexerror

\begin{algorithm*}[H]

\let\oldnl\nl
\newcommand{\nonl}{\renewcommand{\nl}{\let\nl\oldnl}}
 
\SetKwInOut{Input}{Input}\SetKwInOut{Output}{Output}
\SetKwFunction{LP}{LP}
\SetKwFunction{LP}{LP}

\newcommand{\SET}{\textbf{set~}}
\newcommand{\DEFINE}{\textbf{define~}}
\newcommand{\AND}{\textbf{and~}}
\newcommand{\LET}{\textbf{let~}}
\newcommand{\WITH}{\textbf{with~}}
\newcommand{\COMPUTE}{\textbf{compute~}}
\newcommand{\CHOOSE}{\textbf{choose~}}
\newcommand{\DECOMPOSE}{\textbf{decompose~}}
\newcommand{\FORALL}{\textbf{for all~}}
\newcommand{\OBTAIN}{\textbf{obtain~}}
\newcommand{\ORDER}{\textbf{order~}}
\newcommand{\EXECUTE}{\textbf{execute~}}
\newcommand{\BREAK}{\textbf{break~}}
\newcommand{\WITHPROBABILITY}{\textbf{with probability~}}

\Input{~Substrate~$\SG=(\SV,\SE)$, set of requests~$\requests$, \\
\,~approximation factors~$\beta,\gamma \geq 0$,\\
\,~maximal number of rounding tries~$Q$}
\Output{~Approximate solution for SCEP-C}

\COMPUTE solution~$(\vec{x}, \vec{f}, \vec{l})$ of Linear Program~\ref{alg:SCEP-IP_without} \\
\If{\textnormal{Linear Program~\ref{alg:SCEP-IP_without} was infeasible}}{
	\KwRet{\textnormal{``no solution exists''}}\\
}
\COMPUTE $\{\PotEmbeddings | \req \in \requests \}$  using Algorithm~\ref{alg:decompositionAlgorithm}\\
\For{$\req \in \requests$}{
	\SET $\WAC \gets \sum_{\decomp \in \PotEmbeddings} \prob \cdot c(\mapping)$\\
	\SET $\lambda_{\req} \gets \sum_{\decomp \in \PotEmbeddings : c(\mapping)~\leq 2 \cdot \WAC } \prob $\\
	\SET $\PotEmbeddingsHat \gets \emptyset$\\
	\For{$(\prob,\mapping,\load) \in  \PotEmbeddings$}
	{
		\eIf{$c(\mapping)~\leq 2\cdot \WAC$}{
			\SET $\probHat \gets \prob / \lambda_{\req}$\\
		}{
			\SET $\probHat \gets 0$ \\
		}
		\SET $\PotEmbeddingsHat \gets \PotEmbeddingsHat \cup \{ (\probHat,\mapping,\load)\}$ \\	
	}
}

\SET $q \gets 1$\\
\While{$q \leq Q$}{
	\SET  $\hat{m}_{\req} \gets \emptyset$  \FORALL $\req \in \requests$ \\
	\SET $L[x,y] = 0$ \FORALL $(x,y) \in  \SR$\\
	\For{$\req \in R$}
	{
		\CHOOSE $p \in [0,1]$ uniformly at random \\
		\SET $\hat{k} \gets \min \{ k \in \{1, \dots, |\PotEmbeddingsHat|\}| \sum^k_{l=1} \probHat[\req][l] \geq p\}$\\
		\SET $\hat{m}_{\req} \gets \mapping[\req][\hat{k}]$ \AND $\hat{l}_{\req} \gets \load[\req][\hat{k}]$\\
		\For{$(x,y) \in  \SR$}{
			\SET $L[x,y] \gets L[x,y] + \hat{l}_{\req}(x,y)$ \\
		}
	}
	\If{\scalebox{0.97}{$\left(
	\begin{array}{l}
		 \hspace{20pt} L[\type,u] \leq (2+\beta)~\cdot \Scap(\type,u)~~\textnormal{\textbf{for~}} (\type,u) \in  \SRV \\
		 \textnormal{\AND} 	L[u,v] \leq (2+\gamma)~\cdot \Scap(u,v)~~\textnormal{\textbf{for~}}(u,v) \in  \SE
			\end{array}
	\right)$}}{
		\KwRet{$\{\hat{m}_\req | \req \in \requests \}$}\\					
	}
	$q \gets q +1$\\
}
\KwRet{\NULL}

\caption{Approximation Algorithm for Embeddings without Admission Control}
\label{alg:approxWithoutAdmissionControl}
\end{algorithm*}
\endgroup
\end{minipage}
 }
\end{figure}

After having preprocessed the decomposed solution of the linear relaxation of Integer Program~\ref{alg:SCEP-IP_without}, the randomized rounding scheme already presented in Section~\ref{sec:randround} is employed. For each request one of the decompositions $\decompHat \in \PotEmbeddingsHat$ is chosen according to the fractional embedding values $\probHat$ (see Lines~20-24). Note that as $\sum_{\decompHat \in \PotEmbeddingsHat} \probHat = 1$ holds, the index $\hat{k}$ will always be well-defined and hence for each request $\req \in \requests$ exactly one mapping $\hat{m}_{\req}$ will be selected. Analogously to Algorithm~\ref{alg:approxAdmissionControl}, the variables $L[x,y]$ store the induced loads by the current solution on substrate resource $(x,y) \in  \SR$ and a solution is only returned if neither any of the  node or any of the edge capacities are violated by more than a factor of $(2+\beta)$ and $(2+\gamma)$ respectively, where $\beta, \gamma > 0$ are again the respective approximation guarantee parameters (cf. Section~\ref{sec:randround}), which are an input to the algorithm. In the following sections, we prove that Algorithm~\ref{alg:approxWithoutAdmissionControl} yields solutions having at most two times the optimal costs and which exceed node and edge capacities  by no more than a factor of $(2+\beta)$ and $(2+\gamma)$ respectively with high probability.

\subsection{Deterministic Guarantee for the Cost}

In the following we show that Algorithm~\ref{alg:approxAdmissionControl} will only produce solutions whose costs are upper bounded by two times the optimal costs. The result follows from restricting the set of potential embeddings $\PotEmbeddingsHat$ to only contain decompositions having less than two times the weighted cost $\WAC$ for each request $\req \in \requests$. To show this, we first reformulate Lemma~\ref{lem:relation-of-net-profit-in-decomposition-and-the-LP-cost} in terms of the weighted cost as follows.

\begin{corollary}
\label{cor:relation-of-net-profit-in-decomposition-and-the-LP-cost-new}
Let $(\vec{x}, \vec{f}, \vec{l}) \in  \FeasibleLP$ denote an \emph{optimal} solution to the linear relaxation of Integer Program~\ref{alg:SCEP-IP_without} having a cost of $\optLP$ and let $\PotEmbeddings$ denote the decomposition of this linear solution computed by Algorithm~\ref{alg:decompositionAlgorithm}, then the following holds:
\begin{align}
\sum_{\req \in \requests} \WAC = \optLP~.
\end{align}
\end{corollary}

The above corollary follows directly from Lemma~\ref{lem:relation-of-net-profit-in-decomposition-and-the-LP-cost} and the definition of $\WAC$ as computed in Line~6. As a next step towards proving the bound on the cost, we show that Algorithm~\ref{alg:approxWithoutAdmissionControl} is indeed well-defined, as $\PotEmbeddingsHat \neq \emptyset$. We even show a stronger result, namely that $\lambda_{\req} \geq 1/2$ holds and hence the cumulative embedding weights $\prob$ of the decompositions $(\prob,\mapping,\load)$ \emph{not set to 0} makes up half of the original embedding weights.

\begin{lemma}
\label{lem:wac-pruning}
The sum of fractional embedding values of decompositions whose cost is upper bounded by two times $\WAC$ is at least $1/2$. Formally,
$\lambda_{\req} \geq 1/2$  holds for all requests $\req \in \requests$.
\begin{proof}
For the sake of contradiction, assume that $\lambda_{\req} < 1/2$ holds for any request $\req \in \requests$. By the definition of $\WAC$ and the assumption on $\lambda_{\req}$, we obtain the following contradiction:
\begin{alignat}{5}
& \WAC ~& = & ~\sum_{\decomp \in \PotEmbeddings} \prob \cdot c(\mapping)~\label{eq:a}\\
&      & \leq & ~\sum_{\decomp \in \PotEmbeddings: c(\mapping)~> 2 \cdot \WAC} \prob \cdot c(\mapping)~\label{eq:b}\\
&      & \leq  & ~\sum_{\decomp \in \PotEmbeddings: c(\mapping)~> 2 \cdot \WAC} \prob \cdot 2\cdot \WAC \label{eq:c}\\
&      & \leq  & ~(1-\lambda_{\req})~\cdot 2 \cdot \WAC \label{eq:d}\\
&      & < & ~\WAC \label{eq:e}
\end{alignat}
For Equation~\ref{eq:a} the value of $\WAC$ as computed in Algorithm~\ref{alg:approxWithoutAdmissionControl} was used. Equation~\ref{eq:b} holds as only a subset of decompositions, namely the ones with mappings of costs higher than two times $\WAC$, are considered and $\prob \geq 0$ holds by definition. The validity of Equation~\ref{eq:c} follows as all the considered decompositions have a cost of at least two times $\WAC$ and Equation~\ref{eq:d} follows as $(1-\lambda_{\req})~> 1/2$ holds by assumption. Lastly, Equation~\ref{eq:e} yields the contradiction, showing that indeed $\lambda_{\req} \geq 1/2$ holds for all requests $\req \in \requests$.
\end{proof}
\end{lemma}

By the above lemma, the set $\PotEmbeddingsHat$ will indeed not be empty. As per rescaling $\sum_{\decompHat \in \PotEmbeddingsHat} \probHat = 1$ holds, for each request exactly one decomposition will be selected. Finally, we derive the following lemma.

\begin{lemma}
\label{lem:deterministic-objective-guarantee-without}
The cost~$\sum_{\req \in \requests} c(\hat{m}_{\req})$ of any solution returned by Algorithm~\ref{alg:approxWithoutAdmissionControl} is upper bounded by two times the optimal cost.
\begin{proof}
Let $\optLP$ denote the cost of the optimal linear solution to the linear relaxation of Integer Program~\ref{alg:SCEP-IP_without} and let $\optIP$ denote the cost of the respective optimal integer solution. By allowing to select only decompositions $\map$ for which $c(\map)~\leq 2\cdot \WAC$ holds in  Algorithm~\ref{alg:approxWithoutAdmissionControl} and as for each request a single mapping is chosen, we have $c(\hat{m}_{\req})~\leq 2\cdot \WAC$, where $\hat{m}_{\req}$ refers to the actually selected mapping. Hence $\sum_{\req \in \requests} c(\hat{m}_{\req})~\leq 2\cdot \sum_{\req \in \requests} \WAC$ holds. By Corollary~\ref{cor:relation-of-net-profit-in-decomposition-and-the-LP-cost-new} and the fact that $\optLP \leq \optIP$ holds (cf. Fact~\ref{obs:IP-solutions-are-a-subset-of-LP-solutions}), we can conclude that $\sum_{\req \in \requests} c(\hat{m}_{\req})~\leq 2 \cdot \optIP$ holds, proving the lemma.
\end{proof}
\end{lemma}

Note that the above upper bound on the cost of any produced solution is deterministic and does not depend on the random choices made in Algorithm~\ref{alg:approxWithoutAdmissionControl}.

\subsection{Probabilistic Guarantees for Capacity Violations}

Algorithm~\ref{alg:approxWithoutAdmissionControl} employs the same rounding procedure as Algorithm~\ref{alg:approxAdmissionControl} presented in Section~\ref{sec:randround} for computing approximations for SCEP-P. Indeed, the only major change with respect to Algorithm~\ref{alg:approxAdmissionControl} is the preprocessing that replaces the fractional embedding values $\prob$ with $\probHat$. As the changed values $\probHat$ are used for probabilistically selecting the mapping for each of the requests, the analysis of the capacity violations needs to reflect these changes. To this end, Lemma~\ref{lem:wac-pruning} will be instrumental as it shows that each of the decompositions is scaled by at most a factor of two. As the following lemma shows, this implies that the decompositions contained in $\PotEmbeddingsHat$ are using at most two times the original capacities of nodes and functions.

\begin{lemma}
\label{lem:relation-of-loads-to-capacities-without-admission}
The fractional decompositions of all requests violate node function and edge capacities by at most a factor of two, i.e. 
\begin{alignat}{3}
\sum_{\req \in \requests} \sum_{\decompHat \in \PotEmbeddingsHat} \probHat \cdot \load(x,y)~\leq 2\cdot \Scap(x,y)
\label{eq:lem:scaling-allocations}
\end{alignat}
holds for all resources $(x,y) \in  \SR$.
\begin{proof}
By Lemma~\ref{lem:allocations}, the fractional allocations of the original decompositions $\PotEmbeddings$ do not violate capacities. Hence, by multiplying Equation~\ref{eq:lem:allocations} by two, we obtain that 
\begin{align}
2\cdot \sum_{\req \in \requests} \sum_{D^k_{\req} \in \mathcal{D}_\req} f^k_\req \leq 2\cdot \Scap(x,y)
\label{eq:lem:scaling-allocations-2}
\end{align}
holds for all resources $(x,y) \in  \SRV$. By Lemma~\ref{lem:wac-pruning}, $\lambda_{\req} \geq 1/2$ holds for all requests $\req \in \requests$. Hence, the probabilities of the decompositions having costs less than $\WAC$ are scaled by at most a factor of two, i.e. $\probHat \leq 2\cdot \prob$ holds for all requests $\req \in \requests$ and each decomposition. As the mappings and the respective loads of $\PotEmbeddings$ and $\PotEmbeddingsHat$ are the same, we  obtain that
\begin{align*}
\sum_{\req \in \requests} \sum_{\decompHat \in \PotEmbeddingsHat} \probHat \cdot \load(x,y)~\leq 2\cdot \sum_{\req \in \requests} \sum_{\decomp \in \PotEmbeddings} \prob \cdot \load(x,y)
\end{align*} 
holds for all resources $(x,y) \in  \SRV$. Together with Equation~\ref{eq:lem:scaling-allocations-2} this proves the lemma.
\end{proof}
\end{lemma}

Given the above lemma, we restate most of the lemmas already contained in Section~\ref{sec:performance-guarantee-cap-violation-admission-control} with only minor changes. We note that the definition of the maximal loads (cf. Definition~\ref{def:maximal-loads})~are independent of the decompositions found and that the corresponding assumptions made in Lemma~\ref{lem:assumptions} are still valid when using the relaxation of Integer Program~\ref{alg:SCEP-IP_without} and are independent of the scaling. Concretely, forbidding mappings of network functions or edges onto elements that can never support them is still feasible. Furthermore note that while we effectively scale the probabilities of choosing specific decomposition, this does not change the corresponding loads of the (identical)~mappings. 

We again model the load on the different network functions $(\type,u) \in  \SRV$ and the substrate edges $(u,v) \in  \SE$ induced by each of the requests $\req \in \requests$ as random variables $L_{\req,\type,u} \in [0,\maxLoadVSum]$ and $L_{\req,u,v} \in [0,\maxLoadE \cdot |\VE|]$ respectively. To this end, we note that Algorithm~\ref{alg:approxWithoutAdmissionControl} chooses the $k$-th decomposition $\decompHat = (\probHat, \load,\mapping)$ with probability $\probHat$. Hence, with probability $\probHat$ the load $\load$ is induced for request $\req \in \requests$. The respective variables can therefore be defined as $
\mathbb{P}(L_{\req,x,y} = \load(x,y))= \probHat $ (assuming pairwise different loads)~and $\mathbb{P}(L_{\req,x,y} = 0)=  1 - \sum_{\decompHat \in \PotEmbeddingsHat } \probHat $ for $(x,y) \in  \SRV$.
Again, we denote by $L_{x,y} = \sum_{\req \in \requests} L_{\req,x,y}$ the overall load induced on resource $(x,y) \in  \SR$.

By definition of the expectation, the expected load on the network resource $(x,y) \in  \SRV$ computes to
\begin{alignat}{3}
\mathbb{E}(L_{x,y})~&& = & \sum_{\req \in \requests} \sum_{\decompHat \in \PotEmbeddingsHat} \probHat \cdot \load(x,y)~ \,.
\end{alignat}
Together with Lemma~\ref{lem:relation-of-loads-to-capacities-without-admission} we obtain:
\begin{alignat}{5}
\mathbb{E}(L_{x,y})~&& \leq & 2\cdot \Scap(x,y)~& \quad & \textnormal{for}~ (x,y) \in  \SRV \label{eq:load-vs-capacity-without}
\end{alignat}
Using the above, we again apply Hoeffding's Inequality, while slightly adapting the probabilistic bounds compared with  Lemma~\ref{lem:approximation-single-node}.
\begin{lemma}
\label{lem:approximation-single-node-without}
Let~$\DeltaV = \sum_{\req \in \requests}~(\maxLoadVSum / \maxLoadV)^2$. The probability that the capacity of a single network function~$\type \in \types$ on node~$u \in \SVTypes$ is exceeded by more than a factor~$(2+\sqrt{\log~(|\SV|\cdot|\types|)~\cdot \DeltaV})$ is upper bounded by~$(|\SV| \cdot |\types|)^{-2}$.
\begin{proof}
Each variable~$\randVarLNode$ is clearly contained in the interval~$[0, \maxLoadVSum]$. We choose~$t = \sqrt{\log~(|\SV|\cdot|\types|)~\cdot \DeltaV  } \cdot \Scap(\type,u)$ and obtain:
{
\noindent
\begin{alignat*}{2}
& \mathbb{P} \Big(L_{\type,u} - \mathbb{E}(L_{\type,u})~\geq t \Big)~&&  \notag \\
& ~~~ \leq exp \Bigg(\frac{-2\cdot t^2}{\sum_{\req \in \requests}~(\maxLoadVSum)^2} \Bigg)~&& \\
& ~~~ \leq exp \Bigg(\frac{-2 \log~(|\SV|\cdot|\types|)~\cdot \DeltaV \cdot  \Scap(\type,u)^2}{\sum_{\req \in \requests}~(\Scap(\type,u)\cdot \maxLoadVSum / \maxLoadV)^2} \Bigg)~&& \numberthis \label{formula:proof:singleNodeViolation:upperCapacities-without}\\
& ~~~ = exp \Bigg(\frac{-2 \log~(|\SV|\cdot|\types|)~\sum_{\req \in \requests}~(\maxLoadVSum / \maxLoadV)^2 }{\sum_{\req \in \requests}~( \maxLoadVSum / \maxLoadV)^2} \Bigg)~&& \\
& ~~~ = exp \Bigg(\frac{-2 \log~(|\SV|\cdot|\types|)~\sum_{\req \in \requests}~(\maxLoadVSum / \maxLoadV)^2 }{\sum_{\req \in \requests}~( \maxLoadVSum / \maxLoadV)^2} \Bigg)~&& \\
& ~~~ =~(|\SV|\cdot |\types|)^{-2} && 
\end{alignat*}
}
Analogously to Lemma~\ref{lem:approximation-single-node}, we have used again
\begin{align}
\maxLoadVSum \leq \Scap(\type,u)~\cdot \maxLoadVSum / \maxLoadV \label{formula:proof:singleNodeViolation:estimation-without}
\end{align}
in Equation~\ref{formula:proof:singleNodeViolation:upperCapacities-without}. Analogously to the proof of Lemma~\ref{lem:approximation-single-node}, we have $\mathbb{E}(L_{\type,u})~\leq 2 \cdot \Scap(\type, u)$ by Equation~\ref{eq:load-vs-capacity-without} and the lemma follows.
\end{proof}
\end{lemma}

We restate Lemma~\ref{lem:approximation-single-edge} without proof as the only change is replacing the factor of  $\sqrt{2}$ with the factor $\sqrt{3/2}$ and
observing that by using Equation~\ref{eq:load-vs-capacity-without} we obtain a $2+\gamma$ approximation for the load instead of a $1+\gamma$ one.

\begin{lemma}
\label{lem:approximation-single-edge-without}
Let~$\DeltaE = \sum_{\req \in \requests} |\VE|^2$. The probability that the capacity of a single substrate edge~$(u,v) \in  \SE$ is exceeds the capacity by more than a factor~$(2+\sqrt{3/2\cdot\log~(|\SV|)~\cdot \DeltaE})$ is bounded by~$|\SV|^{-2}$.
\end{lemma}

The following corollaries show that the above bounds play out nicely if 
we assume more strict bounds on the maximal loads.

\begin{corollary}
Assume that~$\maxLoadV \leq \varepsilon \cdot \Scap(\type, u)$ holds for~$0 < \varepsilon < 1$ and all~$\type \in \types, u \in \SVTypes$. With~$\DeltaV$ as defined in Lemma~\ref{lem:approximation-single-node}, we obtain:
The probability that in Algorithm~\ref{alg:approxAdmissionControl} the capacity of a single network function~$\type \in \types$ on node~$u \in \SVTypes$ is exceeded by more than a factor~$(2+ \varepsilon \cdot \sqrt{\log~(|\SV|\cdot|\types|)~\cdot \DeltaV})$ is upper bounded by~$(|\SV| \cdot |\types|)^{-2}$.
\end{corollary}

\begin{corollary}
Assume that~$\maxLoadE \leq \varepsilon \cdot \Scap(u,v)$ holds for~$0 < \varepsilon < 1$ and all~$(u,v) \in  \SE$. With~$\DeltaE$ as defined in Lemma~\ref{lem:approximation-single-edge}, we obtain:
The probability that in Algorithm~\ref{alg:approxAdmissionControl} the capacity of a single substrate node~$(u,v) \in  \SE$ is exceeded by more than a factor~$(2+\varepsilon \cdot \sqrt{3/2\cdot \log~(|\SV|)~\cdot \DeltaE})$ is upper bounded by~$|\SV|^{-2}$.
\end{corollary}

\subsection{Main Results}
\label{sec:main-results-without-admission-control}

We can now proceed to state the main tri-criteria approximation results obtained for SCEP-C. The argument for Algorithm~\ref{alg:approxWithoutAdmissionControl} having a polynomial runtime is the same as for Algorithm~\ref{alg:approxAdmissionControl}, since the preprocessing can be implemented in polynomial time. 

The following lemma shows that Algorithm~\ref{alg:approxAdmissionControl} can produce solutions of high quality \emph{with high probability}.

\begin{lemma}
Let~$0 < \varepsilon \leq 1$ be chosen minimally, such that~$\maxLoadV[\req][x][y] \leq \varepsilon \cdot \Scap(\type,u)$ holds for all network resources~$(x,y) \in  \SR$.  Setting $\beta = \varepsilon \cdot \sqrt{\log~(|\SV|\cdot|\types|)~\cdot \DeltaV  }$ and~$\gamma = \varepsilon \cdot \sqrt{3/2\cdot \log |\SV|)~\cdot \DeltaE  }$ with~$\DeltaV,\DeltaE$ as defined in Lemmas~\ref{lem:approximation-single-node-without} and \ref{lem:approximation-single-edge-without}, the probability that a solution is found within~$Q \in \mathbb{N}$ rounds is lower bounded by~$1 -(2/3)^Q$ if $|\SV| \geq 3$ holds.
\begin{proof}
We again apply a union bound argument. By Lemma~\ref{lem:approximation-single-node-without} the probability that for a single network function of type~$\type \in \types$ on node~$u \in \SVTypes$ the allocations exceed $(2+\beta)~\cdot \Scap(\type,u)$ is less than~$(|\SV| \cdot |\types|)^{-2}$. Given that there are maximally~$|\SV| \cdot |\types|$ many network functions overall, the probability that on any network function more than~$(2+\beta)\cdot \Scap(\type,u)$ resources will be used is less than~$(|\SV| \cdot |\types|)^{-1} \leq |\SV|^{-1}$. Similarly, by Lemma~\ref{lem:approximation-single-edge-without} the probability that the allocation on a specific edge is larger than $(2+\beta)\cdot \Scap(u,v)$ is less than~$|\SV|^{-3}$. As there are at most~$|\SV|^2$ edges, the union bound gives us that the probability that any of these edges' allocations will be higher than $(2+\gamma)\cdot \Scap(u,v)$ is less than $|\SV|^{-1}$. As the cost of any solution found is deterministically always smaller than two times the optimal cost (cf. Lemma~\ref{lem:deterministic-objective-guarantee-without}), the probability of \emph{not} finding an appropriate solution within a single round is upper bounded by $2/|\SV|$. For $|\SV| \geq 3$ the probability of finding a feasible solution within $Q \in \mathbb{N}$ rounds is therefore $1-(2/3)^Q$.
\end{proof}
\end{lemma}

Finally, we can state the main theorem showing that Algorithm~\ref{alg:approxWithoutAdmissionControl} is indeed a tri-criteria approximation algorithm.

\begin{theorem}
\label{thm:result-for-admission-control-without}
Assuming that $|\SV| \geq 3$ holds and that $\maxLoadV[\req][x][y] \leq \varepsilon \cdot \Scap(\type,u)$ holds for all network resources~$(x,y) \in  \SR$ with $0 < \varepsilon \leq 1$ and by setting~$\beta = \varepsilon \cdot \sqrt{\log~(|\SV|\cdot|\types|)~\cdot \DeltaV  }$ and~$\gamma = \varepsilon \cdot \sqrt{3/2\cdot \log |\SV| \cdot \DeltaE  }$ with~$\DeltaV,\DeltaE$ as defined in Lemmas~\ref{lem:approximation-single-node} and \ref{lem:approximation-single-edge}, Algorithm~\ref{alg:approxAdmissionControl} is a~$(\alpha,2+\beta,2+\gamma)$ tri-criteria approximation algorithm for SCEP-C, such that it finds a solution \emph{with high probability}, with costs less than~$\alpha = 2$ times higher than the optimal cost and violates network function and edge capacities only within the factors~$2+\beta$ and~$2+\gamma$ respectively.
\end{theorem}

\section{Approximate Cactus Graph Embeddings}\label{sec:cactus}

Having discussed approximations for \emph{linear} service chains, we now turn towards more general service \emph{graphs} (essentially a \emph{virtual network}), i.e. service specifications that may contain cycles or branch separate sub-chains. 
Concretely, we propose a novel linear programming formulation 
in conjunction with a novel decomposition algorithm for service 
graphs whose undirected interpretation is a cactus graph. 
Given the ability to decompose fractional solutions, 
we show that we can still apply the results of Sections~\ref{sec:randround} and \ref{sec:approximation-without-admission-control} for this case. Our results show that our approximation scheme can be applied \emph{as long as linear solutions can be appropriately decomposed}. We highlight the advantage of our novel formulation by showing that the standard multi-commodity flow approach employed in the Virtual Network Embedding Problem
(VNEP)~literature \emph{cannot} be decomposed and hence cannot be used in our approximation framework.

This section is structured as follows. In Section~\ref{sec:cactus-motivation} we motivate why considering more complex service \emph{graphs} is of importance. Section~\ref{sec:cactus:service-cactus-graph-embedding-problems} introduces the notion of \emph{service cactus graphs} and introduces the respective generalizations of SCEP-P and SCEP-C. Section~\ref{sec:cactus:decomposing-service-graphs} shows how these particular service graphs can be decomposed into subgraphs. Building on this a priori decomposition of the service graphs, we introduce extended graphs for cactus graphs and the respective integer programming formulation in Sections~\ref{sec:cactus:extended-graphs} and \ref{sec:cactus-linear-programming}. In Section~\ref{sec:cactus:decomposing-service-graphs} we show how the linear solutions can be decomposed into a set of fractional embeddings analogously to the decompositions computed in Section~\ref{sec:decompo}. Section~\ref{sec:cactus:approximation-service-cacti-embeddings} shows that the approximation results for SCEP-P and SCEP-C still hold for this general graph class. Lastly, Section~\ref{sec:cactus:non-decomposability} shows that the standard approach of using multi-commodity flow formulations yields non-decomposable solutions, i.e. our approximation framework cannot be applied when using the standard approach. This also sheds light on the question why no approximations are known for the Virtual Network Embedding Problem, which considers the embedding of general graphs.

\subsection{Motivation \& Use Cases}
\label{sec:cactus-motivation}

While service chains were originally understood as \emph{linear sequences
of service functions} (cf.~\cite{merlin}), 
we witness a trend toward more 
complex chaining models, where a single network function may spawn multiple flows towards other functions, or merge multiple incoming flows. We will discuss one of these use cases in detail and refer the reader to 
\cite{etsi2013network, gember2014opennf,  ietf-draft-sc-use-cases-mobile-networks,ietf-draft-sc-use-cases-dc,mehraghdam2014specifying} for an overview on more complex chaining models.

\begin{figure}[tbhp]
 \includegraphics[width=1\columnwidth]{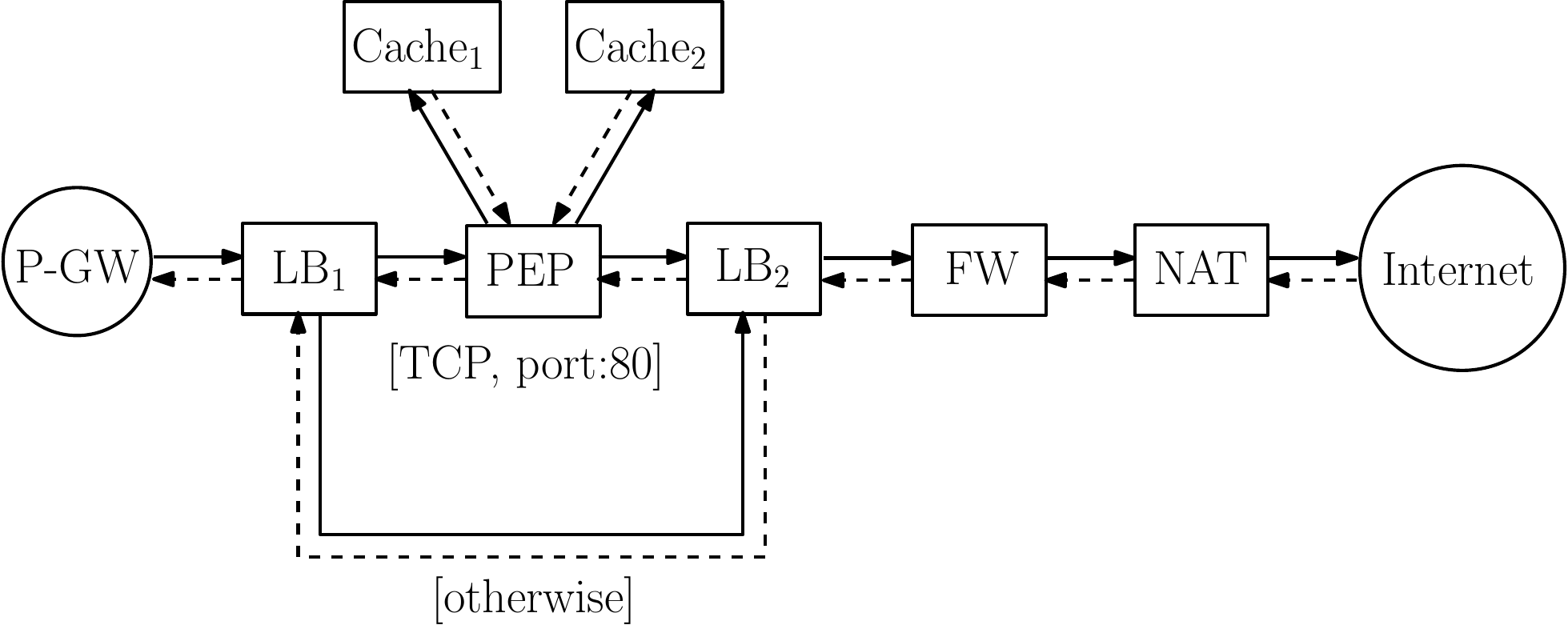}
 \caption{Actual service chain example for HTTP  optimization taken from~\cite{ietf-draft-sc-use-cases-mobile-networks} with up- (\emph{solid})~and downstream communications (\emph{dashed}). The packet gateway (P-GW)~terminates the mobile (3GPP)~network and forwards all traffic to a load balancer (LB)~which splits the traffic flows: TCP traffic on port 80 is forwarded to the performance enhancement proxy (PEP)~which connects to two (load balanced)~caches. The load balancer $\textnormal{LB}_2$ merges the outgoing traffic flows and forwards them through a firewall (FW)~and a network address translator (NAT)~and finally to the Internet.
 Depending on the ratio on the amount of web-traffic, the different up- and downstream connections will have different bandwidth requirements. Furthermore, e.g., video streams received by the PEP from one of the caches may be transcoded on-the-fly, and hence the outgoing bandwidth of the PEP towards $\textnormal{LB}_1$ might be less than the traffic received.
 }
 \label{fig:service-chain-example}
 \end{figure}
 
 Let us give an example which includes  
functionality for load balancing, flow splitting and merging.
It also makes the case for ``cyclic'' service chains. 
The use case is situated in the context of 
LTE networks and Quality-of-Experience (QoE)~
for mobile users. The service chain is depicted in Figure~\ref{fig:service-chain-example}
and is discussed in-depth in the IETF draft~\cite{ietf-draft-sc-use-cases-mobile-networks}.
Depending on whether the incoming traffic from the packet gateway (P-GW)~at the first load balancer $\textnormal{LB}_1$ is destined for port 80 and has type TCP,
it is forwarded through a performance enhancement proxy (PEP). Otherwise, $\textnormal{LB}_1$ forwards the flow directly to a second load balancer $\textnormal{LB}_2$ which merges all outgoing traffic destined to the Internet. Depending on the deep-packet inspection 
performed by the PEP (not depicted in Figure~\ref{fig:service-chain-example}, cf.~\cite{ietf-draft-sc-use-cases-mobile-networks})~and whether the content is cached, the PEP may redirect the traffic towards one of the caches. Receiving the content either from the caches or the Internet, the PEP may -- depending on user agent information -- additionally perform transcoding if the content is a video. This allows to offer both a higher quality of experience for the end-user and reduces network utilization.

It must be noted that all depicted network functions (depicted as \emph{rectangles})~are stateful. The load balancer $\textnormal{LB}_2$ e.g., needs to know whether the traffic received from the firewall needs to be passed through the PEP (e.g., for caching)~or whether it can be forwarded directly towards $\textnormal{LB}_1$. Similarly, the firewall and the network-address translation need to keep state on incoming and outgoing flows. Note that the PEP spawns sub-chains towards the different caches and that the example also contains multiple types of cycles. Based on the statefulness of the network functions, all connections between network functions are bi-directed. Furthermore, there exists a ``cycle'' $\textnormal{LB}_1, \textnormal{PEP}, \textnormal{LB}_2, \textnormal{LB}_1$ (in the undirected interpretation). Our approach presented henceforth allows to embed either the outbound or the inbound connections of the example in Figure~\ref{fig:service-chain-example} (depicted as dashed or solid), but can be extended to also take bi-directed connections into account.

\subsection{Service Cactus Graphs Embedding Problems}
 
\label{sec:cactus:service-cactus-graph-embedding-problems}
 
Given the above motivation, we will now extend the service chain definition of Section~\ref{sec:model} to more complex service \emph{graphs},
like the one in Figure~\ref{fig:service-chain-example}. While for service chains request graphs~$\VG =~(\VV,\VE)$ were constrained to be lines (cf.~Section~\ref{sec:model}), we relax this constraint in the following definition.
\begin{definition}[Service Cactus Graphs]
\label{def:service-cacti-graphs}
Let~$\VG =~(\VV, \VE)$ denote the service graph of request~$\req \in \requests$. We call~$\VG$ a service cactus graph if the following two conditions hold:
\begin{enumerate}
\item$\VE$ does not contain opposite edges, i.e. for all~$(v,u) \in  \VE$ the opposite edge~$(u,v)$ is not contained in~$\VE$.
\item The undirected interpretation~$\VGbar =~(\VVbar,\VEbar)$, with~$\VVbar = \VV$ and~$\VEbar = \{\{u,v\} |~(u,v) \in  \VE\}$, is a cactus graph, i.e. any two simple cycles share at most a single node.
\end{enumerate} 
\end{definition}

Moreover, we do not assume that the service cactus graphs have unique start or end nodes. Specifically, all virtual nodes~$i \in \VV$ can have arbitrary types. The definitions of the capacity function~$\Vcap : \VV \cup \VE \to \preals$ and the cost functions~$\Scost: \SR \to \preals$ are not changed. As the definitions of valid and feasible embeddings as well as the definition of SCEP-P and SCEP-C~(see Definitions~\ref{def:valid-mapping}, \ref{def:feasible-embedding}, \ref{def:scep-with-admission-control}, and \ref{def:scep-without-admission-control} respectively)~do not depend on the underlying graph model, these are still valid, when considering service cactus graphs as input. To avoid ambiguities, we refer to the respective embedding 
problems on cactus graphs as SCGEP-P and SCGEP-C.

\subsection{Decomposing Service Cactus Graphs}
\label{sec:cactus:decomposing-service-graphs}

As the Integer Programming formulation~(see IP~\ref{alg:SCGEP-IP} in Section~\ref{sec:cactus-linear-programming})~for service cactus graphs is much more complex, we first introduce a \emph{graph decomposition} for cactus graphs to enable the concise descriptions of the subsequent algorithms.

Concretely, we apply a breadth-first search to re-orient edges as follows. For a service cactus graph $\VG =~(\VV,\VE)$, we choose \emph{any} node as the root node and denote it by~$\VVroot$. Now, we perform a simple breadth-first search in the undirected service graph~$\VGbar$~(see Definition~\ref{def:service-cacti-graphs})~and denote by~$\VVpred : \VV \to \VV \cup \{\NULL\}$ the predecessor function, such that if~$\VVpred(i)~= j$  then virtual node~$i\in \VV$ was explored first from virtual node~$j \in \VV$. Let~$(\VVbfs, \VEbfs)$ denote the graph~$\VGbfs$ with~$\VVbfs = \VV$ and~$\VEbfs = \{(i,j)~| j \in \VV,  \VVpred(j)~= i \wedge \VVpred(j)~\neq \NULL \}$. Clearly,~$\VGbfs$ is a directed acyclic graph. We note the following based on the cactus graph nature of~$\VG$:
\begin{lemma}~\\
\vspace{-12pt}
\label{lem:observations-bfs-graphs}
\begin{enumerate}
\item The in-degree of any node is bounded by 2, i.e. $|\delta_{\VEbfs}^-(i)| \leq 2$ holds for all $i \in \VVbfs$.
\item Furthermore, if~$|\delta_{\VEbfs}^-(i)| = 2$ holds for some~$i \in \VVbfs$, then there exists a node~$j \in \VVbfs$, such that there exist exactly two paths~$P_\req,P_l$ in~$\VVbfs$ from~$j$ to~$i$ which only coincide in~$i$ and~$j$.
\end{enumerate}
\begin{proof}
With respect to Statement 1)~note that node~$i$ must be reached from the (arbitrarily)~chosen root~$\VVroot \in \VVbfs$. For the sake of contradiction, if~$|\delta_{\VEbfs}^-(i)| > 2$ then there exist at least three~(pair-wise different)~paths~$P_1,P_2,P_3$ from~$\VVroot$ to~$i$. Hence, in the undirected representation~$\VGbar$ there must exist at least two cycles overlapping in more than one node, namely in~$i$ as well as some common predecessor of~$i$. This is not allowed by the cactus nature of~$\VGbar$ (cf. Definition~\ref{def:service-cacti-graphs})~and hence~$|\delta_{\VEbfs}^-(i)| \leq 2$ must hold for any~$i \in \VV$.

Statement 2)~holds due to the following observations. As all nodes are reachable from the root~$\VVroot$, we are guaranteed to find a common predecessor~$j \in \VVbfs$ while backtracking along the reverse directions of edges in~$\VEbfs$. If this was not the case, then there would exist two different sources, which is -- by construction -- not possible.
\end{proof}
\end{lemma}

Based on the above lemma, we will now present a specific graph decomposition for service cactus graphs. Concretely, we show that service cactus graphs can be decomposed into a set of cyclic subgraphs together with a set of line subgraphs with unique source and sink nodes. For each of these subgraphs, we will define an extended graph construction that will be used in our Integer Programming formulation. Concretely, the graph decomposition given below, will enable a structured induction of flow in the extended graphs: the flow reaching a sink in one of the subgraphs will induce flow in all subgraphs having this node as source. This will also enable the efficient decomposition of the computed flows (cf. Section~\ref{sec:cactus:decomposition-algorithm-for-service-cacti-and-correctness}).

\begin{definition}[Service Cactus Graph Decomposition]
\label{def:service-cactus-graph-decomposition}
Any graph~$\VGbfs$ of a service cactus graph~$\VG$ can be decomposed into a set of cycles~$\Cycles = \{C_1, C_2, \dots \}$ and a set of simple paths~$\Paths = \{P_1, P_2, \dots\}$ with corresponding subgraphs $\VGcycle =~(\VVcycle, \VEcycle)$ and~$\VGpath =~(\VVpath, \VEpath)$ for $C_k \in \Cycles$ and $P_k \in \Paths$, such that:
\begin{enumerate}
\item The graphs $\VGcycle$ and~$\VGpath$ are \emph{connected} subgraphs of~$\VGbfs$ for $C_k \in \Cycles$ and $P_k \in \Paths$ respectively.
\item Sources~$\VVcycleSource, \VVcycleTarget \in \VVcycle$ and sinks~$\VVpathSource, \VVpathTarget \in \VVpath$ are given for~$C_k \in \Cycles$ and~$P_k \in \Paths$ respectively, such that~$\delta_{\VEcycle}^-(\VVcycleSource)~=  \delta_{\VEcycle}^+(\VVcycleTarget)~= \delta_{\VEpath}^-(\VVpathSource)~= \delta_{\VEpath}^+(\VVpathTarget)~= \emptyset$ holds within~$\VGcycle$ and~$\VGpath$ respectively.
\item $\{\VEcycle | C_k \in \Cycles\} \cup \{\VEpath | P_k \in \Paths\}$ is a~(pair-wise disjoint)~partition of~$\VEbfs$.
\item Each~$C_k \in \Cycles$ consists of exactly two branches~$\VEcycleBranchL, \VEcycleBranchR \subset \VEcycle$, such that both these branches start at~$\VVcycleSource$ and terminate at~$\VVcycleTarget$ and cover all edges of~$\VGcycle$.
\item For~$P_k \in \Paths$ the graph~$\VGpath$is a simple path from~$\VVpathSource$ to~$\VVpathTarget$.
\item Paths may only overlap at sources and sinks:~$\forall P_k, P_{k'} \in \Paths, P_k \neq P_{k'}: \VVpath \cap \VVpath[\req][k'] \subset \{\VVpathSource, \VVpathTarget, \VVpathSource[\req][k'], \VVpathTarget[\req][k']\}$.
\end{enumerate}
We introduce the following sets:
{
\abovedisplayskip=0pt
\vskip0cm
\belowdisplayskip=8pt
\begin{alignat*}{6}
\VVcycleSources &=&&~ \{\VVcycleSource| C_k \in \Cycles\} &\quad&&~
\VVpathSources &=&&~ \{\VVpathSource | P_k \in \Paths\} \\
\VVcycleTargets &=&&~ \{\VVcycleTarget | C_k \in \Cycles\} 
&\quad&&~
\VVpathTargets &=&&~ \{\VVpathTarget | P_k \in \Paths\} \\
\VVcycleSourcesTargets &=&&~ \{\VVcycleSource, \VVcycleTarget | C_k \in \Cycles\} 
&\quad&&~
 \VVpathSourcesTargets &=&&~ \{\VVpathSource, \VVpathTarget | P_k \in \Paths\}
\end{alignat*}
\begin{alignat*}{3}
\VVSourcesTargets &= &&~ \VVcycleSourcesTargets \cup \VVpathSourcesTargets
\end{alignat*}
\begin{alignat*}{6}
\VEcycles &=&&~ \{ e \in \VEcycle | C_k \in \Cycles\} 
&\quad&&~
\VEpaths  &=&&~ \{e \in \VEpath | P_k \in \Paths\} \\
\VEcycleSame &=&&~ \VE \cap \VEcycle
&\quad&&~
\VEcycleDiff &=&&~ \VEcycle \setminus \VEcycleSame \\
\VEpathSame &=&&~ \VE \cap \VEpath
&\quad&&~
\VEpathDiff &=&&~ \VEpath \setminus \VEpathSame
\end{alignat*}\begin{alignat*}{3}
\VEDiff &=&&~ \{\VEcycleDiff | C_k \in \Cycles\} \cup \{\VEpathDiff | P_k \in \Paths\} 
\end{alignat*}
\begin{alignat*}{3}
\VVbranchingcycle &=&&~ \{i \in \VVcycle | |\delta_{\VEbfs}^+(i)| > 1, i \notin \VVcycleSourcesTargets \}
\end{alignat*}
}
Note that the node sets $\VVcycleSources$, $\VVpathSources$, $\VVcycleTargets$, $\VVpathTargets$, $\VVcycleSourcesTargets$, and $\VVpathSourcesTargets$ only contain virtual nodes of paths or cycles that are either the source, or the target or any of both.
Similarly, $\VEcycles$ and $\VEpaths$ contain all virtual edges that are covered by any of the paths or any of the cycles. With respect to $\VEcycleSame$ and $\VEpathSame$ we note that these edges' orientation agrees with the original specification $\VE$ and the edges in $\VEcycleDiff$ and $\VEpathDiff$ are reversed. $\VEDiff$ contains all edges whose orientation was reversed in $\VGbfs$. $\VVbranchingcycle$ denotes the set of \emph{branching nodes} of cycle $C_k \in \Cycles$, i.e. nodes which have an out-degree of larger than one in the graph $\VGbfs$ and which are not a source or a target of any of the cycles.

Lastly, we abbreviate $\SVTypes[\Vtype(\VVcycleTarget)]$ by $\SVTypesCycle$, i.e. the substrate nodes onto which the target $\VVcycleTarget$ of cycle $C_k$ can be mapped.
\end{definition}

The constructive existence of the above decomposition follows from Lemma~\ref{lem:observations-bfs-graphs} as proven in the following lemma.

\begin{lemma}
Given a service cactus graph $\VG = (\VV,\VE)$, the graph $\VGbfs = (\VVbfs,\VEbfs)$ and its decomposition can be constructed in polynomial time.
\begin{proof}
Note that the graph $\VGbfs$ is constructed in polynomial time by choosing an arbitrary node as root and then performing a breadth-first search. Having computed the graph $\VGbfs$, the decomposition of $\VGbfs$ (according to Definition~\ref{def:service-cactus-graph-decomposition})~can also be computed in polynomial time. First identify all `cycles' in~$\VGbfs$ by performing backtracking from nodes having an in-degree of 2. Afterwards, no cycles exist anymore and for each remaining edge~$e =~(i,j) \in  \VEbfs \setminus \{\VEcycle | C_k \in \Cycles \}$ a path~$\VGpath =~(\VVpath, \VEpath)$ with~$\VVpath = \{i,j\}$ and~$\VEpath = \{(i,j)\}$ can be introduced.
\end{proof}
\end{lemma}

An example of a decomposition is shown in Figure~\ref{fig:service-cactus-request}. The request graph $\VG$ is first reoriented by a breadth-first search to obtain $\VGbfs$. As the virtual network function $l$ is the only node with an in-degree of $2$, first the ``cycle'' along $j,k,l$ and $j,m,l$ is decomposed to obtain cycle $C_1$. Afterwards, all remaining edges are decomposed into single paths $\Paths = \{P_1,P_2,P_3,P_4\}$, consisting only of one edge.
With respect to the notation introduced in Definition~\ref{def:service-cactus-graph-decomposition}, we have e.g. ~$\VVpathSource[\req][1] = j$,~$\VVcycleSource[\req][1] = j$,~$\VVpathTarget[\req][1] = i$, and~$\VVcycleTarget[\req][1] = l$. Furthmerore, according to the original edge orientation we have~$\VEpathSame[\req][1]=\emptyset$,~$\VEpathDiff[\req][1] = \{(j,i)\}$,~$\VEcycleSame[\req][1] = \{(j,k),~(k,l)\}$, and~$\VEcycleDiff[\req][1] = \{(j,m),~(m,l)\}$. We note that node $m$ is a branching node, i.e. $\VVbranchingcycle[\req][1] = \{m\}$, as it is contained in the subgraph of $C_1$ and has a degree larger than 1 and is neither a source nor the target of a cycle.

\begin{figure}[t!]
\centering
\includegraphics[width=0.95\columnwidth]{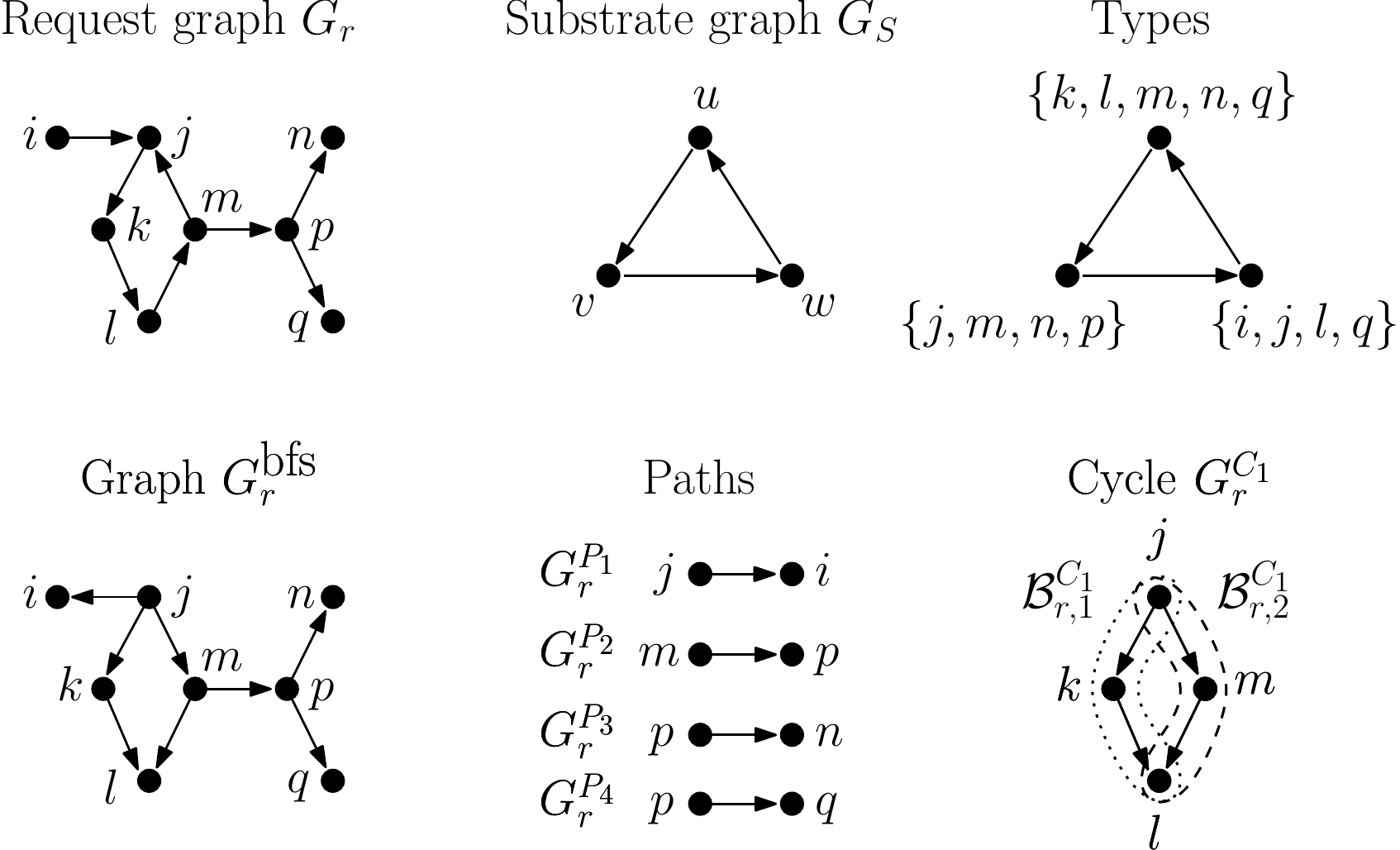}
\caption{Service cactus graph as well as the decomposition of it according to Definition~\ref{def:service-cactus-graph-decomposition}. We visualize the network types by annotating the substrate nodes with the virtual nodes that can be mapped to them and we have e.g. $\SVTypes[\Vtype(j)] = \{v,w\}$.
The decomposition graph $\VGbfs$ is rooted at the -- arbitrarily chosen -- substrate node~$j=\VVroot \in \VVbfs$. The graph $\VGbfs$ is decomposed into four paths and a single cycle. Note that some of the edges in $\VGbfs$ are reversed with respect to the original graph $\VG$ and that these re-orientations also reflect in the edge orientations within the decompositions.
}
\label{fig:service-cactus-request}
\end{figure}

\subsection{Extended Graphs for Service Cacti Graphs}
\label{sec:cactus:extended-graphs}

Based on the above decomposition scheme for service cactus graphs, we now introduce extended graphs for each path and cycle respectively. Effectively, the extended graphs will be used in the Integer Programming formulation for SCGEP as well as for the decomposition algorithm. In contrast to the Definition~\ref{def:extended-graph} in Section~\ref{sec:integer-linear-program}, these extended graphs will not be directly connected by edges, but thoughtfully stitched using additional variables (see Section~\ref{sec:cactus-linear-programming}).

\begin{figure}[t!]
\centering
\includegraphics[width=0.9\columnwidth]{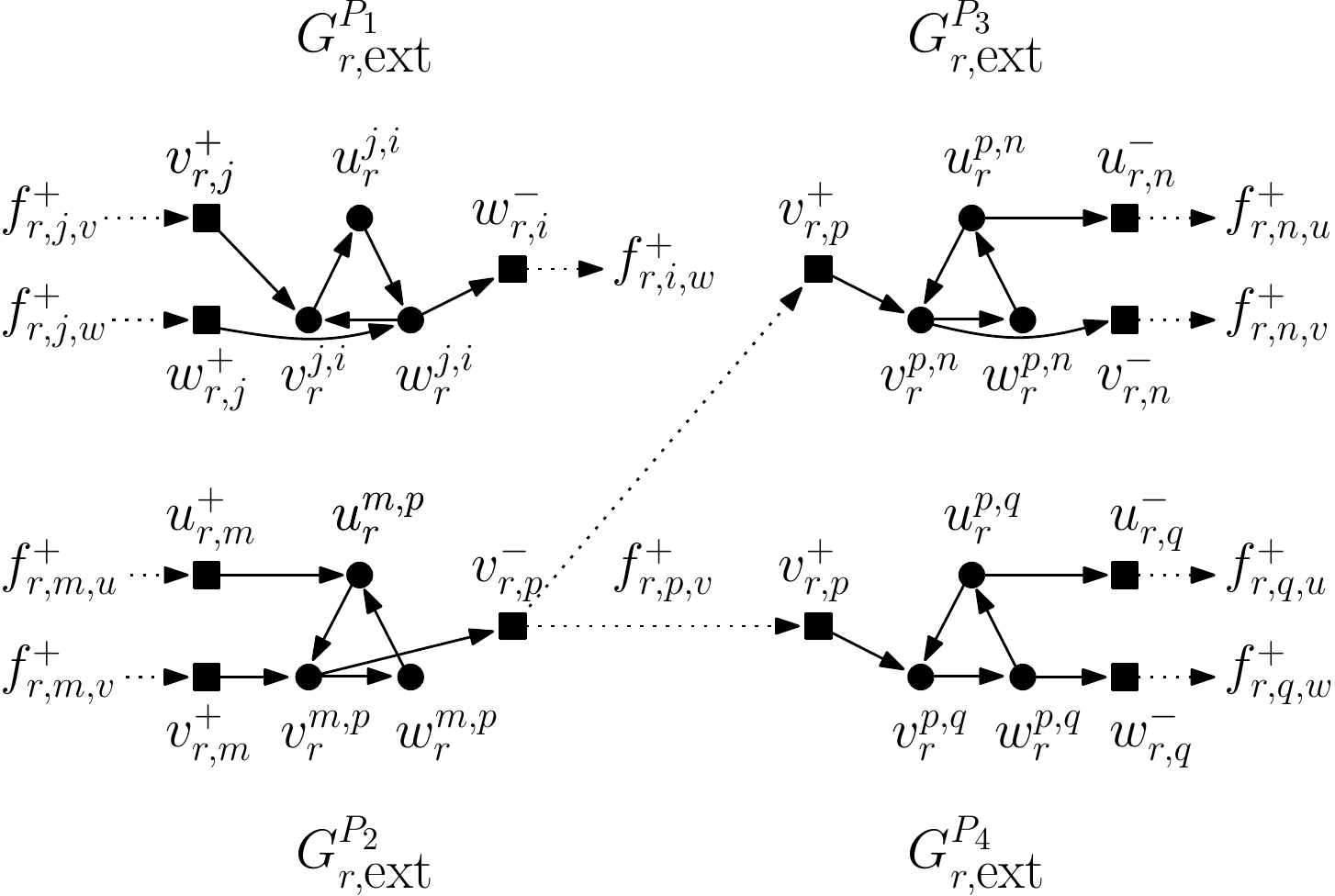}
\caption{Extended graphs~$\VGextpath[\req][i]$ of the paths of service cactus request~$r$ of Figure~\ref{fig:service-cactus-request}. Note that the orientation of edges is reversed, if the virtual edge was reversed in the breadth-first search. This is e.g. the case for the edge $(i,j) \in  \VG$ respectively the path $P_1$.
The flow values~$f^+_{r,\cdot,\cdot}$ depicted here are used in the Integer Program~\ref{alg:SCGEP-IP} to induce flows in the extended graphs.}
\label{fig:service-cactus-construction-path}
\end{figure}

The definition of the extended graphs for paths~$P_k \in \Paths$ generally follows the Definition~\ref{def:extended-graph} for linear service chains: for each path $P_k \in \Paths$ and each virtual edge $(i,j) \in  \VEpath$ a copy of the substrate graph is generated~(cf.~Figure~\ref{fig:service-cactus-construction-path}). If an edge's orientation was reversed when constructing $\VGbfs$, the edges in the substrate copies are reversed. Additionally, each extended graph~$\VGpath$ contains a set of source and sink nodes with respect to the types of~$\VVpathSource$ and~$\VVpathTarget$. Concretely, the extended graph $\VGextpath$ contains two super sources, as the virtual node $j$ can be mapped onto the substrate nodes $v$ and $w$. The concise definition of the extended graph construction for the extended graphs of paths  is given below.

\pagebreak

\begin{definition}[Extended Graph for Paths]
\label{def:service-cactus-extended-graph-path}
The extended graph~$\VGextpath =~(\VVextpath, \VEextpath)$ for path~$P_k \in \Paths$ is defined as follows:
$\VVextpath = \VVextpathSources \cup \VVextpathTargets \cup \VVextpathSubstrate$, with
{
\small
\begin{alignat*}{3}
\VVextpathSources & = &&~ \{ u^+_{r,i} | i = \VVpathSource, u \in \SVTypes[\Vtype(i)]\}\\
\VVextpathTargets & = &&~  \{ u^-_{r,j} | j = \VVpathTarget, u \in \SVTypes[\Vtype(j)]\} \\
\VVextpathSubstrate & = &&~ \{ u^{i,j}_{r} |~(i,j) \in  \VEpath, u \in \SV \}
\end{alignat*}
}
We set~$\VEextpath = \VEextpathSources \cup \VEextpathTargets \cup \VEextpathSubstrate \cup \VEextpathF$, with:
{
\small
\begin{alignat*}{3}
\VEextpathSubstrate && = & \{~(u^{i,j}_{\req}, v^{i,j}_{\req})~|~(i,j) \in  \VEpathSame,~(u,v) \in  \SE \} \cup \\
 && & \{~(v^{i,j}_{\req}, u^{i,j}_{\req})~|~(i,j) \in  \VEpathDiff,~(u,v) \in  \SE \} \\
\VEextpathSources && = & \{~(u^+_{\req,i}, u^{i,j}_{\req})~| i = \VVpathSource,~(i,j) \in  \VEpath, u \in \SVTypes[\Vtype(i)]\} \\
\VEextpathTargets &&= &\{~(u^{i,j}_{\req}, u^-_{\req,j})~| j = \VVpathTarget,~(i,j) \in  \VEpath, u \in \SVTypes[\Vtype(j)]\}  \\
\VEextpathF && = & \{~(u^{i,j}_{\req}, u^{j,l}_{\req})~|~(i,j),~(j,l) \in  \VEpath, u \in \SVTypes[\Vtype(j)] \}
\end{alignat*}
}
\end{definition}

For cycles we employ a more complex graph construction (cf. Definition~\ref{def:service-cactus-extended-graph-cycle}), that is exemplarily depicted in Figure~\ref{fig:service-cactus-construction-cycle}. As noted in Definition~\ref{def:service-cactus-graph-decomposition}, the edges $\VVextcycle$ of cycle $C_k \in \Cycles$ are partitioned into two sets, namely the branches $\VEcycleBranchR$ and $\VEcycleBranchL$. Within the extended graph construction, these branches are transformed into \emph{a set of parallel paths}, namely one for each potential substrate node that can host the function of target function $\VVcycleTarget$. Concretely, in Figure~\ref{fig:service-cactus-construction-cycle}, two parallel \emph{path constructions} are employed for both realizing the branches $\VEcycleBranchR$ and $\VEcycleBranchL$, such that for the left construction the network function $l\in \VV$ will be mapped to $v \in \SV$ and in the right construction the function $l \in \VV$ will be hosted on the substrate node $w \in \SV$. This construction will effectively allow the decomposition of flows, as the amount 
of flow that will be sent into the left path construction of branch $\VEcycleBranchR$ will be required to be equal to the amount of flow that is sent into the left path construction of branch $\VEcycleBranchL$. Furthermore, for each substrate node $u \in \SVTypes[\Vtype(\VVcycleSource)]$ that may host the source network function $\VVcycleSource$ of cycle $C_k \in \Cycles$, there exists a single super source $u^+_{\req,\VVcycleSource}$. Together with the parallel paths construction, the flow along the edges from the super sources to the respective first layers will effectively determine how much flow is forwarded from each substrate node $u \in \SVTypes[\Vtype(\VVcycleSource)]$ hosting the source functionality towards each substrate node $v \in \SVTypes[\Vtype(\VVcycleTarget)]$ hosting the target functionality $\VVcycleTarget$. The amount of flow along edge $(v^+_{\req,j}, v^{j,k}_{r,w})$ in Figure~\ref{fig:service-cactus-construction-cycle} will e.g. indicate to which extent the mapping of virtual node $j \in \VV$ to substrate node $v \in \SV$  will coincide with the mapping of virtual function $l \in \VV$ to substrate node $w \in \SV$. In fact, to be able to decompose the linear solutions later on, we will enforce the equality of flow along edges  $(v^+_{\req,j}, v^{j,k}_{r,w})$ and $(v^+_{\req,j}, v^{j,m}_{r,w})$. We lastly note that the flow variables $f^+_{\req,\cdot,\cdot}$ depicted in Figure~\ref{fig:service-cactus-construction-cycle} will be used to induce flows inside the respective constructions or will be used to propagate flows respectively.

\begin{definition}[Extended Graph for Cycles]
\label{def:service-cactus-extended-graph-cycle}
The extended graph~$\VGextcycle =~(\VVextcycle, \VEextcycle)$ for cycle~$C_k \in \Cycles$ is defined as follows:

$\VVextcycle = \VVextcycleSources \cup \VVextcycleTargets \cup \VVextcycleSubstrate$, with
\begin{alignat*}{3}
\VVextcycleSources && = &  \{ u^+_{r,i} | i = \VVcycleSource, u \in \SVTypes[\Vtype(i)]\}\\
\VVextcycleTargets && = &  \{ u^-_{r,j} | j = \VVcycleTarget, u \in \SVTypes[\Vtype(j)]\} \\
\VVextcycleSubstrate && = &  \{ u^{i,j}_{r,w} |~(i,j) \in  \VEcycle, u \in \SV, w \in \SVTypesCycle \}
\end{alignat*}
We set~$\VEextcycle = \VEextcycleSources \cup \VEextcycleTargets \cup \VEextcycleSubstrate \cup \VEextcycleF$, with:
{
\begin{alignat*}{3}
\VEextcycleSubstrate && = & \{(u^{i,j}_{\req,w}, v^{i,j}_{\req,w})|(i,j)\in \VEcycleSame,(u,v)\in \SE, w \in \SVTypesCycle \} \cup \\
 && & \{(v^{i,j}_{\req,w}, u^{i,j}_{\req,w})|(i,j) \in \VEcycleDiff,(u,v)\in \SE, w \in \SVTypesCycle \} \\
\VEextcycleSources && = & \{(u^+_{\req,i}, u^{i,j}_{\req,w})| i = \VVcycleSource,(i,j)\in \VEcycle, u \in \SVTypes[\Vtype(i)]\hspace{-5pt},  w \in \SVTypesCycle \} \\
\VEextcycleTargets &&= &\{(w^{i,j}_{\req,w}, w^-_{\req,j})| j = \VVcycleTarget,(i,j) \in \VEcycle, w \in \SVTypesCycle \}  \\
\VEextcycleF && = &\{(u^{i,j}_{\req,w}, u^{j,k}_{\req,w})|(i,j),(j,k)\in \VEcycle, u \in \SVTypes[\Vtype(j)], w \in \SVTypesCycle \}
\end{alignat*}

Note that we employ the abbreviation $\SVTypesCycle$ to denote the substrate nodes that may host the cycle's target network function, i.e. $\SVTypes[\Vtype(\VVcycleTarget)]$ (cf.~Definition~\ref{def:service-cactus-extended-graph-cycle}).

}
\end{definition}

\begin{figure}[t!]
\centering
\includegraphics[width=1\columnwidth]{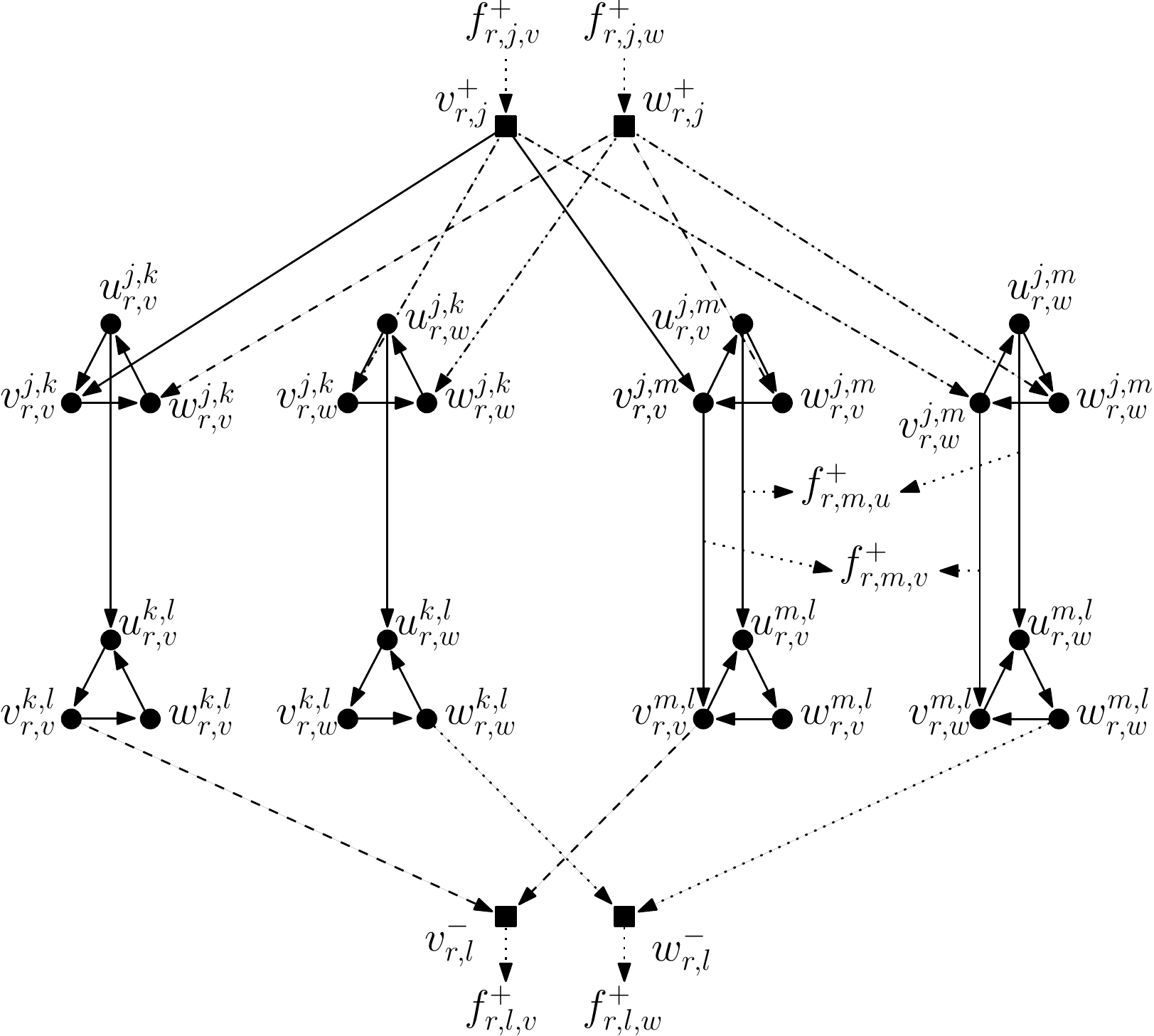}
\caption{Extended graph~$\VGextcycle[\req][1]$ of service cactus request~$r$ and cycle~$C_1$ depicted in  Figure~\ref{fig:service-cactus-request}. Similar styles of dashing the edges incident to the super sources~$v^+_{r,j}$,~$w^+_{r,j}$ and sinks~$v^-_{r,l}$ and~$v^+_{r,j}$ indicates that the same amount of flow will be sent along them.}
\label{fig:service-cactus-construction-cycle}
\end{figure}

\subsection{Linear Programming Formulation}
\label{sec:cactus-linear-programming}
We propose Integer Programs~\ref{alg:SCGEP-IP} and \ref{alg:SCGEP-IP-C} to solve SCGEP-P and SCGEP-C respectively. The IPs are based on the individual service cactus graph decompositions and the corresponding extended graph constructions discussed above. As IP~\ref{alg:SCGEP-IP-C} reuses all of the core Constraints, we restrict our presentation to IP~\ref{alg:SCGEP-IP}.
We use variables~$\varFlowInput \in \{0,1\}$ to indicate the amount of flow that will be induced at the super sources of the different
extended graphs~(paths or cycles)~for virtual node $i \in \VVSourcesTargets$ and substrate node~$u \in \SVTypes[\Vtype(i)]$. Note that these flow variables are included in the Figures~\ref{fig:service-cactus-construction-path} and \ref{fig:service-cactus-construction-cycle}.
Additionally flow variables~$f_{r,e} \in \{0,1\}$ are defined for all edges contained in any extended graph construction, i.e. for~$\req \in \requests$  and~$e \in \VEextCGFlowEdges$, with~$\VEextCGFlowEdges =~(\bigcup_{P_k \in \Paths} \VEextpath)~\cup~(\bigcup_{C_k \in \Cycles} \VVextcycle)$. Furthermore, the already introduced variables~(see Section~\ref{sec:integer-linear-program})~$x_\req \in \{0,1\}$ and~$l_{\req, x,y} \geq 0$ are used to model the embedding decision of request $\req \in \requests$ and the allocated loads on resource $(x,y) \in  \SR$ for request $\req \in \requests$ respectively.

\begin{figure*}[tbhp]
{
 \LinesNotNumbered
 \renewcommand{\arraystretch}{1.4}

\removelatexerror

 \begin{IPFormulationStar}{H}

 \SetAlgorithmName{Integer Program}{}{{}}

 \scalebox{0.98}{

 \newcommand{\spaceIt}{\qquad\quad\quad}
 \newcommand{\miniSpace}{\hspace{1.5pt}}

 \centering

\hspace{-30pt} 
 \begin{tabular}{FRLQB}
    \multicolumn{3}{C}{~~~~~~~~~~~~~~~~~~~~~\textnormal{max~}  \sum \limits_{\req \in \requests} \Vprofit \cdot x_{\req} }& & 
  \tagIt{alg:SCGEPObj}\\[10pt]
  
  \sum \limits_{u \in \SVTypes[\Vtype(\VVroot)]} \varFlowInput[\req][\VVroot] & = & x_{\req} & \forall \req \in \requests & \tagIt{alg:SCGEP:FlowInduction-Initial} \\
  
      \sum \limits_{w \in \SVTypesCycle} f_{\req,~(u^+_{\req,i}, u^{i,j}_{\req,w})} & = & \varFlowInput[\req][i] & \forall \req \in \requests, C_k \in \Cycles, ~(i,j) \in  \VEcycleBranchR, i = \VVcycleSource, u \in \SVTypes[\Vtype(i)] & \tagIt{alg:SCGEP:FlowInduction-cycle-Right} \\
  
  f_{\req,~(u^+_{\req,i}, u^{i,j}_{\req,w})} -  f_{\req,~(u^+_{\req,i}, u^{i,j'}_{\req,w})}& = & 0 & \forall \req \in \requests, C_k \in \Cycles,~(i,j) \in  \VEcycleBranchR,~(i,j') \in  \VEcycleBranchL, i = \VVcycleSource, w \in  \SVTypesCycle & \tagIt{alg:SCGEP:FlowInduction-cycle-Left} \\
  
  f_{\req,~(u^+_{\req,i}, u^{i,j}_{\req})} & = & \varFlowInput[\req][i] & \forall \req \in \requests, P_k \in \Paths, ~(i,j) \in  \VEpath, i = \VVpathSource, u \in \SVTypes[\Vtype(i)] & \tagIt{alg:SCGEP:FlowInduction-path} \\

	\sum \limits_{e \in \delta^+(u)} f_{\req, e}  -  \sum \limits_{e \in \delta^-(u)}  f_{\req, e}& = & 0   &  \forall \req \in \requests, u \in \VVextCGFlowNodes & \tagIt{alg:SCGEP:FlowPreservation} \\
	
   f_{\req,~(w^{i,j}_{\req,w}, w^-_{\req,j})}  & = &  \varFlowInput[\req][j][w]   &  \forall \req \in \requests, C_k \in \Cycles,~(i,j) \in  \VEcycleBranchR, j = \VVcycleTarget, w \in \SVTypes[\Vtype(j)]\  & \tagIt{alg:SCGEP:setting-next-input-cycle-end} \\
   
   f_{\req,~(u^{i,j}_{\req}, u^-_{\req,j})}  & = &  \varFlowInput[\req][j][u]   &  \forall \req \in \requests, P_k \in \Paths,~(i,j) \in  \VEpath, j = \VVpathTarget, u \in \SVTypes[\Vtype(j)]\  & \tagIt{alg:SCGEP:setting-next-input-path-end} \\
   
   \sum \limits_{w\in \SVTypesCycle}  f_{\req,~(u^{i,j}_{\req,w}, u^{j,l}_{\req,w})}  & = &  \varFlowInput[\req][j][u]   &  \forall \req \in \requests, C_k \in \Cycles, j \in \VVbranchingcycle, (i,j),(j,l) \in  \VEcycle, u \in \SVTypes[\Vtype(j)]\  & \tagIt{alg:SCGEP:setting-next-input-cycle-middle} \\

	 \sum \limits_{\scriptsize
	 
	(e,i) \in  \VEextCGVertical[\req][\type]
	}  \hspace{-12pt} \Vcap(i)~\cdot f_{\req,e} + \hspace{-12pt}
	 \sum \limits_{\scriptsize
	 \begin{array}{c}
	 i \in  \VVSourcesTargets \setminus \VVbranchingcycle\\
	 \Vtype(i)~= \type
	 \end{array}
	 } \hspace{-16pt} \Vcap(i)~\cdot \varFlowInput[\req][i][u]   & = & l_{\req,\type,u} & \forall \req \in \requests, (\type,u) \in  \SRV & \tagIt{alg:SCGEP:nodeLoad} \\

	\sum \limits_{	(e,i,j) \in  \VEextCGHorizontal} \Vcap(i,j)~\cdot f_{\req,e} & = & l_{\req,u,v} & \forall \req \in \requests, (u,v) \in  \SE & \tagIt{alg:SCGEP:edgeLoad} \\

\sum \limits_{\req \in \requests } l_{\req,x,y}& \leq & \Scap(x,y)~& \forall (x,y) \in  \SR & \tagIt{alg:SCGEP:capacities} \\

  x_{\req} & \in & \{0,1\} &  \forall \req \in \requests & \tagIt{alg:SCGEP:embedding_variable} \\
  
  \varFlowInput & \in & \{0,1\} & \forall i \in \VVSourcesTargets & \tagIt{alg:SCGEP:flow-input-variable} \\
  
  f_{\req,e} & \in & \{0,1\} &  \forall \req \in \requests, e \in \VEextCGFlowEdges     & \tagIt{alg:SCGEP:fow-variables} \\

  l_{\req,x,y} & \geq & 0 &  \forall \req \in \requests, (x,y) \in  \SR     & \tagIt{alg:SCGEP:VariableLoad} \\
 
 \end{tabular}
 }
 \caption{SCGEP-P}
 \label{alg:SCGEP-IP}
 \end{IPFormulationStar}
 }
\end{figure*}

Embeddings are realized as follows. If,~$x_\req=1$ holds, by Constraint~\ref{alg:SCGEP:FlowInduction-Initial} exactly one of the flow values~$\varFlowInput[\req][\VVroot][u]$ is set to~$1$ for~$u \in \SVTypes[\Vtype(\VVroot)]$ and hence a flow will be induced in all paths and cycles having the arbitrarily chosen root~$\VVroot\in \VV$ as source. Specifically, for all paths~$P_k \in \Paths$, Constraint~\ref{alg:SCGEP:FlowInduction-path} sets the flow values of edges incident to the respective super source nodes in the extended graphs~$\VGpath$ according to the values to~$\varFlowInput$. For cycles, Constraint~\ref{alg:SCGEP:FlowInduction-cycle-Right} induces a flow analogously, albeit deciding which target substrate node~$w \in \SVTypesCycle$ will host the function~$\VVcycleTarget$. Importantly, Constraint~\ref{alg:SCGEP:FlowInduction-cycle-Right} only induces flow in the branch~$\VEcycleBranchR$ of cycle~$C_k \in \Cycles$. Indeed, Constraint~\ref{alg:SCGEP:FlowInduction-cycle-Left} forces the flow in the branch $\VEcycleBranchL$ to reflect the flow decisions of the branch~$\VEcycleBranchR$. Having induced flow in the extended networks for paths~$P_k \in \Paths$ and cycles~$C_k \in \Cycles$, Constraint~\ref{alg:SCGEP:FlowPreservation} enforces flow preservation in all of the extended networks. Here we define ~$\VVextCGFlowNodes =~(\bigcup_{P_k \in \Paths} \VVextpathSubstrate)~\cup~(\bigcup_{C_k \in \Cycles} \VVextcycleSubstrate)$ to be all nodes that are neither a super source nor a super sink in the respective extended graphs. As flow preservation is enforced within all extended graphs, the flow of path and cycle graphs must eventually terminate in the respective super sinks. By Constraints~\ref{alg:SCGEP:setting-next-input-cycle-end} and \ref{alg:SCGEP:setting-next-input-path-end} the variables~$\varFlowInput[\req][j][\cdot]$ of the target of the cycle~$C_k \in \Cycles$ -- i.e.~$j=\VVcycleTarget$ -- or the target of the path~$P_k \in \Paths$ -- i.e.~$j=\VVpathTarget$ -- are set according to the amount of flow that reaches the respective sink in the respective extended graph. This effectively induces flows in all the extended graphs of cycle~$C_k \in \Cycles$ or path~$P_k \in \Paths$ for which~$\VVcycleSource \in \VVcycleTargets \cup \VVpathTargets$ or~$\VVpathSource \in \VVcycleTargets \cup \VVpathTargets$ holds, respectively.
However, paths~$P_k \in \Paths$ or cycles~$C_k \in \Cycles$ may be spawned also at nodes lying inside another cycle~$C_{k'}$. This is e.g. the case for path $P_2$ in Figure~\ref{fig:service-cactus-request}, as the virtual node $m \in \VV$ is an inner node of the branch $\VEcycleBranchL$ of cycle $C_1$. To nevertheless, induce the right amount of flow at the different super sources in the graph $\VGextpath[\req][2]$, the sum of flows along processing edges inside the \emph{parallel path constructions} have to be computed, as visualized in Figure~\ref{fig:service-cactus-construction-cycle}: the flow variable $f^+_{\req,m,v}$ needs to equal the sum of the flow along edges $(v^{j,m}_{\req,v},v^{m,l}_{\req,v})$ and $(v^{j,m}_{\req,w},v^{m,l}_{\req,w})$ and the flow variable $f^+_{\req,m,u}$ needs to equal the sum of the flow along edges $(u^{j,m}_{\req,v},u^{m,l}_{\req,v})$ and $(u^{j,m}_{\req,w},u^{m,l}_{\req,w})$ as the respective edges denote the processing of function $m \in \VV$ on node $v \in \SV$ and $w \in \SV$ respectively. Constraint~\ref{alg:SCGEP:setting-next-input-cycle-middle} realizes this functionality, i.e. the inter-layer edges that denote processing are summed up to set the respective flow inducing variables $f^+_{\req,\cdot,\cdot}$. As paths may only overlap at sources and sinks (cf. Definition~\ref{def:service-cactus-graph-decomposition}), Constraints~\ref{alg:SCGEP:setting-next-input-cycle-end}-\ref{alg:SCGEP:setting-next-input-cycle-middle} propagate flow induction across extended graphs.

Constraints~\ref{alg:SCGEP:nodeLoad} and \ref{alg:SCGEP:edgeLoad} set the load variables $l_{\req,x,y}$ induced by the embeding of request $\req \in \requests$ substrate resources $(x,y) \in  \SR$. Within the Integer Program~\ref{alg:SCGEP-IP}, we again make use of specific index sets to simplify notation (cf. Definition of $\VEextHorizontal$ and~$\VEextVertical$ in Section~\ref{sec:integer-linear-program} and Integer~Program~\ref{alg:SCEP-IP}). Concretely, we introduce index sets~$\VEextCGHorizontal$ and~$\VEextCGVertical$ to contain the edges that indicate the processing on substrate resource $(\type,u) \in  \SRV$  and the edge $(u,v) \in  \SE$  respectively as follows:
\begin{align*}
& \VEextCGVertical = \\
& ~ \{((u^{i,j}_{\req}, u^{j,l}_{\req}),j) | P_k \in \Paths, (i,j),(j,l) \in \VEpath, \Vtype(j) = \type \} \cup \\
&~ \{((u^{i,j}_{\req,w}, u^{j,l}_{\req,w}),j)| C_k \in \hspace{-2pt} \Cycles, \hspace{-2pt}(i,j),\hspace{-2pt}(j,l) \hspace{-2pt}\in \hspace{-2pt}\VEcycle \hspace{-2pt}, \hspace{-2pt} \Vtype(j)\hspace{-2pt} = \hspace{-2pt}\type \hspace{-2pt}, w \hspace{-2pt}\in \hspace{-2pt} \SVTypesCycle  \} \\
& \VEextCGHorizontal  = \\
&~\{((u^{i,j}_{\req}, v^{i,j}_{\req}), i,j) | P_k \in \Paths,(i,j)\in \VEpathSame\}  \cup\\
&~ \{((v^{i,j}_{\req}, u^{i,j}_{\req}), i,j)| P_k \in \Paths,(i,j)\in \VEpathDiff \}  \cup\\
&~ \{((u^{i,j}_{\req,w}, v^{i,j}_{\req,w}), i,j)| C_k \in \Cycles,(i,j)\in \VEcycleSame, w\in \SVTypesCycle  \} \cup\\
&~ \{((v^{i,j}_{\req,w}, u^{i,j}_{\req,w}), i,j)| C_k \in \Cycles,(i,j)\in \VEcycleDiff, w\in \SVTypesCycle  \}
\end{align*}

In contrast to the definition of~$\VEextVertical$ in Section~\ref{sec:integer-linear-program}, the set~$\VEextCGVertical$ only collects the internal edges of extended graphs representing the processing at node $u$ having type $\type$. Hence, processing at nodes that are the start or the target of any of the decomposed paths or cycles must be separately accounted for. Concretely, in Constraint~\ref{alg:SCGEP:nodeLoad}, for all nodes $i \in \VVSourcesTargets \setminus \bigcup_{C_k \in \Cycles} \VVbranchingcycle$, the respective flow induction variables~$\varFlowInput[\req][i][u]$ are added. 

With respect to accounting for the load of substrate edges~$(u,v) \in  \SE$, we note that the only difference to the construction in Section~\ref{sec:integer-linear-program} is the re-reorientation of edge directions and the consideration of both paths and cycles. 

Lastly, the objective~\ref{alg:SCGEPObj} as well as the Constraint~\ref{alg:SCGEP:capacities} have not changed with respect to~(Integer)~Linear Program~\ref{alg:SCEP-IP}.

Summarizing the workings of~Integer Program~\ref{alg:SCGEP-IP}, we note the following:
\begin{itemize}
\item For each request $\req \in \requests$, a breadth-first search from an arbitrarily root $\VVroot \in \VV$ is performed to obtain graph $\VGbfs$. This graph is decomposed according to Definition~\ref{def:service-cactus-graph-decomposition}.
\item Within IP~\ref{alg:SCGEP-IP}, for each request~$\req \in \requests$ flow is induced in the respective extended graphs of cycles and paths via the variables~$\varFlowInput$ for all nodes~$i \in \VVSourcesTargets~$ and substrate nodes~$u \in \SVTypes[\Vtype(i)]$~(see Constraints~\ref{alg:SCGEP:FlowInduction-cycle-Right}, \ref{alg:SCGEP:FlowInduction-cycle-Left}, and \ref{alg:SCGEP:FlowInduction-path}).
\item Within each extended graph flow preservation holds (excluding super sources and sinks).
\item For cycles~$C_k \in \Cycles$, it is additionally enforced that the amount of flow sent towards destination~$w \in  \SVTypesCycle$ must be equal for both branches~(see Constraints~\ref{alg:SCGEP:FlowInduction-cycle-Right} and \ref{alg:SCGEP:FlowInduction-cycle-Left}).
\item As flow must terminate at the sinks of the extended graphs, the variables~$\varFlowInput[\req][\VVpathTarget][u]$ of path~$P_k \in \Paths$ and the variables~$\varFlowInput[\req][\VVcycleTarget][u]$ are eventually set~(see Constraints~\ref{alg:SCGEP:setting-next-input-cycle-end} and \ref{alg:SCGEP:setting-next-input-path-end}), thereby potentially inducing flows in other extended graphs. Additionally, if another cycle or path is spawned at a node internal to a branch of a cycle the respective flow input variables $f^+_{\req,\cdot,\cdot}$ are set in  Constraint~\ref{alg:SCGEP:setting-next-input-cycle-middle}.
\item Hence, if~$x_\req=1$ holds for request~$\req \in \requests$, then by Constraint~\ref{alg:SCGEP:FlowInduction-Initial} one of the variables~$\varFlowInput[\req][\VVroot][u]$ must be set to one. By the above observations, this will induce flows in \emph{all} extended graphs. 
\end{itemize}

{

 \newcolumntype{F}{>{$\displaystyle\,}r<{$}@{\hspace{0.0em}}}
 \newcolumntype{C}{>{$\displaystyle\,}c<{$}@{\hspace{0.0em}}}
 \newcolumntype{B}{>{$\displaystyle\,}r<{$}@{\hspace{0.0em}}}
 \newcolumntype{R}{>{$\displaystyle}r<{$}@{\hspace{0.2em}}}
 \newcolumntype{S}{>{$\displaystyle}r<{$}@{\hspace{0.2em}}}
 \newcolumntype{L}{>{$\displaystyle}l<{$}@{\hspace{0.2em}}}
 \newcolumntype{Q}{>{$\displaystyle}l<{$}@{\hspace{0.3em}}}

\LinesNotNumbered
\renewcommand{\arraystretch}{1.5}

\begin{IPFormulation}{tb}

\SetAlgorithmName{Integer Program}{}{{}}

\scalebox{0.94}{

\newcommand{\spaceIt}{\qquad\quad\quad}
\newcommand{\miniSpace}{\hspace{1.5pt}}

\centering
\hspace{-32pt}
\begin{tabular}{FFCLLBB}
\multicolumn{5}{C}{\textnormal{min~}  \sum_{\req \in \requests} \sum \limits_{(x,y) \in   \SR} \Scost(x,y)~\cdot l_{\req,x,y} } & ~ & \tagIt{alg:SCGEPObj-C}\\[10pt]
                  
\qquad \qquad \qquad \qquad \qquad &  \textnormal{\ref{alg:SCGEP:FlowInduction-Initial} - \ref{alg:SCGEP:capacities}} & ~\textnormal{and}~  & \textnormal{\ref{alg:SCGEP:flow-input-variable} - \ref{alg:SCGEP:VariableLoad}} & & & \\
&  x_{\req}  &=&  1 & \qquad \forall \req \in \requests & & \tagIt{alg:SCGEP:embedding_variable_without} \\
\end{tabular}
}
\caption{SCGEP-C}
\label{alg:SCGEP-IP-C}
\end{IPFormulation}
}

\subsection{Decomposition Algorithm / Correctness}
\label{sec:cactus:decomposition-algorithm-for-service-cacti-and-correctness}

In the following we present Algorithm~\ref{alg:decompositionAlgorithmSCGEP} to decompose a given solution to the~(relaxation of)~Integer Program~\ref{alg:SCGEP-IP}. As in the decomposition algorithm for linear service chains (cf. Algorithm~\ref{alg:decompositionAlgorithm}), the output of the algorithm is a set~$\mathcal{D}_{\req}$ for each request~$\req \in \requests$ consisting of triples~$(\prob,\mapping,\load)$, where~$\prob \in [0,1]$ is the fractional embedding value of the~(frational)~embedding defined by the mapping~$\mapping$. We again note that the load function $\load: \SR \to \preals$ indicates the load on substrate resource $(x,y) \in  \SR$, if the mapping $\mapping$ is \emph{fully} embedded.

\begin{figure}

\scalebox{0.90}{
\begin{minipage}{1.07\columnwidth}

\begingroup
\removelatexerror

\begin{algorithm*}[H]

\let\oldnl\nl
\newcommand{\nonl}{\renewcommand{\nl}{\let\nl\oldnl}}

\SetKwInOut{Input}{Input}\SetKwInOut{Output}{Output}
\SetKwFunction{ProcessPath}{ProcessPath}{}{}
\SetKwFunction{LP}{LP}
\SetKwFunction{LP}{LP}

\newcommand{\SET}{\textbf{set~}}
\newcommand{\DEFINE}{\textbf{define~}}
\newcommand{\AND}{\textbf{and~}}
\newcommand{\LET}{\textbf{let~}}
\newcommand{\WITH}{\textbf{with~}}
\newcommand{\COMPUTE}{\textbf{compute~}}
\newcommand{\CHOOSE}{\textbf{choose~}}
\newcommand{\DECOMPOSE}{\textbf{decompose~}}
\newcommand{\FORALL}{\textbf{for all~}}
\newcommand{\OBTAIN}{\textbf{obtain~}}
\newcommand{\WITHPROBABILITY}{\textbf{with probability~}}

\Input{Substrate~$\SG=(\SV,\SE)$, set of requests~$\requests$, solution~$(\vec{x},\vec{f^+},\vec{f},\vec{l}) \in  \FeasibleLP$}
\Output{Set of fractional mappings~$\PotEmbeddings = \{(\prob,\mapping,\load)\}_k$ for each request~$\req \in \requests$}

\For{$\req \in \requests$}
{
	\SET $\PotEmbeddings  \gets \emptyset$ \AND $k \gets 1$\\
	\While{$x_\req > 0$ }
	{
		
		\SET $\mapping = (\mapV,\mapE)~\gets (\emptyset,\emptyset)$ \label{alg:sc-decomposition:init-maps}\\
		\SET $l^k_{\req}(x,y)~\gets 0$ \FORALL $(x,y) \in  \SR $\\
		
		\SET $\Queue = \{\VVroot \}$\\
		
		\CHOOSE $u \in \SVTypes[\Vtype(\VVroot)]$ \WITH $\varFlowInput[\req][\VVroot][u] > 0$\\
		\SET $\mapV(\VVroot)~\gets u$\\
		\SET $\Variables \gets \{x_\req, f^+_{\req,\VVroot,u}\}$\\
		
		\While{$|\Queue| > 0$}{
			\CHOOSE $i \in \Queue$ \AND \SET$\Queue \gets \Queue \setminus \{i\}$\\
			\ForEach{$C_k \in \Cycles$ \textnormal{\WITH} $\VVcycleSource = i$}{
				\CHOOSE $P_1 = \langle s, \dots, t \rangle \in \VGextcycleFlowBranchR$ \label{alg:sc-decomposition:chooseP_1} \\ 
\nonl				\quad \WITH $s =(\mapV(i))^+_{\req,\VVcycleSource}$ \AND $t \in \VVextcycleTargets$  \\
				\CHOOSE $P_2 = \langle s, \dots, t \rangle \in \VGextcycleFlowBranchL$  \label{alg:sc-decomposition:chooseP_2}\\
				\ProcessPath{$C_k,\VEcycleBranchR, P_1, \mapV, \mapE, \Variables, \Queue$}\\
				\ProcessPath{$C_k,\VEcycleBranchL, P_2, \mapV, \mapE, \Variables, \Queue$}\\
			}
			\ForEach{$P_k \in \Paths$ \textnormal{\WITH} $\VVpathSource = i$}{
				\CHOOSE $P = \langle s, \dots, t \rangle \in \VGextpathFlow$ \\ 
				\quad \WITH $s =(\mapV(i))^+_{\req,\VVpathSource}$ \AND $t \in \VVextpathTargets$  \\
				\ProcessPath{$P_k,\VEpath P, \mapV, \mapE, \Variables, \Queue$}\\
			}
		}
					
		\For{$i \in \VV$}{
			\SET $\load(\Vtype(i),\mapV(i))~\gets \load(\Vtype(i),\mapV(i))~+ \Vcap(i)$ \\
		}
		
		\For{$(i,j) \in  \VE$}{	
			\For{$(u,v) \in  \mapE(i,j)$}{
				\SET $\load(u,v)~\gets \load(u,v)~ + \Vcap(i,j)$ \\ \label{alg:sc-decomposition:endMappingCreation}
			} 
		}
		
		\SET $\prob \gets \min \Variables$ \\
		\SET $v \gets v - \prob$ \FORALL $v \in \Variables$ \\	
		\SET $l_{\req,x,y} \gets l_{\req,x,y} - \prob \cdot \load(x,y)$ \FORALL $(x,y) \in  \SR$\\
		\SET $\PotEmbeddings \gets \PotEmbeddings \cup \{D^k_{\req}\}$ \WITH $D^k_{\req} = (f^k_{\req},m^k_{\req},l^k_{\req})$\\
		\SET $k \gets k + 1$\\
	}
}
\KwRet{$\{\PotEmbeddings | \req \in \requests\}$}
\BlankLine
\BlankLine
\ProcessPath{$K,P_V,P_{\textnormal{ext}}, \mapV, \mapE, \Variables, \Queue$}\\
\Indp
\SetInd{2em}{0em}
\ForEach{$(i,j) \in  P_V$}{
	\SET $P_S \gets \langle (u,v)~| (u^{i,j}_{\req,\cdot},v^{i,j}_{\req,\cdot}) \in  P_{\textnormal{ext}}, u\neq v \rangle $\\
	\LET $P_S = \langle(u^1,u^2), (u^2,u^3), \dots,(u^{n-1},u^n)~\rangle$\\
	\eIf{$(i,j) \in  \VEDiff $}{
		
		\SET $\mapE(j,i)~\gets \langle u^{n}, u^{n-1}, \dots, u^1 \rangle$ \\
	}{
		\SET $\mapE(j,i)~\gets \langle u^{1}, u^{2}, \dots, u^n \rangle$ \\
	}
	
	\uIf{$j \neq \VVKTarget$ \textnormal{\AND} $(u^{i,j}_{\req,\cdot},u^{j,\cdot}_{\req,\cdot}) \in  P_{\textnormal{ext}}$}{
		\SET $\mapV(j)~\gets u$\\
	}
	\uElseIf{$j = \VVKTarget$ \textnormal{\AND} $(u^{i,j}_{\req,\cdot},u^-_{\req,j}) \in  P_{\textnormal{ext}}$}{
		\SET $\mapV(j)~\gets u$\\
	}
	\If{$j \in \VVSourcesTargets$}{
		\SET $\Variables \gets \Variables \cup \{f^+_{\req, j, \mapV(j)} \}$\\
		\SET $\Queue \gets \Queue \cup \{j\}$\\
	}

}
\SET $\Variables \gets \Variables \cup \{f_{\req, e} | e \in P_{\textnormal{ext}} \}$
\caption{Decomposition Algorithm Service Cacti}
\label{alg:decompositionAlgorithmSCGEP}
\end{algorithm*}
\endgroup
\end{minipage}}
\end{figure}

For each request $\req \in \requests$, the algorithm decomposes the flows in the extended graphs iteratively as long as~$x_\req > 0$ holds. This is done by placing Initially, only the root~$\VVroot \in \VV$ is placed in the queue~$\mathcal{Q}$ and one of the substrate nodes~$u \in \SVTypes[\Vtype(\VVroot)]$ with~$f^+_{r, \VVroot, u} >0$ is chosen as the node to host $\VVroot$. Such a node must always exist by Constraint~\ref{alg:SCGEP:FlowInduction-Initial} of Integer Program~\ref{alg:SCGEP-IP} and the node mapping is set accordingly. The queue $\mathcal{Q}$ will only contain nodes that are sources of paths or cycles in the decomposition, i.e. which are contained in $\VVSourcesTargets$. Furthermore, by construction, each extracted node from the queue will already be mapped to some specific substrate node via the function $\mapV$. The set $\mathcal{V}$ is used to keep track of all variables of the Integer Linear Program whose value used in the decomposition process.

For each node in the queue $i \in \mathcal{Q}$, all cycles~$C_k \in \Cycles$ and paths~$P_k \in \Paths$ starting at~$i$, i.e. that $\VVcycleSource = i$ or $\VVpathSource = i$ holds,  are handled one after another using the function \texttt{ProcessPath}. We start by discussing how paths~$P_k$ are handled~(see Lines 17-20). Note that by construction, the source of~$P_k$ is already mapped to the substrate node $\mapV(i)$. Within the graph~$\VGextpathFlow =~(\VVextpathFlow, \VEextpathFlow)$ with~$\VVextpathFlow=\VVextpath$ and~$\VEextpathFlow = \{e \in \VEextpath | f_{\req, e} > 0 \}$ a path~$P$ is chosen from the respective super source~$(\mapV(i))^+_{r,P_k} \in \VVextpath$ towards any of the sinks~$\VVextpathTargets$. By definition of~$\VEextpathFlow$ the flow along~$P$ is greater~$0$. The Function \texttt{ProcessPath} is handed the path identifier~$K = P_k$, the virtual path~$P_V = \VEpath \subseteq \VE$, the path in the extended network~$P_{\textnormal{ext}} = P \subseteq \VEextpath$, as well as the mappings and the list of variables~$\mathcal{V}$ and the queue~$\mathcal{Q}$. 

The function \texttt{ProcessPath} proceeds as follows. For all virtual edges~$(i,j) \in  P_V$ the edge mapping is set by projecting the edges used in the extended graph to the substrate path~$P_S \in \SE$  in the substrate network (see Line~34). Then, depending on whether the virtual edge's direction was reversed, the edge mapping $\mapE$ is set by either reversing the node order of path $P_S$ or keeping it the same. Note that by Definition~\ref{def:valid-mapping} the mapping $\mapE$ must be node path. 
After having taken care of the edge mappings in Lines~34 to 39, the node mappings are set in Lines~40 to 43. The node mapping of node $j$ -- for each virtual edge $(i,j) \in  P_V$ -- is set by considering the inter-layer edges or the edge towards the super source in $P_{\textnormal{ext}}$. Concretely, if $e =~(u^{i,j}_{\req,\cdot}, u^{j,\cdot}_{\req,\cdot}) \in  \VEextpathF$ exists and if~$j$ is not the target of the respective path or cycle, then $j$ is mapped onto node $u$. On the other hand, if $j$ is the target of the respective path or cycle, and $e =~(u^{i,j}_{\req,\cdot}, u^-_{\req, j}) \in  \VEextpathTargets$ holds, then $j$ is mapped onto substrate node $u$. Then, if~$j$ is the source of other paths or cycles or if~$j$ is the target of the given path, the respective flow input variable of the node mapping~$f^+_{\req, j, \mapV(j)}$ is added to~$\mathcal{V}$ and the node is added to the queue of nodes~$\mathcal{Q}$ to process.
Lastly, also the edge variables used along the path~$P_{\textnormal{ext}}$ are added to~$\mathcal{V}$. Note that the \texttt{ProcessPath} function indeed maps all edges and nodes of $P_V$. By construction, the start node of $P_V$ will always be mapped. Hence, it is sufficient to only map the tail $j$ of the virtual edges $(i,j) \in  P_V$. 

The processing of cycles is handled similarly. First the path~$P_1 \in \VGextcycleFlowBranchR$ is chosen where the index~$f$ again implies that the flow along edges in the path is greater~$0$. While for~$P_1$ only the source is fixed according to the mapping of the source of the cycle, the path~$P_2 \in \VGextcycleFlowBranchL$ is constrained to have the same source and target as $P_1$. Finding appropriate paths is always possible by Constraints~\ref{alg:SCGEP:FlowInduction-cycle-Right} and \ref{alg:SCGEP:FlowInduction-cycle-Left}. Concretely, Constraint~\ref{alg:SCGEP:FlowInduction-cycle-Right} enforces that there is a path towards any of the sinks given that~$f^+_{\req, s, (\mapV(s))^+_{\req, s}} > 0~$ holds for~$s = \VVcycleSource$. By Constraint~\ref{alg:SCGEP:FlowInduction-cycle-Left}, the amount of flow sent from~$s$ towards the destination~$t$ in one branch must equal the amount of flow sent from~$s$ towards destination~$t$ in the other.

The remaining code in Lines~21-30, first the load allocations $\load$ for resources are computed in Lines~21-25 (cf. Lines~). Afterwards, the minimum (fractional)~embedding value~$f^k_{\req}$ over all variable contained in $\mathcal{V}$ (always being greater than~$0$ by construction)~is computed. Accordingly, all variables' values are decremented by $\prob$ in Line~27 together with the respectively scaled load variables of the integer program in Line~28. Lastly, the triple consisting of the~(fractional)~embedding value, the mapping and the loads is added to the set of decompositions~$\mathcal{D}_{\req}$ and $k$ is incremented.

Summarizing the above, we highlight that by Line~21 the mapping~$\mapping=(\mapV,\mapE)$ is valid, i.e. all virtual nodes are mapped to a node supporting the respective network function and edges between virtual nodes are realized using paths in the substrate~(cf.~Definition~\ref{def:valid-mapping}).  Without further proof, we state the analogue to Lemma~\ref{lem:validity-of-kth-decomposition}:

\begin{lemma}
\label{lem:validity-of-kth-decomposition-scgep}
Each mapping~$\mapping=(\mapV,\mapE)$ found by Algorithm~\ref{alg:decompositionAlgorithmSCGEP} for request~$\req \in \requests$ in the~$k$-th iteration is \emph{valid}.
\begin{proof}[Proof sketch]
The proof works by induction and we consider only a single request $\req \in \requests$. We argue that as long as the constraints specified in Integer Program~\ref{alg:SCGEP-IP} hold, always a valid mapping can be constructed. It is clear, that initially the Constraints  of Integer Program~\ref{alg:SCGEP-IP} are all satisfied.  

In the above synopsis of Algorithm~\ref{alg:decompositionAlgorithmSCGEP} we have already argued that the ``choose''-operations in Lines 7, 13, 14, and 18 are well-defined, given that all constraints hold. Similarly, we have argued that the \texttt{ProcessPath} function correctly maps nodes and edges. Within this proof we therefore only show (i)~that indeed all nodes and edges are mapped and that (ii)~all constraints of Integer Program~\ref{alg:SCGEP-IP} hold after decreasing the respective variables.

With respect to (i), we note that by the definition of~$\VGbfs$ all virtual nodes are reachable from~$\VVroot$. As the cycles and paths cover all edges~(and hence all nodes)~and sources and targets are added by the function \texttt{ProcessPath} to the queue~$\mathcal{Q}$, eventually all virtual edges~(and hence all virtual nodes)~will be processed. As the mapping of~$\VVroot$ is fixed initially and the targets of edges in~$\VEbfs$ are always mapped, the node mapping will be complete and valid by construction. Furthermore, by definition of the extended graphs, the path assignment of Lines~34-39 connects the respective substrate node locations and therefore is also valid.

Considering (ii), i.e. the preservation of the integer programs constraints' validity, we note that all variables denoting the usage of resources are collected in the set $\mathcal{V}$. Concretely, the flow variables $f_{\req,\cdot}$ along each path or cycle branch in the extended graph are added (cf. Line~46)~together with the flow induction variables $f^+_{\req,\cdot,\cdots}$ (cf. Lines~9 and 45)~and the embedding variable $x_r$ (cf. Line~9). Then in  Line~27 the values of the respective variables are reduced exactly by the amount $\prob > 0$. Furthermore, the load variables are reduced in Line~28. We claim that again all constraints of Integer Program~\ref{alg:SCGEP-IP} hold. We omit the proof that the Constraints~58 and 59, computing the loads $l_{\req,x,y}$ of substrate resource $(x,y) \in  \SR$ for request $\req \in \requests$, still hold as the ``choose''-operations only depend on the other Constraints. Validating the correctness of the other constraints after having reduced the values of the variables in Line~27 is quite easy. The function \emph{ProcessPath} includes all variables along the edges of $P_{\textnormal{ext}}$ in the set $\mathcal{V}$. Hence, flow preservation will hold inside each extended graph. Furthermore, starting with the initial flow induction variable $f^+_{\req,\VVroot,\cdot}$ all intermediate flow induction variables are added to $\mathcal{V}$. Concretely, if a node is placelem:validity-of-kth-decomposition-scgepd into $\mathcal{Q}$, then its respective flow induction variable will have already been placed in $\mathcal{V}$ (cf. 45). Also, the embedding variable $x_{\req}$ is added, such that Constraints~50-57 are still valid.
all flow inducing variables $f^+_{\req,\cdot,\cdot}$ together with flows in the respective extended networks are reduced together with the variable $x_\req$. 

Hence, as after reducing the Integer Program's variables, all constraints are still valid and the later on found mappings are also valid by induction.
\end{proof}
\end{lemma}

\subsection{Approximating Service Cacti Embeddings}
\label{sec:cactus:approximation-service-cacti-embeddings}
Given Algorithm~\ref{alg:decompositionAlgorithmSCGEP} to decompose (relaxed)~solutions to Integer Programs~\ref{alg:SCGEP-IP} and \ref{alg:SCGEP-IP-C}, we can readily apply the approximation framework presented in Sections~\ref{sec:randround} and \ref{sec:approximation-without-admission-control} to obtain approximations for SCGEP-P and SCGEP-C. For completeness, we reiterate the important lemmas that link the Integer Programs and its decompositions to the respective approximation Algorithms~\ref{alg:approxAdmissionControl} and \ref{alg:approxWithoutAdmissionControl}.


The analogue of Lemma~\ref{lem:sum-of-f-mu} follows from Lemma~\ref{lem:validity-of-kth-decomposition-scgep}, as in each iteration the variable $x_{\req}$ is decremented by $\prob$ while adding a decomposition of value $\prob$:

\begin{lemma}
\label{lem:sum-of-f-mu-scgep}
The decomposition $\PotEmbeddings$ computed in Algorithm~\ref{alg:decompositionAlgorithmSCGEP} is complete, i.e. $\sum_{\decomp \in \PotEmbeddings} \prob = x_\req$ holds, for $\req \in \requests$.
\end{lemma}

We next prove the analogue of Lemma~\ref{lem:allocations}, stating that the cumulative load induced by the decompositions does not exceed the cumulative load in the respective Integer Programs~\ref{alg:SCGEP-IP} and \ref{alg:SCGEP-IP-C}.

\begin{lemma}
\label{lem:allocations-scgep}
The cumulative load induced by the fractional mappings obtained by Algorithm~\ref{alg:decompositionAlgorithmSCGEP} is less than the load computed in the respective integer program and hence less than the offered capacity, i.e. for all resources $(x,y) \in  \SR$ holds
\begin{align}
\sum_{\req \in \requests} \sum_{\decomp \in \PotEmbeddings} \prob \cdot \load(x,y)~\leq \sum_{\req \in \requests} l_{\req, x,y} \leq  \Scap(x,y)\,,
\label{eq:lem:allocations-scgep}
\end{align}
where $l_{\req,x,y}$ refers to the respective variables of the respective Integer Program.
\begin{proof}
It suffices to consider a single iteration of Algorithm~\ref{alg:decompositionAlgorithmSCGEP}. By Lemma~\ref{lem:validity-of-kth-decomposition-scgep} each virtual node $i \in \VV$ and each virtual edge $(i,j) \in  \VE$ is mapped to a substrate node and a substrate path respectively. By construction, each of these mappings was accounted for by the flow variables and the flow induction variables contained in $\mathcal{V}$. The load variables for resource $(x,y) \in  \SR$ are computed in Integer Programs~\ref{alg:SCGEP-IP} and \ref{alg:SCGEP-IP-C} via the flow values in Constraints~\ref{alg:SCGEP:nodeLoad} and \ref{alg:SCGEP:edgeLoad}. As the variables in $\mathcal{V}$ are uniformly decremented by the value $\prob$, and each node $i \in \VV$ or edge $(i,j) \in  \VE$ induces a load of $\Vcap(i)$ or $\Vcap(i,j)$ respectively and the computation of loads in Algorithm~\ref{alg:decompositionAlgorithmSCGEP} agrees with the computation in Constraints~\ref{alg:SCGEP:nodeLoad} and \ref{alg:SCGEP:edgeLoad}, the respective constraints pertaining to the loads, will be valid again after Line~28 of Algorithm~\ref{alg:decompositionAlgorithmSCGEP}. Hence, in each iteration the load accounted for in the decomposition $\decomp$ is upper bounded by the \emph{decrease} in the load variables in the respective iteration.
\end{proof}
\end{lemma}

Based on the above Lemmas, it is easy to obtain the analogues of Lemmas~\ref{lem:relation-of-net-profit-in-decomposition-and-the-LP-profit} and \ref{lem:relation-of-net-profit-in-decomposition-and-the-LP-cost}:

\begin{lemma}
\label{lem:relation-of-net-profit-in-decomposition-and-the-LP-profit-scgep}
Let $(\vec{x}, \vec{f^+}, \vec{f}, \vec{l}) \in  \FeasibleLP$ denote a \emph{feasible} solution to the linear relaxation of Integer Program~\ref{alg:SCGEP-IP} achieving a net profit of $\hat{B}$ and let $\PotEmbeddings$ denote the respective decompositions of this linear solution for requests $\req \in \requests$ computed by Algorithm~\ref{alg:decompositionAlgorithmSCGEP}, then the following holds:
\begin{align}
\sum_{\req \in \requests} \sum_{\decomp \in \PotEmbeddings} \prob \cdot \Vprofit = \hat{B}~.
\end{align}
\end{lemma}

\begin{lemma}
\label{lem:relation-of-net-profit-in-decomposition-and-the-LP-cost-scgep}
Let $(\vec{x}, \vec{f^+}, \vec{f}, \vec{l}) \in  \FeasibleLP$ denote a \emph{feasible} solution to the linear relaxation of Integer Program~\ref{alg:SCGEP-IP-C} having a cost of $\hat{C}$ and let $\PotEmbeddings$ denote the respective decompositions of this linear solution computed by Algorithm~\ref{alg:decompositionAlgorithmSCGEP} for requests $\req \in \requests$, then the following holds:
\begin{align}
\sum_{\req \in \requests} \sum_{\decomp \in \PotEmbeddings} \prob \cdot c(\mapping)~\leq \hat{C}~.
\end{align}
Additionally, equality holds, if the solution $(\vec{x}, \vec{f}, \vec{l}) \in  \FeasibleLP$, respectively the objective $\hat{C}$, is optimal.
\end{lemma}

The proofs and lemmas contained in Sections~\ref{sec:randround} and \ref{sec:approximation-without-admission-control} to obtain the approximations to SCEP-P and SCEP-C are purely based on Lemmas~\ref{lem:validity-of-kth-decomposition} - \ref{lem:allocations}. The above presented analogues (Lemmas~\ref{lem:validity-of-kth-decomposition-scgep} - \ref{lem:relation-of-net-profit-in-decomposition-and-the-LP-cost-scgep})~hence give all the prerequisites to employ the approximation framework developed in Sections~\ref{sec:randround} and \ref{sec:approximation-without-admission-control}. 

The only changes to the respective approximations (cf. Algorithms~\ref{alg:approxAdmissionControl} and \ref{alg:approxWithoutAdmissionControl})~are to employ the novel Integer Programs~\ref{alg:SCGEP-IP} and \ref{alg:SCGEP-IP-C} in conjunction with the novel decomposition Algorithm~\ref{alg:decompositionAlgorithmSCGEP}. We lastly note that the number of variables and constraints of Integer Programs~\ref{alg:SCGEP-IP} and \ref{alg:SCGEP-IP-C} is still polynomial in the respective graph sizes. Concretely, the number of variables and constraints is bounded by $\mathcal{O}(\sum_{\req \in \requests} |\VE| \cdot |\SV| \cdot |\SE|)$, as for each edge $(i,j) \in  \VE$, which lies on a cycle $C_k \in \Cycles$, exactly $|\SVTypesCycle | \leq |\SV|$ many parallel path constructions with $\mathcal{O}(|\SE|)$ many edges is used. Hence, the runtime to compute the linear relaxations of the respective integer programs is still polynomial and the runtime of the novel decomposition algorithm increases at most by a factor $|\SV|$.

Without further proofs, we state the following two theorems.

\begin{theorem}
\label{thm:result-for-admission-control-scgep}
We adapt Algorithm~\ref{alg:approxAdmissionControl} by replacing Integer Program~\ref{alg:SCEP-IP} by Integer Program~\ref{alg:SCGEP-IP} and Algorithm~\ref{alg:decompositionAlgorithm} by Algorithm~\ref{alg:decompositionAlgorithmSCGEP}.
Assuming that $|\SV| \geq 3$ holds, and that $\maxLoadV[\req][x][y] \leq \varepsilon \cdot \Scap(x,y)$ holds for all resources~$(x,y) \in  \SR$ with $0 < \varepsilon \leq 1$ and by setting~$\beta = \varepsilon \cdot \sqrt{2\cdot \log~(|\SV|\cdot|\types|)~\cdot \DeltaV  }$ and~$\gamma = \varepsilon \cdot \sqrt{2\cdot \log |\SV| \cdot \DeltaE  }$ with~$\DeltaV,\DeltaE$ as defined in Lemmas~\ref{lem:approximation-single-node} and \ref{lem:approximation-single-edge}, we obtain a~$(\alpha,1+\beta,1+\gamma)$ tri-criteria approximation algorithm for  SCGEP-P, such that it finds a solution \emph{with high probability}, that achieves at least an~$\alpha = 1/3$ fraction of the optimal profit and violates network function and edge capacities only within the factors~$1+\beta$ and~$1+\gamma$ respectively.
\end{theorem}

\begin{theorem}
\label{thm:result-for-admission-control-without-scgep}
We adapt Algorithm~\ref{alg:approxWithoutAdmissionControl} by replacing Integer Program~\ref{alg:SCEP-IP_without} by Integer Program~\ref{alg:SCGEP-IP-C} and Algorithm~\ref{alg:decompositionAlgorithm} by Algorithm~\ref{alg:decompositionAlgorithmSCGEP}.
Assuming that $|\SV| \geq 3$ holds and that $\maxLoadV[\req][x][y] \leq \varepsilon \cdot \Scap(\type,u)$ holds for all network resources~$(x,y) \in  \SR$ with $0 < \varepsilon \leq 1$ and by setting~$\beta = \varepsilon \cdot \sqrt{\log~(|\SV|\cdot|\types|)~\cdot \DeltaV  }$ and~$\gamma = \varepsilon \cdot \sqrt{3/2\cdot \log |\SV| \cdot \DeltaE  }$ with~$\DeltaV,\DeltaE$ as defined in Lemmas~\ref{lem:approximation-single-node} and \ref{lem:approximation-single-edge}, we obtain a~$(\alpha,2+\beta,2+\gamma)$ tri-criteria approximation algorithm for SCGEP-C, such that it finds a solution \emph{with high probability}, with costs less than~$\alpha = 2$ times higher than the optimal cost and violates network function and edge capacities only within the factors~$2+\beta$ and~$2+\gamma$ respectively.
\end{theorem}

\begin{figure*}[t!]

 {
  \LinesNotNumbered
  \renewcommand{\arraystretch}{1.5}
 
 \removelatexerror

  \begin{IPFormulationStar}{H}
 
  \SetAlgorithmName{Integer Program}{}{{}}

  \newcommand{\spaceIt}{\qquad\quad\quad}
  \newcommand{\miniSpace}{\hspace{1.5pt}}
 
  \centering
  \begin{tabular}{FRLQB}
   \multicolumn{4}{C}{\textnormal{max~}  \sum \limits_{\req \in \requests}  \Vprofit x_{\req}  ~~~~~~~~~~~~~~~~} & \tagIt{alg:VNEP-old:obj}\\
   \sum \limits_{u \in \SVTypes[\Vtype(i)]} y_{\req, i,u} & = & x_{\req} & \forall \req \in \requests, i \in \VV &  \tagIt{alg:VNEP-old:node-embedding} \\
   \sum \limits_{(u,v) \in  \delta^+(u)} z_{\req,i,j,u,v} - \sum \limits_{(v,u) \in  \delta^-(u)} z_{\req,i,j,v,u} & = & y_{\req, i,u} - y_{\req,j,u} \qquad \qquad &  \forall \req \in \requests, (i,j) \in  \VE &  \tagIt{alg:VNEP-old:edge-embedding}\\

 	\sum \limits_{\req \in \requests} \sum \limits_{i \in \VV, \Vtype(i)~= \type} \Vcap(i)~\cdot y_{\req,i,u}  & \leq  & \Scap(\type,u)~& (\type,u) \in  \SRV &  \tagIt{alg:VNEP-old:load-node}\\
 	
 	\sum \limits_{\req \in \requests} \sum \limits_{(i,j) \in  \VE } \Vcap(i,j)~\cdot z_{\req,i,j,u,v} & \leq &  \Scap(u,v)~& \forall \req \in \requests, (u,v) \in  \SE & \tagIt{alg:VNEP-old:load-edge}\\
 
   x_{\req} & \in & \{0,1\} &  \forall \req \in \requests & \tagIt{alg:foo1}\\
   y_{\req,i,u} & \in & \{0,1\} &  \forall \req \in \requests, i \in \VV, u \in \SVTypes[\Vtype(i)] &  \tagIt{alg:fo2}\\
   y_{\req,i,u} & = & 0 &  \forall \req \in \requests, i \in \VV, u \notin \SVTypes[\Vtype(i)] &  \tagIt{alg:foo3}\\
   z_{\req,i,j,u,v} & \in & \{0,1\} &  \forall \req \in \requests, (i,j) \in  \VE, (u,v) \in  \SE  &  \tagIt{alg:foo4}
  \end{tabular}
  \caption{Classic Multi-Commodity Flow Formulation for SCGEP-P}
  \label{alg:VNEP-IP-old}
  \end{IPFormulationStar}
  }
\end{figure*}

\subsection{Non-Decomposability of the Standard IP}
\label{sec:cactus:non-decomposability}

The a priori graph decompositions, the extended graph constructions, the respective integer program, and the corresponding decomposition algorithm to obtain approximations for SCGEP are very technical (cf. Sections~\ref{sec:cactus:decomposing-service-graphs} to~\ref{sec:cactus:decomposition-algorithm-for-service-cacti-and-correctness}). In the following we argue that standard multi-commodity flow integer programming formulations typically used in the virtual network embedding
literature, see e.g.~\cite{vnep,mehraghdam2014specifying,rostSchmidFeldmann2014}, cannot be employed for our randomized rounding approach. This result also suggests that the hope expressed by Chowdhury et al.~to obtain approximations using this standard formulation~\cite{vnep}, is unlikely to be true in general.

In particular, we consider the solutions of the relaxation of the archetypical
Integer Program~\ref{alg:VNEP-IP-old} for SCEP-C, which uses multi-commodity flows to connect virtual network functions in the substrate (cf.~\cite{vnep,mehraghdam2014specifying,rostSchmidFeldmann2014}), and show that this formulation allows for solutions which cannot be decomposed into valid embeddings. As we will show, this has ramifications beyond not being able to apply the randomized rounding approach, as the relaxations obtained from Integer Program~\ref{alg:VNEP-IP-old} are provably weaker than the ones obtained by Integer Program~\ref{alg:SCGEP-IP}.

Integer Program~\ref{alg:VNEP-IP-old} employs three classes of variables. The variable $x_{\req} \in \{0,1\}$ indicates whether request $\req \in \requests$ shall be embedded. Variables $y_{\req,i,u} \in \{0,1\}$ and $z_{\req, i,j, u,v}$ denote the embedding of virtual node $i\in \VV$ on substrate node $u \in \SVTypes[\Vtype(i)]$ and the embedding of virtual edge $(i,j) \in  \VE$ on substrate edge $(u,v) \in  \SE$ respectively.  By Constraint~\ref{alg:VNEP-old:node-embedding}, the virtual node $i \in \VV$ of request $\req \in \requests$ must be placed on any of the appropriate substrate nodes in $\SVTypes[\Vtype(i)]$ iff. $x_\req = 1$ holds. Constraint~\ref{alg:VNEP-old:edge-embedding} induces a unit flow for each virtual edge $(i,j) \in  \VE$: if virtual node $i \in \VV$ is mapped onto substrate node $u \in \SV$ and virtual node $j \in \VV$ is mapped onto $v \in \SV$, then the flow balance at node $u$ is $1$, and the flow balance at substrate node $v$ is $-1$ and flow preservation holds elsewhere. Note also that in the case that both the source and the target are mapped onto the same node no flow is induced. Lastly, Constraints~\ref{alg:VNEP-old:load-node} to \ref{alg:VNEP-old:load-edge} compute the effective node and edge allocations and bound these by the respective capacities.

By the above explanation, it is easy to check that any integral solution to Integer Program~\ref{alg:VNEP-IP-old} defines a valid mapping. Indeed, for each request $\req \in \requests$ with $x_\req = 1$ the node $i\in \VV$ is mapped onto substrate node $u$, i.e. $\mapV(i)~= u$ holds, iff. $y_{\req, i,u} = 1$ holds and the edge mapping of $(i,j) \in  \VE$ can be recovered from the flow variables $y_{\req,i,j,\cdot,\cdot}$ by performing a breadth-first search from $\mapV(i) \in  \SV $ to $\mapV(j) \in  \SV$ where an edge $(u,v)\in \SE$ is only considered if $y_{\req,i,j,u,v} = 1$ holds. We denote the set of feasible integral solutions of Integer Program~\ref{alg:VNEP-IP-old} by $\spaceIP$ and the set of linear solutions by $\spaceLP$.

Reversely, each solution to SCEP-P, i.e. each subset $R'\subseteq \requests$ of requests with mappings $(\mapV,\mapE)$ for $r \in R'$, induces a specific solution to the Integer Program~\ref{alg:VNEP-IP-old}. We denote by $\varphi$ the function that given a set $R'$ and corresponding mappings $\{\map | r \in R'\}$ yields the respective integer programming solution $(\vec{x}, \vec{y}, \vec{z}) \in  \spaceIP$ defined as follows:

\begin{itemize}
\item $x_{\req} = 1$ iff. $\req \in R'$ for all requests $\req \in \requests$,
\item $y_{\req, i, u} = 1$ iff. $\req \in R'$ and $\mapV(i)~= u $ for all requests $\req \in \requests$, virtual nodes $i \in \VV$, and substrate nodes $u \in \SVTypes[\Vtype(i)]$.
\item $z_{\req, i,j,u,v} = 1$ iff. $\req \in R'$ and $(u,v) \in  \mapE(i,j)$ for all requests $\req \in \requests$, virtual edges $(i,j) \in  \VE$ and substrate edges $(u,v) \in  \SE$.
\end{itemize}

By the above argumentation, we observe that Integer Program~\ref{alg:VNEP-IP-old} indeed is a valid formulation for SCEP-P. 

We will now investigate the linear relaxations of Integer Program~\ref{alg:VNEP-IP-old} and show that the formulation may produce solutions, which cannot be embedded. As an important precursor to this result, we first define the set of all linear programming solutions that may originate from convex combinations of singular embeddings as follows.

\begin{definition}[Decomposable LP Solution Spaces] ~\\
Let $\spaceSolReq$ denote the set of all \emph{valid} mappings for the request $\req \in \requests$. We denote by $\spaceLPDreq$ the set of all convex combinations of projections of valid mappings into the solutions space $\spaceLP$ of Integer Program~\ref{alg:VNEP-IP-old}, i.e. $\spaceLPDreq$ equals
\[
\{\sum^n_{k=1} \lambda_k \cdot \varphi(\{r\},M_k)~| n \in \mathbb{N}, M_k \in \spaceSolReq, \lambda_k \geq 0, \sum^n_{k=1} \lambda_k \leq 1 \}\,.
\]
Note that the above also allows for the possibility of the non-embedding, as $\sum_{k=1}^n \lambda_k = 0$ is allowed.
We denote the space of all decomposable linear programming solutions as $\spaceLPD = \{ (\vec{x},\vec{y}, \vec{z}) \in  \bigtimes \limits_{\req \in \requests} \spaceLPDreq | (\vec{x},\vec{y}, \vec{z})~\textnormal{satisfies Constraints}~\ref{alg:VNEP-old:node-embedding} - \ref{alg:VNEP-old:load-edge}\}$.
\end{definition}

As stated before, we denote by $\spaceLP = \{(\vec{x},\vec{y},\vec{z})| (\vec{x},\vec{y},\vec{z})~ \textnormal{satisfies Constraints}~\ref{alg:VNEP-old:node-embedding} - \ref{alg:VNEP-old:load-edge}\}$ denote the set of linear programming solutions feasible according to the constraints of Integer Program~\ref{alg:VNEP-IP-old}. Using the definitions of $\spaceLPD$ and $\spaceLP$ we can now formally prove, that the linear relaxation of Integer Program~\ref{alg:VNEP-IP-old} does contain solutions, which cannot be decomposed.

\begin{lemma}
The set of feasible LP solutions of Integer Program~\ref{alg:VNEP-IP-old} may contain non decomposable solutions, i.e. $\spaceLP \not \subseteq \spaceLPD$.
\label{lem:non-decomposability}
\end{lemma}
\begin{proof}
We show $\spaceLP \setminus \spaceLPD \neq \emptyset$ by constructing an example in Figure~\ref{fig:non-decomp} using a single request and a 8-node substrate. We assume that $\SVTypes[\Vtype(i)] =  \{u_1,u_5\}$, $\SVTypes[\Vtype(j)] =  \{u_2,u_6\}$, $\SVTypes[\Vtype(k)] =  \{u_4,u_8\}$, $\SVTypes[\Vtype(l)] =  \{u_3,u_7\}$ holds. The depicted request shall be fully embedded, i.e. $x_{\req} = 1$ holds. The fractional embeding is represented indicating the node mapping variables, $\frac{1}{2}i$ means that e.g. the corresponding mapping value $y_{\req, i,u_1}$ is $1/2$. Edge mappings are represented according to the dash style of the request and always carry a flow value of $1/2$. Clearly, the depicted fractional embedding is feasible and therefore contained in $\spaceLP$.
\begin{itemize}
\item Constraint~\ref{alg:VNEP-old:node-embedding} holds as each virtual node is mapped with a cumulative value of 1 to substrate nodes supporting the respective function.
\item Constraint~\ref{alg:VNEP-old:edge-embedding} holds as the substrate nodes onto which the tails of virtual edges (i.e. $\{i,j,k\}$)~have been mapped on, have corresponding outgoing flows while the heads of virtual edges have corresponding incoming flows.
\item Constraints~\ref{alg:VNEP-old:load-node} and \ref{alg:VNEP-old:load-edge} hold trivially if we assume large enough capacities on the substrate.
\end{itemize}

Assume for the sake of contradiction that the depicted embedding is a linear combination of elementary solutions, i.e. there exist mappings $M_k$ for $k \in K$ such that the depicted solution is a linear combination $\sum_{k\in K} \lambda_k \cdot \varphi(\{r\}, M_k)$ with $\lambda_k \geq 0$ and $\sum_{k \in K} \lambda_k \leq 1$. As virtual node $i$ is mapped onto substrate node $u_1$, and $u_2$ and $u_8$ are the only neighboring nodes that host $j$ and $k$ respectively there must exist a mapping $(\mapV, \mapE)$ within $M_k$ with $\mapV(i)=u_1$, $\mapV(j)=u_2$, $\mapV(k)=u_8$ and $\mapE(i,j)~= (u_1,u_2)$ and $\mapE(i,k)~= (u_1,u_8)$. Similarly, as the flow of virtual edge $(j,l)$ at $u_2$ only leads to $u_3$ and the flow of virtual edge $(j,k)$ at $u_8$ only leads to $u_6$, the virtual node $l$ must be embedded both on $u_3$ and $u_7$. As the virtual node $l$ must be mapped onto exactly one substrate node, this partial decomposition cannot possible be extended to a valid embedding. Hence, the depicted solution cannot be decomposed into elementary mappings and $\spaceLP \not \subseteq \spaceLPD$ holds (in general).
\end{proof}

\begin{figure}[tbhp]
\centering
\includegraphics[width=1\columnwidth]{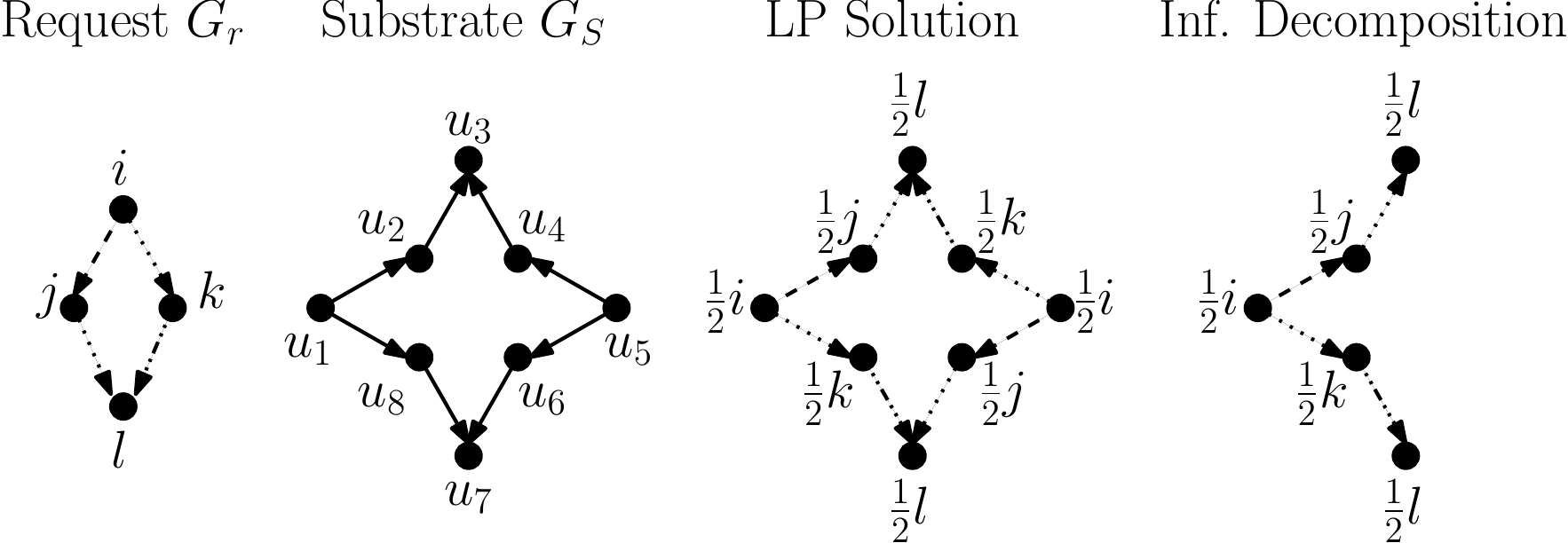}
\caption{Example showing that linear relaxations of Integer Program~\ref{alg:VNEP-IP-old} may not be decomposable. The LP solution with $x_\req=1$ is depicted as follows. Substrate nodes are annotated with the mapping of virtual nodes and $\frac{1}{2}i$ for substrate node $u_1$ stands for $y_{r,i,u_1} = 1/2$. Substrate edges are dashed accordingly to the dash style of the virtual links mapped onto it. All virtual links are mapped with values $1/2$. The dash style of substrate edge $(u_1,u_2)$ therefore implies that $z_{r,i,j,u_1,u_2} = 1/2$ holds.}
\label{fig:non-decomp}
\end{figure}

We lastly prove the following theorem, showing that the relaxations of the novel Integer Program~\ref{alg:SCGEP-IP} are provably stronger than the ones of the standard IP~\ref{alg:VNEP-IP-old}. To compare the \emph{strength} of formulations over different variable spaces, we project relaxed solutions of Integer Program~\ref{alg:SCGEP-IP} to the solution space $\spaceLP$ of Integer Program~\ref{alg:SCEP-IP} (see~\cite{balas2005projection} for an introduction on comparing formulations using projections).
\begin{theorem}
\label{thm:stronger-formulation}
The linear relaxations of Integer Program~\ref{alg:SCGEP-IP} are provably stronger than the linear relaxation of Integer Program~\ref{alg:VNEP-IP-old}.
\begin{proof}
We denote by $\spaceLPNew$ the solution space of the linear relaxations of Integer Program~\ref{alg:SCGEP-IP}. Let $\pi : \spaceLPNew \to \spaceLP$ denote the projection of the novel solution space to the solution space of the standard formulation. Concretely, given a solution $(\vec{x},\vec{f^+},\vec{f},\vec{l}) \in  \spaceLPNew$ the function $\pi$ first computes the according decomposition $\PotEmbeddings$ according to Algorithm~\ref{alg:decompositionAlgorithmSCGEP} and then returns the following solution contained in $\spaceLP$:
\begin{itemize}
\item $x_\req = \sum \limits_{\decomp \in \PotEmbeddings} \prob$ for all $\req \in \requests$. \\[2pt]
\item $y_{\req,i,u} = \sum \limits_{\decomp \in \PotEmbeddings: \mapping(i)=u} \prob$ for all requests $\req \in \requests$, virtual nodes $i\in \VV$ and substrate nodes $u \in \SVTypes[\Vtype(i)]$. \\[2pt]
\item $z_{\req,i,j,u,v} = \sum \limits_{\decomp \in \PotEmbeddings: (u,v) \in  \mapping(i,j)} \prob$ for all requests $\req \in \requests$, virtual edges $(i,j)\in \VE$ and substrate edges $(u,v) \in  \SE$.
\end{itemize}
As the Formulation~\ref{alg:SCGEP-IP} and the decomposition Algorithm~\ref{alg:decompositionAlgorithmSCGEP} are valid, it is easy to check that the above projection is valid, i.e. $\pi(\spaceLPNew)~\subseteq \spaceLP$ holds.
Furthermore, given any solution $(\vec{x},\vec{f^+},\vec{f},\vec{l}) \in  \spaceLPNew$, it is obvious that the projected solution $\pi(\vec{x},\vec{f^+},\vec{f},\vec{l})$ is decomposable, as the respective mappings were computed explicitly in the decomposition algorithm. Thus, $\pi(\spaceLPNew)~\subseteq \spaceLPD$ holds and under this projection all solutions are decomposable. Together with Lemma~\ref{lem:non-decomposability}, we obtain $\pi(\spaceLPNew)~\subseteq \spaceLPD \subsetneq \spaceLP$.
Lastly, as $\spaceLP \setminus \spaceLPD \neq \emptyset$ holds, the relaxations of Integer Program~\ref{alg:SCGEP-IP} are indeed provably stronger than the ones of Integer Program~\ref{alg:VNEP-IP-old}.
\end{proof}
\end{theorem}

The above theorem is based on the structural deficits of the Integer Program~\ref{alg:VNEP-IP-old} and does neither depend on the particular way we have formulated IP~\ref{alg:VNEP-IP-old} nor does it depend on whether we consider SCGEP-P or SCGEP-C. Furthermore, we note that the above theorem can also have practical implications, when trying to obtain obtain \emph{good bounds} via linear relaxations, as the difference in the benefit (or cost)~can be unbounded.

\begin{lemma}
The objectives of the linear relaxations of Integer Programs~\ref{alg:SCGEP-IP} and \ref{alg:VNEP-IP-old} can diverge arbitrarily. Concretely, considering the embedding of a particular set of requests $\requests$ on substrate $\SG$ and denoting the benefit obtained by Integer Program~\ref{alg:SCGEP-IP} as $\optLPnew$ and the benefit obtained by Integer Program~\ref{alg:VNEP-IP-old} as $\optLPstd$, the absolute difference $\optLPstd - \optLPnew$ as well as the relative difference $(\optLPstd - \optLPnew)~/ \optLPnew$ may be unbounded. 
\begin{proof}
We first note that by Theorem~\ref{thm:stronger-formulation} we have $\optLPnew \leq \optLPstd$ as $\pi(\spaceLPNew)~\subsetneq \spaceLP$ holds. We reuse the example depicted in Figure~\ref{fig:non-decomp} and we assume that the depicted request is the only one. As discussed, the depicted fractional solution is a feasible relaxation of Integer Program~\ref{alg:VNEP-IP-old}. As $x_{\req} = 1$ holds, we have $\optLPstd = b_{\req}$.  

Considering the relaxations of Integer Program~\ref{alg:SCGEP-IP}, we claim that only $x_{\req} = 0$ is feasible. This is easy to check as there does not exist a potential substrate location for virtual node $l$, such that the branches $i \rightarrow j \rightarrow l$ and $i \rightarrow k \rightarrow l$ can end in the same substrate location while also emerging at the same location. Hence, the benefit obtained by the relaxation of Integer Program~\ref{alg:SCGEP-IP} is $\optLPnew = 0$.

Hence, the absolute difference is $\optLPstd - \optLPnew = b_{\req}$ and -- as $b_{\req}$ can be set arbitrarily -- the absolute as well as the relative difference are unbounded.
\end{proof}
\end{lemma}

\section{Related Work}\label{sec:relwork}

Service chaining
has recently received much attention by
both researchers and practitioners~\cite{karl-chains,ewsdn14,stefano-sigc,merlin}. 
%
Soul\'{e} et al.~\cite{merlin} present Merlin,
a flexible framework which allows
to define and embed service chain policies.
The authors also present an integer program
to embed service chains, however, solving
this program requires an exponential runtime. 
Also
Hartert et al.~\cite{stefano-sigc} have studied
the service chain embedding problem,
and proposed a constraint optimization
approach. However also this approach
requires exponential runtime in the worst case.
Besides the optimal but exponential
solutions presented in~\cite{stefano-sigc,merlin}, 
there also exist heuristic solutions, e.g.,~\cite{karl-chains,ewsdn14}:
while heuristic approaches may be attractive for their
low runtime, they do not provide any
worst-case quality guarantees. 
We are the first to present
polynomial time algorithms for service chain embeddings
which comes with formal approximation guarantees.
Concretely, our results based on randomized rounding techniques
are structured into two papers:
In~\cite{swat-guy}, we only consider the admission control variant
and derive a constant approximation under assumptions on the relationship
between demands and capacities as well as on the optimal benefit.
In contrast, we have considered in this paper a more general
setting, in which the decomposition approach based on random walks~\cite{swat-guy}
 is not feasible, as we allow for arbitrary service cactus graphs.
Importantly, while our approach does not require specific assumptions on loads and benefits, we obtain
worse approximation ratios. 

From an algorithmic perspective, 
the service chain embedding problem
can be seen as a variant of  a
graph embedding problem, see~\cite{diaz2002survey} for a survey.
Minimal Linear Arrangement~(MLA)~is the archetypical
embedding footprint minimization problem,
where the substrate network has a most simple topology:
a linear chain.
It is known that MLA can be~$O(\sqrt{\log{n}}\log\log{n})$ approximated in
polynomial time~\cite{best-mla-1,mla-best-2}.
However, the approximation algorithms 
for VLSI layout problems~\cite{diaz2002survey,vlsi-layout},
cannot be adopted for embedding service chains in 
general substrate graphs.

Very general graph embedding problems have recently also 
been studied by the networking community
in the context of virtual network embeddings~(sometimes
also known as testbed mapping problems).
Due to the inherent computational hardness of the underlying
problems, the problem has mainly been approached using mixed integer
programming~\cite{vnep} and heuristics~\cite{vnep-rethink}.
For a survey on the more practical literature, we refer
the reader to~\cite{vnep-survey}.
Algorithmically interesting and computationally tractable embedding problems
arise in more specific contexts, e.g., 
in fat-tree like datacenter networks~\cite{oktopus,ccr15emb}.
In their interesting work, 
Bansal et al.~\cite{bansal2011minimum} 
give an
$n^{O(d)}$ time~$O(d^2 \log{(nd)})$-approximation algorithm 
for minimizing the load of embeddings in tree-like datacenter
networks, based on a strong LP relaxation
inspired by the Sherali-Adams hierarchy.
The problem of embedding star graphs 
has recently also been explored on various substrate topologies~(see also~\cite{ccr15emb}),
but to the best of our knowledge, no approximation algorithm
is known for embedding chain- and cactus-like virtual networks on arbitrary topologies.

Finally, our work is closely related to unsplittable
path problems, for which various approximation algorithms
exist,
also for admission control variants~\cite{Kleinberg-admission-control}.
In particular, our work leverages  
techniques introduced by Raghavan and Thompson~\cite{Raghavan-Thompson}:
in their seminal work, the authors 
develop provably
good algorithms based on relaxed, polynomial time versions of 
0-1 integer programs. 
However, our more general setting not only requires a novel
and more complex decomposition approach,   but also a novel and advanced
formulation of the mixed integer program itself: as we have shown, 
standard formulations cannot be decomposed at all.
Moreover, we are not aware of any extensions of the randomized rounding
approach to problems allowing for admission control and the objective of maximizing the profit.

\section{Conclusion}\label{sec:conclusion}

This paper initiated the study of 
polynomial time approximation algorithms
for the service chains embedding problem, and beyond. 
In particular, we have presented novel approximation algorithms,
which apply randomized rounding on the decomposition of linear programming solutions, and which also support admission control. 
We have shown that using this approach, a constant approximation of the objective
is possible in multi-criteria models with non-trivial augmentations, both for
service chains as well as for more complex virtual networks, particularly
service cactus graphs.
Besides our result and decomposition technique,
we believe that also our new integer program formulation may be of independent interest.

Our paper opens several interesting directions
for future research. In particular, we believe that our
algorithmic approach is of interest beyond the service chain 
and service cactus graph embedding problems considered in this paper, and can be used 
more generally. In particular, we have shown that our approach can be employed as long as the relaxations can be decomposed. Hence, we strongly believe that our approach can be applied to further graph classes as well.
Moreover, it will also be interesting 
to study the tightness of the novel bounds obtained in this work.

\textbf{Acknowledgement.}
This research was supported by the EU project UNIFY FP7-IP-619609.

{

\bibliographystyle{plain}
\bibliography{references}
}

\end{document}